\documentclass[a4paper,onecolumn,11pt,unpublished]{quantumarticle}
\pdfoutput=1
\usepackage{xcolor}
\usepackage{mathtools}
\usepackage[utf8]{inputenc}
\usepackage[english]{babel}
\usepackage[T1]{fontenc}
\usepackage{amsmath, amsthm}
\usepackage{quantikz}
\usepackage{hyperref}
\usepackage{tabularx}
\usepackage{threeparttable}
\usepackage{tikz}
\usetikzlibrary{decorations.pathmorphing}
\usepackage{lipsum}
\usepackage[backend=bibtex]{biblatex}
\newtheorem{theorem}{Theorem}
\usepackage{thm-restate}
\newtheorem{proposition}[theorem]{Proposition}
\newtheorem{lemma}[theorem]{Lemma}
\newtheorem{problem}[theorem]{Problem}
\newtheorem{definition}[theorem]{Definition}
\usepackage{amsfonts}  
\usepackage{braket}
\usepackage{tikz}
\usepackage{bm}
\usepackage{algorithm}  
\usepackage{algpseudocode}  
\usepackage{algorithmicx}   
\usepackage{float}    
\usepackage[capitalize,nameinlink]{cleveref} 
\usetikzlibrary{positioning, arrows.meta, decorations.pathmorphing}

\begin{document}

\title{Simulating Time Dependent and Nonlinear Classical Oscillators through Nonlinear Schr\"{o}dingerization}

\author{Abhinav Muraleedharan}
\affiliation{Department of Computer Science, University of Toronto, Toronto ON, Canada}
\orcid{0000-0002-2445-2701}
\email{abhi@cs.toronto.edu}
\author{Nathan Wiebe}
\affiliation{Department of Computer Science, University of Toronto, Toronto ON, Canada}
\affiliation{Pacific Northwest National Laboratory, Richland WA, USA}
\affiliation{Canadian Institute for Advanced Research, Toronto ON, Canada}
\email{nawiebe@cs.toronto.edu}

\maketitle
\begin{abstract}
We present quantum algorithms for simulating the dynamics of a broad class of classical oscillator systems containing $2^n$ coupled oscillators (Eg: $2^n$ masses coupled by springs), including those with time-dependent forces, time-varying stiffness matrices, and weak nonlinear interactions. This generalization of the Harmonic oscillator simulation algorithm is achieved through an approach that we call ``Nonlinear Schr\"{o}dingerization'', which involves reduction of the dynamical system to a nonlinear Schr\"{o}dinger equation  and then reduced to a time-independent Schrodinger Equation through perturbative techniques.  The linearization of the equation is performed using an approach that allows the dynamics of a nonlinear Schr\"odinger equation to be approximated as a linear Schr\"odinger equation in a higher dimensional space.  This allows Hamiltonian Simulation algorithms to be applied to simulate the dynamics of resulting system.
When the properties of the classical dynamical systems can be efficiently queried, and when the initial state can be efficiently prepared, the complexity of our quantum algorithm is polynomial in $n$, and almost linear in evolution time for most dynamical systems. Our work extends the applicability of quantum algorithms to simulate the dynamics of non-conservative and nonlinear classical systems, addressing key limitations in previous approaches.  

\end{abstract}
\clearpage 
\tableofcontents
\newpage
\section{Introduction}



The simulation of complex physical systems has been a longstanding challenge at the intersection of physics and computer science. As our understanding of natural phenomena grows more sophisticated, so too does the complexity of the models we use to describe them. This increasing complexity often leads to computational demands that strain or exceed the capabilities of classical computers, particularly when dealing with large-scale systems, those with intricate time-dependent behaviors, and systems that are quantum mechanical in nature. The primary motivation for introducing quantum computers was to address these large-scale simulation problems, particularly those that are quantum mechanical in nature \cite{feynman1985quantum, lloyd1996universal,berry2007efficient,low2019hamiltonian,childs2012hamiltonian}.\\

Recent advancements in quantum algorithms have expanded the potential applications of quantum computers beyond purely quantum systems. Of particular interest is the simulation of large classical physical systems \cite{li2025potential,jennings2024cost, sanavio2025explicit}, which, while not inherently quantum in nature, can benefit from the computational advantages offered by quantum hardware. A significant breakthrough in this direction came with the work of \cite{babbush2023exponential}, which introduced a novel method to map the dynamics of classical oscillator systems onto the Schr\"{o}dinger evolution of an appropriate quantum system. Their approach provided optimal scaling with respect to system properties, demonstrating the potential for quantum computers to offer significant advantages in simulating certain classes of classical systems. Further works  \cite{stroeks2024solving, somma2025shadow} have applied similar ideas to simulate dynamics of observables of quantum systems, such as fermions and bosons.
\\

However, their method relied on two critical assumptions: conservation of energy and linearity of dynamics. While the conservation of energy assumption is valid for many important systems, it precludes the application of their method to a wide range of scenarios where the total energy varies with time, such as systems subject to external time-dependent forces or those with time-varying stiffness matrices. Although in principle, one may apply quantum linear system solver \cite{harrow2009quantum} based methods \cite{berry2014high,clader2013preconditioned,  childs2017quantum,subacsi2019quantum,Berry2017Quantum,berry2024quantum, krovi2023improved} to simulate such systems, the complexity of such methods is far from optimal. Specifically, linear systems-based methods require a large number of queries to the state preparation oracle and often incur significant overhead due to the condition number dependence and the necessity of implementing complex block encodings of the system matrices. Recent techniques such as Schr\"{o}dingerization \cite{jin2024quantum,jin2024schr} and LCHS \cite{an2023quantum, an2023linear} reformulate linear ODE problems as linear combinations of Hamiltonian simulations, achieving optimal query complexity to the state preparation oracle in the general case, but they require additional stability conditions for coefficient matrix which need not be satisfied in general.
\\

For simulating nonlinear dynamical systems, efficient quantum algorithms \cite{liu2021efficient,costa2023further} have been restricted to the simulation of dissipative nonlinear dynamical systems, and have yet to be extended to systems without strong dissipation. The simulation of nonlinear differential equations on quantum computers leverages Carleman linearization to embed a finite‑dimensional nonlinear system into an infinite‑dimensional linear one that can be truncated and solved via quantum linear system algorithms. Liu \emph{et al.}\cite{liu2021efficient} introduced an efficient quantum algorithm for dissipative quadratic ODEs, achieving complexity 
\(\displaystyle O\bigl(T^2 q \,\mathrm{poly}(\log T,\log n,\log(1/\epsilon)) / \epsilon\bigr)\) under the condition \(R<1\), representing an exponential improvement over earlier methods. 
\\

However, these existing methods remain limited in scope, particularly for systems without strong dissipation with time-dependent behavior. Our work aims to address these limitations and significantly expand the applicability of quantum simulations to a broader range of classical systems. We introduce new quantum algorithms capable of simulating classical dynamical systems without the assumption of constant total energy.  These algorithms work on a principle that we call ``nonlinear Schr\"odingerization'' which involves reducing the timedependent problem to a nonlinear Schr\"odinger equation.  This is in a sense the opposite of the Schr\"{o}dingerization strategy proposed in~\cite{jin2024quantum} as it takes a linear differential equation and approximates it by a nonlinear one. Then we can implement the non-linear differential equation by re-linearizing it in a higher dimensional space through a novel approach that allows us to approximate the nonlinear dynamics with Hermitian dynamics.  This in effect allow us to approximate the time-dependent oscillator by a Harmonic oscillator in a higher-dimensional space.  This advancement enables more comprehensive and accurate modeling of complex dynamical phenomena\cite{landau1976mechanics, jackson1998classical, neven1992rate, wilson1980molecular}, including those with time-dependent parameters and external influences. Specifically, this paper presents methods for:
\begin{enumerate}
    \item Solving inhomogeneous differential equations representing oscillator systems subject to time-dependent external forces. This allows for the simulation of systems under varying external influences, such as forced oscillations or systems driven by time-varying fields.
    \item  Simulating systems with time-dependent stiffness matrices. This capability is crucial for modeling systems where the interactions between components change over time, such as in materials undergoing phase transitions or structures subject to varying environmental conditions.
    \item   Handling nonlinear oscillator systems with quadratic nonlinearities. Many real-world systems exhibit nonlinear behavior, and our ability to simulate these systems accurately opens up new avenues for studying complex phenomena.
    \item  Combining time-dependent stiffness matrices with time-dependent external forces. This represents our most general case, allowing for the simulation of systems with both internal time-varying properties and external time-dependent influences.
\end{enumerate}
 We apply perturbation theory to embed the dynamics of non-conservative systems within larger conservative systems, allowing us to leverage existing quantum simulation techniques. Carlemann linearization is used to transform nonlinear systems into linear systems of higher dimensions, making them amenable to quantum simulation. We also introduce novel symmetrization methods to ensure that the resulting operators can be efficiently implemented on quantum hardware. Our work opens up new possibilities for using quantum computers to simulate a wide array of classical physical systems. By removing the constraint of energy conservation and allowing for time-dependent parameters, we significantly expand the range of systems that can be efficiently simulated using quantum algorithms. This has the potential to enable the study of complex dynamical phenomena that are computationally intractable on classical computers.
\\








The rest of the paper is organized as follows. In \autoref{sec: prelim}, we discuss the notations used in the paper, and definitions of dynamics and access models for different oscillator systems. In \autoref{sec:results}, we give an overview of our main results, and in subsequent sections we discuss each of the simulation problems in detail. In \autoref{sec:linear inhomo-diff-eq}, we discuss the forced oscillator simulation problem, wherein external time dependent forces act on individual masses along with interaction forces from spring elements. In \autoref{sec:nonlinear_Schrodinger_eq}, we study the nonlinear Schr\"{o}dinger equation, and introduce methods to simulate the dynamics of quantum systems obeying nonlinear Schr\"{o}dinger equation. In \autoref{sec:nonlinear_oscillator_system}, we show that the dynamics of nonlinear oscillator system can be mapped to the nonlinear Schrodinger Equation, and hence its dynamics can be simulated by applying this transformation and methods from \autoref{sec:nonlinear_Schrodinger_eq}. In \autoref{sec:time_dependent_stiffness} and \autoref{sec:time_dependent_force_and_stiffness}, we show how the dynamics of oscillator systems with time dependent stiffness and force functions can be embedded in a higher dimensional nonlinear dynamical system and simulated using methods introduced in \autoref{sec:nonlinear_oscillator_system}. A conceptual diagram representing the various reductions developed in the paper is shown in \autoref{fig: conceptual_diagram}.

\section{Preliminaries}
\label{sec: prelim}
The aim of this section is to introduce common notation and review key ideas that will be used throughout the document.  In particular, the main results given in~\autoref{sec:results} are expressed in terms of these notations and readers familiar with the Harmonic Oscillator simulation algorithm or Shadow Hamiltonian simulation may wish to jump ahead to the discussion therein.

In our manuscript, we follow specific conventions for notation. Lowercase symbols represent scalars, such as \( a \), \( b \), \( c \), and \( d \). Lowercase bold symbols indicate vectors, like \( \mathbf{v} \) and \( \mathbf{w} \). Uppercase bold symbols denote matrices, for example, \( \mathbf{A} \) and \( \mathbf{B} \). Additionally, \( \| \mathbf{A} \| \) represents the spectral norm of the operator \( \mathbf{A} \) and $\| \mathbf{A}\|_F$ denotes the Frobenius norm of Operator $\mathbf{A}$.

We will find it useful in the following to define the ``logarithmic norm'' of a matrix which is given by
\begin{definition}[Logarithmic norm]
    The logarithmic norm of a matrix \( A \in \mathbb{C}^{n \times n} \), with respect to an induced norm $\|\cdot\|$, is defined as  the limit:

\[
{\mu(A) = \lim_{h \to 0^+} \frac{\| I_{n \times n} + hA \| - 1}{h}}
\]
For the case where $\|\cdot\|$ is the induced infinity-norm the logarithmic norm takes the form

\[
\mu_\infty(A) = \max_i \left( \textrm{Re}(a_{ii}) + \sum_{j \neq i} |a_{ij}| \right).
\]
\end{definition}
{The principal utility of the logarithmic norm for our purposes is derived from the fact that for any square matrix $A$ and $t\ge 0$
\begin{equation}
    \|e^{A t}\| \le e^{\mu(A) t}.
\end{equation}
This notion of a logarithmic norm can give a tighter bound than the standard bound usage where the triangle inequality and the sub-multiplicative property give us $\|e^{At}\|\le e^{\|A\|t}$, which may be substantially looser than this result since $\mu(A)\le \|A\|$.
}

\begin{definition}[Coupled Classical Oscillator System]
\label{Def: coupled_oscillator_system}
A coupled classical oscillator system consists of $N$ masses $\{m_j\}_{j=1}^N$, interconnected by spring elements. We define the mass matrix $\mathbf{M}$ of the system to be a diagonal matrix with terms $\mathbf{M}(j,j) = m_j$. Let $k_{ji}$ denote the spring constant of the spring interconnecting masses $m_j$ and $m_i$, and let $k_{jj}$ denote the spring constant connecting mass $m_j$ to the wall. Let $\mathbf{G} \in \mathbb{R}^{N \times N}$ be the stiffness matrix of the coupled oscillator system. The elements of matrix $G$ are defined as:
\begin{align*}
    \mathbf{G}(i,j) &= k_{ij} = k_{ji} \\
    \mathbf{G}(j,j) &= k_{jj}
\end{align*}
Let $\mathbf{x} = [x_1,x_2,..x_N]^T$ be the displacement vector of the oscillator system. The governing differential equation of the coupled classical oscillator system is given by:
\begin{equation}
    \mathbf{M}\ddot{\mathbf{x}} = -\mathbf{K} \mathbf{x}
\end{equation}
Here, $\mathbf{K}$ is the incidence matrix of the weighted graph defined by $\mathbf{G}$. Specifically,  the diagonal terms of the matrix $\mathbf{K}$ is given by: $\mathbf{K}(j,j) = \sum_i k_{ji}$ and the off-diagonal terms are $\mathbf{K}(j,i) = - k_{ji}$. Note that $k_{ji} > 0$ and the matrix $\mathbf{K}$ is positive semidefinite. 
   
\end{definition}

\begin{definition}[Forced Coupled Classical Oscillator System]
\label{Def: forced_coupled_oscillator_system}
    In a forced coupled classical oscillator system, in addition to the forces exerted by the spring elements, external time-dependent forces $f_i(t)$ act on mass $m_i$. The stiffness matrix of the system, $\mathbf{G}$, and the matrix $\mathbf{K}$ are constructed as per \autoref{Def: coupled_oscillator_system}.  The governing differential equation of the forced coupled classical oscillator system is given by:
    \begin{equation}
        \mathbf{M} \ddot{\mathbf{x}} = -\mathbf{K} \mathbf{x} + \mathbf{f}(t)
    \end{equation}
    Here, $\mathbf{f}(t) = [f_1(t) , f_2(t),,,,f_N(t)]^T$ and we assume that the time-dependent forces can be written as $f_i(t) = \sum_{j=1}^l f_{ij} \cos(\omega_{ij}t + \phi_{ij})$. 
\end{definition}

\begin{definition}
\label{Def:Oracles-forced-oscillator-system}
    (State Prep and Access Model for forced oscillator system)
    Let $m_{max} \geq m_j$ and let $m_j = m_{max} [.b_{j,1} b_{j,2}...]$, $b_{j,i} \in \{0,1\}$. Here, $[.b_{j,1} b_{j,2}...] $ is a binary fraction, and we use the bit representation to define the basis state $\ket{\bar{m}_j}$. Similarly, we represent the stiffness parameter $k_{ij}$, the frequencies $\omega_{ij}$, the phase values $\phi_{ij}$, and the force amplitudes $f_{ij}$ as basis states $\ket{\bar{k}_{ij}}$, $\ket{\bar{\omega}_{ij}}$, $\ket{\bar{\phi}_{ij}}$ and $\ket{\bar{f}_{ij}}$ respectively. Using this representation, we define oracles for accessing elements of the matrices $\mathbf{M}$, $\mathbf{G}$, and $\mathbf{f}(t)$. Specifically, for the matrix $\mathbf{M}$, we define the oracle $O_M: \ket{j}\ket{0} \rightarrow \ket{j}\ket{\bar{m}_{j}} $ to access bitwise representation of $m_j$. The stiffness matrix elements and frequency values are accessed through the following oracles:

\begin{align*}
O_G&: \begin{cases}
    \ket{j}\ket{k}\ket{0}\ket{0} &\longrightarrow \ket{j}\ket{k}\ket{\bar{G}_{jk}}\ket{0} \\
    \ket{j}\ket{k}\ket{0}\ket{i} &\longrightarrow \ket{j}\ket{k}\ket{0}\ket{i} , \forall i \neq 0
\end{cases}\\
O_{F_w}&: \begin{cases}
    \ket{j}\ket{k}\ket{0}\ket{1} &\longrightarrow \ket{j}\ket{k}\ket{\bar{\omega}_{jk}}\ket{1} \\
    \ket{j}\ket{k}\ket{0}\ket{i} &\longrightarrow \ket{j}\ket{k}\ket{0}\ket{i}, \forall i \neq 1
\end{cases}\\
\end{align*}
The phase values and the amplitudes of the forcing function are accessed through the following oracles:
\begin{align}
O_{F_{\phi}}&: \ket{i}\ket{j}\ket{0} \rightarrow \ket{i}\ket{j}\ket{\bar{\phi}_{ij}} \\ \nonumber
O_{F_{f}}&: \ket{i}\ket{j}\ket{0} \rightarrow \ket{i}\ket{j}\ket{\bar{f}_{ij}} \\ \nonumber
\end{align}
Let's also assume we have access to a unitary $\mathcal{W} = \sum_i \frac{1}{2} \ket{0}\bra{0}\otimes\mathcal{W}_i$ that can efficiently prepare the initial state, proportional to the initial displacement and velocity vector of the oscillator system. The unitaries $\mathcal{W}_i$ are defined as: 
\begin{align}
    \mathbf{\mathcal{W}}_1 &: \ket{0}\bm{\ket{0}}^{\otimes n +r} \rightarrow \ket{0} \otimes \ket{0}^{\otimes r}\otimes\frac{ \dot{\mathbf{x}}(0)}{\| \dot{\mathbf{x}}(0)\|} \\ \nonumber
    \mathbf{\mathcal{W}}_2 &: \ket{1}\bm{\ket{0}}^{\otimes n+r} \rightarrow \ket{1}\otimes\frac{\dot{\mathbf{y}}(0)}{\|\dot{\mathbf{y}}(0)\|} \\ \nonumber
    \mathbf{\mathcal{W}}_3 &: \ket{2}\bm{\ket{0}}^{\otimes n+r}  \rightarrow \ket{2}\otimes\ket{0}^{\otimes r} \otimes\frac{\mathbf{x}(0)}{\| \mathbf{x}(0)\|} \\ \nonumber
    \mathbf{\mathcal{W}}_4 &: \ket{3}\bm{\ket{0}}^{\otimes n+r}  \rightarrow \ket{3}\otimes\frac{\mathbf{y}(0)}{\| \mathbf{y}(0)\|} \\ \nonumber
\end{align}.
Here, $\mathbf{x}(0)$, $\dot{\mathbf{x}}(0)$, are displacement and velocity vectors of the oscillator system, and $\mathbf{y}(0)$, $\dot{\mathbf{y}}(0)$ are displacement and velocity vectors of an auxillary dynamical system. The entries of these displacement and velocity vectors are specified by the parameters of the time dependent force. Specifically, the displacement vector $\mathbf{y}(0) = [{\mathbf{y}}_1(0),\cdots \mathbf{y}_N(0)]^T$, where $\mathbf{y}_i(0) = [y_{i,1}(0), \cdots y_{i,l}(0)]^T$,
 and $ y_{i,j}(0) =  \frac{2 l f_{ij} \cos(\phi_{ij})}{k_{ii}} $. The initial velocity vector, of the auxillary oscillator system is given by $\dot{\mathbf{y}}_i(0) = [\dot{y}_{i,1}(0), \cdots \dot{y}_{i,l}(0)]^T$ , and $ \dot{y}_{i,j}(0) =  -{ \frac{2l f_{ij}}{k_{ii}}} \omega_{ij}\sin(\phi_{ij})$. 

\end{definition}

Our approach to tackling time-dependent spring constants in the evolution stems from our perturbative analysis of nonlinear Schr\"{o}dinger equations. Before discussing our subsequent work on this topic, we will therefore need to discuss the case of the nonlinear Schr\"{o}dinger equation.  The nonlinear Schr\"{o}dinger equation has long been studied in quantum computing~\cite{meyer2013nonlinear,liu2021efficient,lloyd2020quantum}.  In essence, the nonlinear Schr\"{o}dinger equation is one where the quantum state evolves in a way that depends more than linearly on the underlying probability amplitudes.  Such nonlinearities are, in full generality, beyond the reach of quantum algorithms to simulate efficiently as such simulations would lead to a violation of the Holevo-Helstrom bound~\cite{liu2021efficient,brustle2024quantum}.  Nonetheless, in some cases nonlinear dynamics can be simulated and here we will study it for its aforementioned use in solving the time-dependent spring-constant problem as well as in its own right.  We give such a definition below for the case of a quadratic nonlinearity.  Generic higher-order polynomial nonlinearities can be reduced to the case of quadratic nonlinearities by substitution.

\begin{definition}[Nonlinear Schr\"{o}dinger Equation]
\label{Def: nonlinear_Schr\"{o}dinger_eq}
    The Nonlinear Schr\"{o}dinger Equation is given by:
    \begin{equation}
        \dot{\ket{\bm{\psi}} } = -i\mathbf{H}_1 \ket{\bm{\psi}} + \mathbf{H}_2 \ket{\bm{\psi}} \otimes \ket{\bm{\psi}}
    \end{equation}
Here, $\mathbf{H}_1 \in \mathbb{C}^{N \times N}$ is a Hermitian square matrix, while $\mathbf{H}_2 \in \mathbb{C}^{N \times N^2}$ is a rectangular matrix.  Furthermore, we assume without loss of generality that the matrix $\mathbf{H}_1$ is $d_1$ sparse, that is it contains at most $d_1$ nonzero entries in any row or column. We denote the maximum number of nonzero elements in any row of the matrix $\mathbf{H}_2$ as $s_r$ and the maximum number of nonzero elements in any column as $s_c$. Let $d = \max \{d_1,s_r,s_c\}$.
\end{definition}
We then assume that the matrix elements are given by a set of oracles that resemble the oracles used for sparse Hamiltonian simulation.  One notable difference, however, is that here the matrices in question are not Hermitian.  For that reason we need to use a slightly more involved definition given to simulate these equations which we give below.
\begin{definition}
\label{Def:Oracles-nonlinear_Schr\"{o}dinger_eq}
    (State prep and Access Model for the Nonlinear Schr\"{o}dinger Equation )
    Let $$H_{\text{max}} = \max_{jk}\{|\mathbf{H}_1(j,k)|,|\mathbf{H}_2(j,k)|\}.$$ Let $\bar{H}_1(j,k) $ be a binary expansion of the $j^{\text{th}}$ row and $k^{\text{th}}$ column of the Hamiltonian matrix $\mathbf{H}_1$. From the binary representation, we define a basis state $\ket{\bar{H_1}({j,k})}$. Similarly, we can define basis states representing elements of the matrix $\mathbf{H}_2$. We now assume access to the following oracles to prepare the basis states representing terms in the Hamiltonian $\mathbf{H}_1$ and the rectangular matrix $\mathbf{H}_2$. For the rectangular matrix $\mathbf{H}_2$, we create a padded matrix $\mathbf{H}_2' = \ket{0}\otimes \mathbf{H}_2 $.  

\begin{align*}
O_{H_1}: 
    \ket{j}\ket{k}\ket{0} &\longrightarrow \ket{j}\ket{k}\ket{\bar{H_1}({j,k})} 
\end{align*}
\begin{align}
O_{H_2'}: 
    \begin{cases}
        \ket{j}\ket{k}\ket{0} \longrightarrow \ket{j}\ket{k}\ket{\bar{H_2}({j,k})} &\text{if } j \leq N-1 \\ \nonumber
       \ket{j}\ket{k}\ket{0} \longrightarrow  \ket{j}\ket{k}\ket{0} &\text{if } j > N-1 \nonumber
    \end{cases}
\end{align}
We also assume access to the following sparse access oracles which return the index of nonzero elements in any given row or column for the matrices $\mathbf{H}_1$ and $\mathbf{H}_2$.
\begin{align*}
O_{s_{H_1}}: 
    \ket{i}\ket{k} &\longrightarrow \ket{i}\ket{s_{ik}} 
\end{align*}
Here, $s_{ik}$ denotes the index of the $k^{\text{th}}$ nonzero element in the $i^{\text{th}}$ row of the matrix $\mathbf{H}_1$. Since $\mathbf{H}_2$ is a non-Hermitian matrix, we define the oracles $O_{r_{H_2}}$ and $O_{c_{H_2}}$, which are defined as:
\begin{align*}
O_{r_{H_2'}}: 
    \ket{i}\ket{k} &\longrightarrow \ket{i}\ket{r_{ik}} 
\end{align*}
\begin{align*}
O_{c_{H_2'}}: 
   \ket{i}\ket{k} &\longrightarrow \ket{i}\ket{c_{ik}} \\
\\
\end{align*}
Here, $r_{ik}$ denote the index of $k^{\text{th}}$ nonzero element in the $i^{\text{th}}$ row of matrix $\mathbf{H}_2$ and $c_{ik}$ denote the index of $k^{\text{th}}$ nonzero element in the $i^{\text{th}}$ coloumn of matrix $\mathbf{H}_2$.
\\

Finally, we also assume access to the state preparation Unitary $\mathcal{W}(\eta,k,\aleph)$, which prepares the following initial state:
\begin{align}
    \mathcal{W}(\eta,k,\aleph) : \bm{\ket{0}} \rightarrow \frac{1}{\aleph}\sum_i\frac{1}{\eta^{k-i}}\ket{i-1}\otimes \bm{\ket{0}}^{\otimes k-i}\bm{\ket{\psi(0)}}^{\otimes i} 
\end{align}

\end{definition}


\begin{definition}[Nonlinear Coupled Classical Oscillator System]
\label{Def: nonlinear_coupled_oscillator_system}
   In a nonlinear coupled oscillator system with quadratic nonlinearity, the governing differential equation of the system contains an additional quadratic term, in addition the linear term. The stiffness matrix of the system, $\mathbf{G}_1$, and the matrix $\mathbf{K}_1$ are constructed as per \autoref{Def: coupled_oscillator_system}.  The governing differential equation of the nonlinear coupled classical oscillator system with quadratic nonlinearity can be written as:
    \begin{equation}
        \mathbf{M} \ddot{\mathbf{x}} = -\mathbf{K}_1 \mathbf{x} + \mathbf{K}_2 \mathbf{x}\otimes \mathbf{x} 
    \end{equation}
    Here, $\mathbf{K}_1 \in \mathbb{R}^{N \times N} $ and the rectangular matrix $\mathbf{K}_2 \in \mathbb{R}^{N \times N^2}$. Further we define $E$ to be a constant such that $E \ge \frac{1}{2} \left(  \dot{\mathbf{x}}(0)^T\mathbf{M} \dot{\mathbf{x}}(0) + \mathbf{x}(0)^T\mathbf{K}_1\mathbf{x}(0) \right)$.
\end{definition}
We assume that the matrix elements for the oscillator system are accessed via a number of quantum oracles
\begin{definition}
\label{Def:Oracles-nonlinear-oscillator-system}
    (Access Model for nonlinear oscillator system)
    Let the maximum mass obey $m_{\max} \geq m_j$ and let $m_j$ be assumed to be fixed point binary representations of the mass as a fraction of $m_{\max}$ 
    and the binary fraction is specifically represented as the basis state $\ket{\bar{m}_j}$. Similarly, we represent the stiffness parameter $k_{ij}$, and the frequencies $\omega_{ij}$ as basis states $\ket{\bar{k}_{ij}}$, and $\ket{\bar{\omega}_{ij}}$. Using this representation, we define oracles for accessing elements of the matrices $\mathbf{M}$, $\mathbf{G}_1$, and $\mathbf{K}_2$. Specifically, for the matrix $\mathbf{M}$, we define the oracle $O_M: \ket{j}\ket{0} \rightarrow \ket{j}\ket{\bar{m}_{j}} $ to yield the binary representation of $m_j$. The elements of linear and nonlinear stiffness matrices are accessed by the following oracles:

\begin{align*}
O_{G_1}: \ket{j}\ket{k}\ket{0} &\longrightarrow \ket{j}\ket{k}\ket{\bar{G_1}(j,k)} \\
O_{K_2}: \ket{j}\ket{k}\ket{0} &\longrightarrow \ket{j}\ket{k}\ket{\bar{K_2}(j,k)} \\
\end{align*}
We also assume the existence of oracles $O_{c_{G_1}}$ etc that yield the column / row numbers of the non-zero matrix elements of the matrices as defined in~\autoref{Def:Oracles-nonlinear_Schr\"{o}dinger_eq}
\end{definition}

The time-dependent coupled oscillator is similarly defined, but here we make a few specific assumptions about the underlying dynamical system.  In particular, we need to assume that a Fourier expansion exists to represent the time-dependence of the spring constants.  We enforce this so that we can specify the time-dependence using a finite amount of information and further so that we can link the dynamics of two-oscillators two each other through a set of nonlinear three body interactions with a fictitious mass introduced into the system.  The assumption that a finite-term Fourier series expansion describes the precise time-dependence simplifies our discussion by eliminating the need to discuss the error tolerance that occurs from truncating the series and discussing the lack of pointwise convergence for general Fourier series expansions.

\begin{definition}[Time-Dependent Coupled Classical Oscillator System]
\label{Def: time_dependent_coupled_oscillator_system}
A time-dependent coupled classical oscillator system consists of $N$ positive masses $\{m_j\}_{j=1}^N$, interconnected by spring elements with time-varying stiffness. We define the mass matrix $\mathbf{M}$ of the system to be a diagonal matrix with terms $\mathbf{M}(j,j) = m_j$. Let $k_{ji}(t)$ denote the time-dependent spring constant of the spring interconnecting masses $m_j$ and $m_i$, and let $k_{jj}(t)$ denote the time-dependent spring constant connecting mass $m_j$ to the wall. Let $\mathbf{G}(t) \in \mathbb{R}^{N \times N}$ be the time-dependent stiffness matrix of the coupled oscillator system. The elements of matrix $\mathbf{G}(t)$ are defined as:
\begin{align*}
    \mathbf{G}_{ij}(t) &= k_{ij}(t) = k_{ji}(t) \\
    \mathbf{G}_{jj}(t) &= k_{jj}(t)
\end{align*}
Let $\mathbf{x} = [x_1,x_2,..x_N]^T$ be the displacement vector of the oscillator system. The governing differential equation of the time-dependent coupled classical oscillator system is given by:
\begin{equation}
    \mathbf{M}\ddot{\mathbf{x}} = -\mathbf{K}(t) \mathbf{x}
\end{equation}
Here, $\mathbf{K}(t)$ is the time-dependent incidence matrix of the weighted graph defined by $\mathbf{G}(t)$. Specifically, the diagonal terms of the matrix $\mathbf{K}(t)$ are given by: $\mathbf{K}(t)(j,j) = \sum_i k_{ji}(t)$ and the off-diagonal terms are $\mathbf{K}(t)(j,i) = - k_{ji}(t)$. We assume that $k_{ji}(t)$ can be written as $k_{ji}(t) = \sum_l k_{ji,l} \cos(\omega_{ji,l}t + \phi_{ji,l})$ where $k_{ji,l} > 0$, and the matrix $\mathbf{K}(t)$ remains positive semidefinite for all $t$.
\end{definition}

\begin{definition}[Time-Dependent Forced Coupled Classical Oscillator System]
\label{Def: time_dependent_forced_coupled_oscillator_system}
In a time-dependent forced coupled classical oscillator system, in addition to the forces exerted by the time-varying spring elements, external time-dependent forces $f_i(t)$ act on mass $m_i$. The time-dependent stiffness matrix of the system, $\mathbf{G}(t)$, and the matrix $\mathbf{K}(t)$ are constructed as per \autoref{Def: time_dependent_coupled_oscillator_system}. The governing differential equation of the time-dependent forced coupled classical oscillator system is given by:
\begin{equation}
    \mathbf{M} \ddot{\mathbf{x}} = -\mathbf{K}(t) \mathbf{x} + \mathbf{f}(t)
\end{equation}
Here, $\mathbf{f}(t) = [f_1(t), f_2(t),...,f_N(t)]$ and we assume that both the time-dependent forces and spring constants can be written as: $f_i(t) = \sum_j f_{ij} \cos(\omega_{ij}t + \phi_{ij})$ and $k_{ji}(t) = \sum_l k_{ji,l} \cos(\omega_{ji,l}t + \phi_{ji,l})$.
\end{definition}
These elements are accessed via the following oracles, which give us the similar ability to locate the nonzero matrix elements of the relevant matrices.
\begin{definition}
\label{Def:Oracles-time-dependent-oscillator-system}
(Access Model for time-dependent oscillator system)
Let $m_{max} \geq m_j$ and let $m_j = m_{max} [.b_{j,1} b_{j,2}...]$, $b_{j,i} \in \{0,1\}$. Here, $[.b_{j,1} b_{j,2}...] $ is a binary fraction, and we use the bit representation to define the basis state $\ket{\bar{m}_j}$. Similarly, we represent the stiffness parameters $k_{ij,l}$, frequencies $\omega_{ij,l}$, and phase terms $\phi_{ij,l}$ as basis states $\ket{\bar{k}_{ij,l}}$, $\ket{\bar{\omega}_{ij,l}}$, and $\ket{\bar{\phi}_{ij,l}}$. Using this representation, we define oracles for accessing elements of the matrices $\mathbf{M}$, $\mathbf{G}(t)$, and $\mathbf{K}(t)$:

\begin{align*}
O_M&: \ket{j}\ket{0} \longrightarrow \ket{j}\ket{\bar{m}_{j}}
\\
O_{G}&: 
    \ket{j}\ket{k}\ket{l}\ket{0} \longrightarrow \ket{j}\ket{k}\ket{l}\ket{\bar{k}_{jk,l}} 
\\
O_{\omega}&: 
    \ket{j}\ket{k}\ket{l}\ket{0} \longrightarrow \ket{j}\ket{k}\ket{l}\ket{\bar{\omega}_{jk,l}} 
\\
O_{\phi}&:
    \ket{j}\ket{k}\ket{l}\ket{0} \longrightarrow \ket{j}\ket{k}\ket{l}\ket{\bar{\phi}_{jk,l}} 
\end{align*}
We additionally define oracles to access the forcing terms for the time-dependent forced oscillator system:

\begin{align*}
O_{F}: \ket{j}\ket{k}\ket{0} &\longrightarrow \ket{j}\ket{k}\ket{\bar{f}_{jk}} 
\\
O_{F_{\omega}}: 
    \ket{j}\ket{k}\ket{0} &\longrightarrow \ket{j}\ket{k}\ket{\bar{\omega}_{jk}} 
\\
O_{F_{\phi}}: 
    \ket{j}\ket{k}\ket{0} &\longrightarrow \ket{j}\ket{k}\ket{\bar{\phi}_{jk}} \\
\end{align*}
\end{definition}

\section{Problem Statement and Main Results}\label{sec:results}

We give a brief summary of the main results in this section. Our primary focus is on coupled oscillator systems - a fundamental model that appears across diverse areas of physics and engineering, from molecular dynamics to structural mechanics. We consider several variations of the coupled classical oscillator simulation problem, starting with the case where external time-dependent forces act on each mass. Specifically, consider a network of $2^n$ point masses connected by springs, where each mass experiences both the internal forces from its spring connections and external time-dependent forces. This system represents a broad class of physical phenomena, including vibrating mechanical structures, molecular systems, and coupled electronic oscillators.  \\


Now, let us state the formal statement of the forced oscillator simulation problem. In this problem, our goal is to simulate a coupled classical oscillator system consisting of $2^n$ masses connected together with spring elements, with external time dependent forces acting on each mass. The formal statement of the problem is given below.

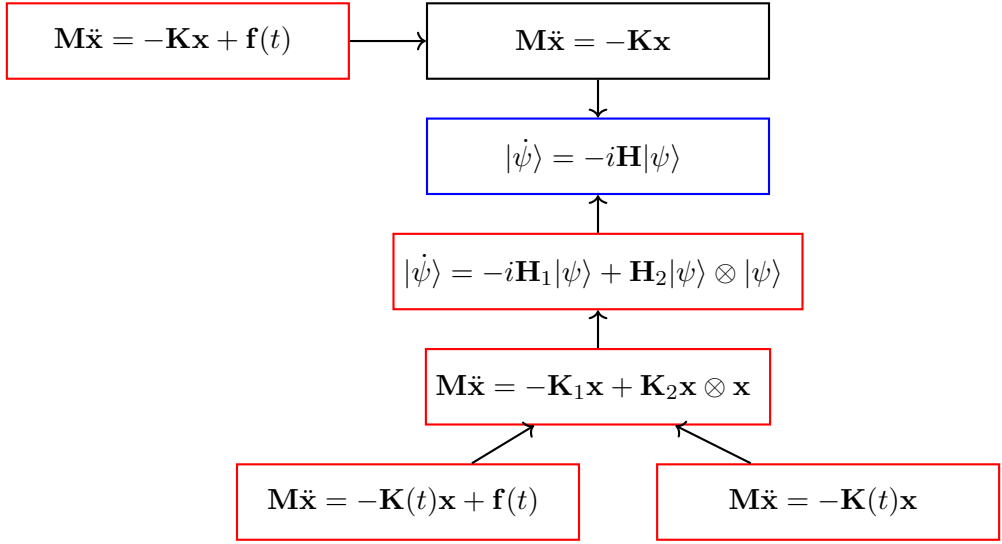
\begin{figure}[t]
\begin{center}
\begin{tikzpicture}[
    node distance=.5cm and 1cm, 
    every node/.style={draw, minimum width=4.5cm, minimum height=1cm, align=center},
    blackbox/.style={draw=black, thick},
    bluebox/.style={draw=blue, thick},
    redbox/.style={draw=red, thick}
]

\node (forcing) [redbox] { $\mathbf{M} \ddot{\mathbf{x}} = -\mathbf{K} \mathbf{x} + \mathbf{f}(t)$ };
\node (reduced1) [blackbox, right=of forcing] { $\mathbf{M} \ddot{\mathbf{x}} = -\mathbf{K} \mathbf{x}$ };

\node (Schr\"{o}dinger) [bluebox, below=of reduced1] { $|\dot{\psi} \rangle = - i \mathbf{H} |\psi \rangle$ };

\node (nonlinear) [redbox, below=of Schr\"{o}dinger] { $|\dot{\psi} \rangle = - i \mathbf{H}_1 |\psi \rangle + \mathbf{H}_2 |\psi \rangle \otimes |\psi \rangle$ };
\node (nonlinearOsc) [redbox, below=of nonlinear] { $\mathbf{M} \ddot{\mathbf{x}} = -\mathbf{K}_1 \mathbf{x} + \mathbf{K}_2 \mathbf{x} \otimes \mathbf{x}$ };

\node (forcingVar) [redbox, below=of nonlinearOsc, xshift=-2.5cm] { $\mathbf{M} \ddot{\mathbf{x}} = -\mathbf{K}(t) \mathbf{x} + \mathbf{f}(t)$ };
\node (varcoeff) [redbox, right=of forcingVar] { $\mathbf{M} \ddot{\mathbf{x}} = -\mathbf{K}(t) \mathbf{x}$ };

\draw[->, thick] (forcing) -- (reduced1);
\draw[->, thick] (reduced1) -- (Schr\"{o}dinger);
\draw[->, thick] (nonlinear) -- (Schr\"{o}dinger);
\draw[->, thick] (nonlinearOsc) -- (nonlinear);
\draw[->, thick] (forcingVar) -- (nonlinearOsc);
\draw[->, thick] (varcoeff) -- (nonlinearOsc);

\end{tikzpicture}
\end{center}
\caption{Conceptual Diagram. The diagram shows two reduction chains leading to the Schrödinger equation.}
\label{fig: conceptual_diagram}
\end{figure}

%


\begin{restatable}[Forced Oscillator Simulation Problem]{problem}{forcedoscillator}
\label{prblm 1}
 Consider the forced coupled oscillator system described in ~\autoref{Def: forced_coupled_oscillator_system}. Assume we have access to elements of the mass matrix $\mathbf{M} \in \mathbb{R}^{N \times N}$, $d_1$ sparse stiffness matrix $\mathbf{G} \in \mathbb{R}^{N \times N}$, and the time-dependent force $\mathbf{f}(t) \in \mathbb{R}^N$ through oracles defined in~\autoref{Def:Oracles-forced-oscillator-system}. The governing differential equation of the forced coupled oscillator system is given by:
\begin{equation}
    \label{eq:forced_oscillator}
    \mathbf{M} \ddot{\mathbf{x}} = - \mathbf{K} \mathbf{x} + \mathbf{f}(t)
\end{equation}


The Forced Oscillator Simulation Problem involves preparing a quantum state $\ket{\psi{(t)}}$, using a minimal number of queries to the oracles and gate operations, that encodes the solution to Eq.~\eqref{eq:forced_oscillator} such that the error in the Euclidean norm between the normalized and projected quantum state and the normalized solution to Eq.~\eqref{eq:forced_oscillator} obeys
\begin{equation}
 \left\|  \frac{\mathbf{P}'\ket{\psi{(t)}}}{\|\mathbf{P}'\ket{\psi{(t)}}\|_2}  - \frac{[\mathbf{x},\dot{\mathbf{x}}]^T}{\|[\mathbf{x},\dot{\mathbf{x}}]^T\|_2} \right\|_2 \leq \epsilon 
\end{equation}
for some $\epsilon>0, t>0$ where $\mathbf{P}'$ is a rectangular matrix that removes any fictitious or ancillary variables in the mapping.


\end{restatable}
A related problem is studied in \cite{danz2024calculating}, but the algorithm is restricted only to modal analysis of the system. 
Our first main result is a quantum algorithm that can efficiently prepare the state defined in \autoref{prblm 1}. The main insight behind our result is that the forced oscillator system can be embedded into a larger coupled oscillator system containing auxillary masses which exert desired time dependent external forces to other masses in the system. Applying perturbation theory, we show that the parameters of this auxillary system can be chosen such that the dynamics of a subspace of the larger oscillator system closely resembles the desired dynamics of the forced oscillator system defined in \autoref{prblm 1}.




\begin{restatable}[Forced Oscillator Simulation]{theorem}{QuantumAlgorithmThm}
\label{thm:forced_oscillator_result}
     There exists a quantum algorithm that can solve \autoref{prblm 1} with $Q = \mathcal{O}(\tau + \log(\frac{1}{\epsilon}))$ queries to the oracles for $\mathbf{G}$, $\mathbf{M}$ and $\mathbf{f}(t)$. The algorithm uses $\mathcal{O}(\log(N(l+1)))$ qubits and require $G = \mathcal{O}\left( Q \times \log^2\left( \frac{N(l+1) \tau}{\epsilon} \frac{m_{\text{max}}}{m_{\text{min}}} \right) \right)$ two qubit gates. Here, $\tau = t \sqrt{2\alpha d} \geq 1$ where $d = d_1 +l$ is the sum of maximum number of nonzero elements in any row of the stiffness matrix $\mathbf{G}$ and the maximum number of fourier terms in the expansion of $\mathbf{f}(t)$. The constant, $\alpha = \max_{jk} \left({\frac{k_{jk}}{m_j}} , \omega_{jk}^2\right)$. 
 
\end{restatable}
The complexity of the simulation scales logarithmically with respect to the number of masses and is proportional to the square root of the maximum number of terms in the decomposition of the time-dependent external force $\mathbf{F}(t)$. Furthermore, the scaling with respect to time is linear, with an extra factor proportional to the square root of the sparsity of the matrix $\mathbf{G}$.  Since a special case of the forced oscillator problem was proven to be BQP-complete in \cite{babbush2023exponential}, classical algorithms for solving \autoref{prblm 1} are expected to exhibit exponential scaling with the number of masses in the oscillator system.
\\

The quantum analogue of the forced oscillator problem is a simulation problem wherein the dynamics is non-unitary, and one notable example of such dynamics arises in the simulation of the nonlinear Schr\"{o}dinger equation. The formal definition of the problem is given below.

\begin{restatable}[Simulating nonlinear quantum dynamics]{problem}{nonlinear_Schr\"{o}dinger}
\label{nonlinear_Schr\"{o}dinger}
 Consider the nonlinear Schr\"{o}dinger Equation described in ~\autoref{Def: nonlinear_Schr\"{o}dinger_eq}. Assume we have access to elements of the Hermitian matrix $\mathbf{H}_1 \in \mathbb{C}^{N \times N}$, rectangular $\mathbf{H}_2 \in \mathbb{C}^{N \times N^2}$, through oracles defined in~\autoref{Def:Oracles-nonlinear_Schr\"{o}dinger_eq}. 
\begin{equation}
    \label{eq:nonlinear_Schr\"{o}dinger}
    \dot{\ket{\bm{\psi}} } = -i\mathbf{H}_1 \ket{\bm{\psi}} + \mathbf{H}_2 \ket{\bm{\psi}} \otimes \ket{\bm{\psi}}
\end{equation}

The nonlinear quantum dynamics simulation problem involves preparing a quantum state $\widetilde{\ket{\psi{(t)}}}$, using a minimal number of queries to the oracles and gate operations, that encodes the solution to Eq.~\eqref{eq:nonlinear_Schr\"{o}dinger} such that the error in the Euclidean norm between the prepared quantum state and the actual solution to Eq.~\eqref{eq:nonlinear_Schr\"{o}dinger} obeys
\begin{equation}
 \left\| \frac{\widetilde{\ket{\bm{\psi}{(t)}}} }{\widetilde{ \braket{\bm{\psi}{(t)}|\bm{\psi}{(t)}}}}    - \frac{\ket{\bm{\psi}(t)}}{\braket{ \bm{\psi}{(t)}|\bm{\psi}{(t)}}}\right\|_2 \leq \epsilon 
\end{equation}
for some $\epsilon>0, t>0$.  

\end{restatable}
In this case, although one can embed the nonlinear dynamical system in a higher dimensional linear dynamical system using Carlemann Embedding techniques \cite{carleman1932application}, the constructed Carlemann Operator need not be Hermitian, and hence the dynamics of the Carlemann Truncated system is non-unitary. 
\\

While nonlinear quantum dynamics generally cannot be efficiently simulated using quantum algorithms without leading to violations in the Holevo-Helstrom bound \cite{liu2021efficient}, we demonstrate that certain classes of nonlinear Schr\"{o}dinger equations become amenable to efficient quantum simulation under specific conditions. The key insight is that when the nonlinearity satisfies particular structural properties - specifically when the system maintains norm-preservation and the dynamics obey non-resonant conditions, we can employ a combination of Carlemann linearization and novel symmetrization techniques to transform the nonlinear dynamics into a linear quantum evolution in a higher-dimensional space. Our main theoretical contribution is a quantum algorithm that efficiently prepares quantum states evolving according to the nonlinear Schrödinger equation stated in \autoref{nonlinear_Schr\"{o}dinger}, with complexity scaling linearly in the simulation time. 

\begin{restatable}[Nonlinear Quantum Dynamics Simulation]{theorem}{NSEQuantumAlgorithmThmResonanceCond}
    
    Consider the nonlinear Schr\"{o}dinger Equation stated in \autoref{nonlinear_Schr\"{o}dinger}. Let $\beta = \braket{\psi(0)|\psi(0)}$, and assume $\braket{\psi(t)|\psi(t)} \leq \beta$, for $t \in [0,T]$. Let $\{\lambda_1,,,\lambda_n\}$ be the eigenvalues of the matrix $-i\mathbf{H}_1$. Then, there exists a quantum algorithm that can solve \autoref{nonlinear_Schr\"{o}dinger} with $Q_1$ queries to the oracles for $\mathbf{H}_1$, and $\mathbf{H}_2$. 
     
     \[Q_1 = \mathcal{O}\left( \alpha k^2 t + k\log\left(\left[{\|\mathbf{H}_2 \|k^2} \left(1 + \frac{\beta (1-\beta^k)}{\epsilon (1-\beta)}\right)t \right] \right) \right)\] and a single query to the state preparation oracle $\mathcal{W}$. Here,
     \[
    k \in \widetilde{\mathcal{O}} \left( \frac{\log(\frac{CT}{\epsilon})}{\log(1/R_r)} \right)
     \]
     Here, $R_r = \frac{4 e \beta \|\mathbf{H}_2\|}{\Delta} < 1$, $W(.)$ is the Lambert W function, $C =  \|\mathbf{H}_2\|\beta^2$, and $\Delta$ is defined as follows:
     \[
\Delta := \inf_{k \in [n]} \inf_{\substack{m_j \geq 0 \\ \sum_{j=1}^n m_j \geq 2}}
\left(|
\lambda_k - \sum_{i=1}^n m_j \lambda_j
|\right),
\]
The algorithm uses $Q_2$ qubits.
     \[ Q_2  \in
     \mathcal{O} \left(\log(k) + k\log(N) \right)\]
      Here  $\alpha = d\max_{jk}\{|\mathbf{H}_1(j,k)|,|\mathbf{H}_2(j,k)|\}$, where $d$ is the maximum number of non zero elements in any row or column of the matrices $\mathbf{H}_1$ and $\mathbf{H}_2$.
\end{restatable}

When the nonlinear Schr\"{o}dinger Equation doesn't satisfy the no-resonance condition defined in \cite{wu2024quantum}, the truncated Carlemann System is only known to converge for small $t$. Hence, in this case, we can efficiently simulate the dynamics only when the product of norm of the nonlinear coupling matrix, time, and the norm of initial conditions is less than one. The formal statement of the theorem is given below. 

\begin{restatable}[Nonlinear Quantum Dynamics Simulation]{theorem}{NSEQuantumAlgorithmThm}
    
    Consider the nonlinear Schr\"{o}dinger Equation stated in \autoref{nonlinear_Schr\"{o}dinger}. Let $\beta = \braket{\psi(0)|\psi(0)}$, and assume $\braket{\psi(t)|\psi(t)} \leq \beta$, for $t < \frac{1}{ \beta\|\mathbf{H}_2\|}$. Then, there exists a quantum algorithm that can solve \autoref{nonlinear_Schr\"{o}dinger} with $Q_1$ queries to the oracles for $\mathbf{H}_1$, and $\mathbf{H}_2$. 
     
     \[Q_1 = \mathcal{O}\left( \alpha k^2 t + k\log\left(\left[{\|\mathbf{H}_2 \|k^2} \left(1 + \frac{\beta (1-\beta^k)}{\epsilon (1-\beta)}\right)t \right] \right) \right)\] and a single query to the state preparation oracle $\mathcal{W}$. Here,\[
     k \in \mathcal{O}\left(\frac{\log(1/\epsilon)}{\log(1/\beta \|\mathbf{H}_2\|)}\left(1 + \frac{\log(t)}{\log(1/ \beta \|\mathbf{H}_2\|)} \right) \right)
     \]
     The algorithm uses $Q_2$ qubits.
     \[ Q_2  \in
     \mathcal{O} \left(\log(k) + k\log(N) \right)\]
      Here  $\alpha = d\max_{jk}\{|\mathbf{H}_1(j,k)|,|\mathbf{H}_2(j,k)|\}$, where $d$ is the maximum number of non zero elements in any row or column of the matrices $\mathbf{H}_1$ and $\mathbf{H}_2$.
\end{restatable}

Having introduced efficient quantum algorithms for simulating nonlinear quantum dynamics, we now pose the following question: What other simulation problems can be mapped to the time evolution of a nonlinear Schrödinger equation? In \autoref{sec:nonlinear_oscillator_system}, we address this question by demonstrating that the dynamics of nonlinear coupled oscillator systems can be mapped into nonlinear quantum dynamics. The formal problem statement is defined below:
\begin{restatable}[Nonlinear Oscillator Simulation Problem]{problem}{nonlinearoscillator}
\label{prblm: nonlinear_oscillator}
Consider the nonlinear coupled oscillator system described in ~\autoref{Def: nonlinear_coupled_oscillator_system}. Assume we have access to elements of the mass matrix $\mathbf{M} \in \mathbb{R}^{N \times N}$, linear stiffness matrix $\mathbf{G}_1 \in \mathbb{R}^{N \times N}$, and the quadratic coupling matrix $\mathbf{K}_2 \in \mathbb{R}^{N \times N^2}$ through oracles defined in~\autoref{Def:Oracles-nonlinear-oscillator-system}. The governing differential equation of the nonlinear coupled oscillator system is given by:
\begin{equation}
    \label{eq:nonlinear_oscillator}
    \mathbf{M} \ddot{\mathbf{x}} = -\mathbf{K}_1 \mathbf{x} + \mathbf{K}_2 \mathbf{x}\otimes \mathbf{x}
\end{equation}

The Nonlinear Oscillator Simulation Problem involves preparing a quantum state $\ket{\psi{(t)}}$, using a minimal number of queries to the oracles and gate operations, that encodes the solution to Eq.~\eqref{eq:nonlinear_oscillator} such that the error in the Euclidean norm between the normalized and projected quantum state and the normalized solution to Eq.~\eqref{eq:nonlinear_oscillator} obeys
\begin{equation}
 \left\|  \frac{\mathbf{P}'\ket{\psi{(t)}}}{\|\mathbf{P}'\ket{\psi{(t)}}\|_2}  - \frac{[\mathbf{x},\dot{\mathbf{x}}]^T}{\|[\mathbf{x},\dot{\mathbf{x}}]^T\|_2} \right\|_2 \leq \epsilon 
\end{equation}
for some $\epsilon>0, t \in [0,T]$ where $\mathbf{P}'$ is a rectangular matrix that removes any fictitious or ancillary variables in the mapping. Let's also assume we have access to a state preparation oracle $\mathcal{W}$ that can efficiently prepare the initial state, $\ket{ \bm{ \psi{(0)}} } \propto [\mathbf{x}(t), \dot{\mathbf{x}}(t)]^T$.
\end{restatable}
Next, we show that the dynamics of the nonlinear oscillator system can be efficiently simulated when the nonlinear oscillator satisfies non-resonant conditions and when the norm of the dynamics is upper bounded by a constant. These conditions are necessary for the nonlinear Schr\"{o}dinger equation simulation subroutine which we utilize in simulating the nonlinear oscillator. 
\begin{restatable}[Nonlinear Oscillator Simulation]{theorem}{NOQuantumAlgorithmThm}
\label{thm:nonlinear_oscillator_simulation}
Consider the nonlinear oscillator system described in \autoref{Def: nonlinear_coupled_oscillator_system}. Let $\alpha = \max_{i,j}\left(\sqrt{\frac{k_{ij}}{m_j}}\right)$, and let $\beta = \frac{1}{2 }\left(\dot{\mathbf{x}}(0)^T\mathbf{M}\dot{\mathbf{x}}(0) + \mathbf{x}^T(0)\mathbf{K}_1\mathbf{x}(0)\right)$. Let $\{\lambda_1,,,\lambda_n\}$ be the set containing square root of eigenvalues of the matrix $\mathbf{A} = \sqrt{\mathbf{M}^{-1}} \mathbf{K}_1 \sqrt{\mathbf{M}^{-1}} $, and let $k_{\max}^{(2)} = \max_{ij}\mathbf{K}_2(i,j)$, $k_{\min}^{(1)} = \min_{ij} \mathbf{K}_1(i,j)$ and let $m_{min} = \min_i \mathbf{M}(i,i)$. Further assume $R_r := \frac{4 e \beta d k^{(2)}_{\max}}{k^{(1)}_{\min} (m_{\min})^{3/2}\Delta }<1$, $\|\mathbf{x}(t) \| \leq c$ and $\Delta$ is defined as follows:
     \[
\Delta := \inf_{k \in [n]} \inf_{\substack{m_j \geq 0 \\ \sum_{j=1}^n m_j \geq 2}}
\left(|
\lambda_k - \sum_{i=1}^n m_j \lambda_j
|\right)>0,
\]
    There exists a Quantum Algorithm that solves Problem \ref{prblm: nonlinear_oscillator} with at least $G$ queries to the Oracles defined in \autoref{Def:Oracles-nonlinear-oscillator-system}.  Here, 
    \[
   G=  \mathcal{O}\left( \alpha k^2 t + \frac{(k-1)}{2}\log\left(\left[\frac{d k^{(2)}_{\max} k\left( k +1 \right)}{2k^{(1)}_{\min} (m_{\min})^{3/2}} \left(1 + \beta\frac{1-\beta^k}{(1-\beta)\epsilon} \right)t  \right] \right) \right)
    \]

where, the truncation order $k$, is given by:

\[
k \in 
 \widetilde{\mathcal{O}}\left(\frac{\log\left( \frac{ CT}{\epsilon}   \right)}{\log(1/R_r)} \right)
\]

Here, $C  \leq \frac{d k^{(2)}_{\max} \beta^2 }{2k^{(1)}_{\min} (m_{\min})^{3/2}} $.

\end{restatable}
\autoref{thm:nonlinear_oscillator_simulation} tells us that we can  efficiently simulate the dynamics of nonlinear oscillator systems, provided the system satisfy three conditions. First, the oscillator system should satisfy non-resonance conditions, which implies the eigenvalues of the matrix $\mathbf{K}_1$ should not be integer multiples of each other. Secondly, the norm of the state vector should be upper bounded by a constant. Third, the quantity $R_r$, which is the ratio of the norm of nonlinearities to the minimum difference between eigenvalues of the matrix $\mathbf
{K}_1$ should be strictly less than $1$. Under these conditions, the truncation error due to Carlemann Linearization can be controlled by increasing the order of truncation, for arbitrary time horizon. \\

Having demonstrated the quantum simulation of nonlinear oscillator systems using a reduction to nonlinear Schrödinger dynamics, we now turn our attention to another class of systems: those with time-dependent stiffness parameters. The formal statement of the problem is given below:
\begin{restatable}[Time-Dependent Stiffness Oscillator Simulation Problem]{problem}{timedependentoscillator}
\label{prblm: time_dependent_oscillator}
Consider the time-dependent coupled oscillator system described in ~\autoref{Def: time_dependent_coupled_oscillator_system}. Assume we have access to elements of the mass matrix $\mathbf{M} \in \mathbb{R}^{N \times N}$ and time-dependent stiffness matrix $\mathbf{G}(t) \in \mathbb{R}^{N \times N}$ through oracles defined in~\autoref{Def:Oracles-time-dependent-oscillator-system}. The governing differential equation of the time-dependent coupled oscillator system is given by:
\begin{equation}
    \label{eq:time_dependent_oscillator}
    \mathbf{M} \ddot{\mathbf{x}} = -\mathbf{K}(t) \mathbf{x}
\end{equation}

The Time-Dependent Oscillator Simulation Problem involves preparing a quantum state $\ket{\psi{(t)}}$, using a minimal number of queries to the oracles and gate operations, that encodes the solution to Eq.~\eqref{eq:time_dependent_oscillator} such that the error in the Euclidean norm between the normalized and projected quantum state and the normalized solution to Eq.~\eqref{eq:time_dependent_oscillator} obeys
\begin{equation}
 \left\|  \frac{\mathbf{P}'\ket{\psi{(t)}}}{\|\mathbf{P}'\ket{\psi{(t)}}\|_2}  - \frac{[\mathbf{x},\dot{\mathbf{x}}]^T}{\|[\mathbf{x},\dot{\mathbf{x}}]^T\|_2} \right\|_2 \leq \epsilon 
\end{equation}
for some $\epsilon>0, t>0$ where $\mathbf{P}'=[\mathbf{I},\mathbf{0}]$ is a rectangular matrix that removes any fictitious or ancillary variables in the mapping. Let's also assume we have access to a state preparation oracle $\mathcal{W}$ that can efficiently prepare the initial state, $\ket{ \bm{ \psi{(0)}} } = \bm{\mathcal{W}} \ket{{0}}$.
\end{restatable}
We solve \autoref{prblm: time_dependent_oscillator} by embedding the dynamics of oscillator systems with time dependent stiffness matrices in a higher dimensional noninear oscillator system. This reduction enables us to apply \autoref{thm:nonlinear_oscillator_simulation} to simulate the dynamics of the resultant nonlinear oscillator system, constructed based on parameters defined in \autoref{prblm: time_dependent_oscillator}. The formal statement of our result is given below. 
\begin{restatable}[Time-Dependent Oscillator Simulation]{theorem}{TDOQuantumAlgorithmThm}
\label{thm: time_dependent_oscillator_sim}
    Consider the time dependent oscillator system defined in \autoref{Def: time_dependent_coupled_oscillator_system}. Let the time-varying stiffness element be $k_{ij}(t) =  \sum_{l=1}^m  \alpha_{ij,l} y_{ij,l}(t)$, where the time-dependent function $y_{ij,l}(t)$ is defined as $y_{ij,l}(t) = A_{ij,l} \cos(\omega_{ij,l}t + \phi_{ij,l})$.There exists a quantum algorithm that solves \autoref{prblm time_var_stiffness} with at least $G$ calls to oracles defined in \autoref{Def:Oracles-time-dependent-oscillator-system}. Here, $G$ is given by:
    \[
      G \in  \mathcal{O}\left( \alpha k^2 t + k\log\left(\left[\frac{d k^{(2)}_{\max} k\left( k +1 \right)}{2k^{(1)}_{\min} (m_{\min})^{3/2}} \left(1 + \frac{\beta}{(1-\beta)\epsilon} \right)t  \right] \right) \right)
    \]

    Here,  $\alpha =  \max_{ij} \omega_{ij}$ , $k_{\max}^{(2)} = \max_{ij}\mathbf{K}_2(i,j) = \max_{ij,l}\{\alpha_{ij,l}\}$, $k_{\min}^{(1)} = \min_{ij} \{\omega_{ij}\}$, and the truncation order $k$, is given by:
\[
     k \in \widetilde{\mathcal{O}}\left(\frac{\log(CT/\epsilon)}{\log(1/R_r)}\right)
     \]
     Here, $C  \leq \frac{d k^{(2)}_{\max} \beta^2 }{2k^{(1)}_{\min} (m_{\min})^{3/2}} $, and $R_r := \frac{4 e \beta d k^{(2)}_{\max}}{k^{(1)}_{\min} (m_{\min})^{3/2}\Delta }$, where $\Delta$ is defined as follows:
     \[
\Delta := \inf_{k \in [n]} \inf_{\substack{m_j \geq 0 \\ \sum_{j=1}^n m_j \geq 2}}
\left(|
\omega_k - \sum_{i=1}^n m_j \omega_j
|\right)
\]
\end{restatable}
\autoref{thm: time_dependent_oscillator_sim} tells us that we can simulate the dynamics of time dependent oscillator systems, when the frequency parameters obeys non-resonance conditions, or specifically the frequencies should not be integer multiples of each other. The cost of simulation scales linearly with respect to time and logarithmically with respect to the system dimension. 
\\
Now, we consider a slightly modified problem statement wherein in addition to time dependent stiffness elements, external forces also act on individual masses in the oscillator system. The formal statement of the problem is given below. 

\begin{restatable}[Time-Dependent Forced Oscillator Simulation Problem]{problem}{timedependentforcedoscillator}
\label{prblm: time_dependent_forced_oscillator}
Consider the time-dependent forced coupled oscillator system described in ~\autoref{Def: time_dependent_forced_coupled_oscillator_system}. Assume we have access to elements of the mass matrix $\mathbf{M} \in \mathbb{R}^{N \times N}$, time-dependent stiffness matrix $\mathbf{G}(t) \in \mathbb{R}^{N \times N}$, and time-dependent force $\mathbf{f}(t) \in \mathbb{R}^N$ through oracles defined in~\autoref{Def:Oracles-time-dependent-oscillator-system}. The governing differential equation of the time-dependent forced coupled oscillator system is given by:
\begin{equation}
    \label{eq:time_dependent_forced_oscillator}
    \mathbf{M} \ddot{\mathbf{x}} = -\mathbf{K}(t) \mathbf{x} + \mathbf{f}(t)
\end{equation}

The Time-Dependent Forced Oscillator Simulation Problem involves preparing a quantum state $\ket{\psi{(t)}}$, using a minimal number of queries to the oracles and gate operations, that encodes the solution to Eq.~\eqref{eq:time_dependent_forced_oscillator} such that the error in the Euclidean norm between the normalized and projected quantum state and the normalized solution to Eq.~\eqref{eq:time_dependent_forced_oscillator} obeys
\begin{equation}
 \left\|  \frac{\mathbf{P}'\ket{\psi{(t)}}}{\|\mathbf{P}'\ket{\psi{(t)}}\|_2}  - \frac{[\mathbf{x},\dot{\mathbf{x}}]^T}{\|[\mathbf{x},\dot{\mathbf{x}}]^T\|_2} \right\|_2 \leq \epsilon 
\end{equation}
for some $\epsilon>0, t>0$ where $\mathbf{P}'=[\mathbf{I},\mathbf{0}]$ is a rectangular matrix that removes any fictitious or ancillary variables in the mapping. Let's also assume we have access to a state preparation oracle $\mathcal{W}$ that can efficiently prepare the initial state, $\ket{ \bm{ \psi{(0)}} } = \bm{\mathcal{W}} \ket{{0}}$.
\end{restatable}

Similar to \autoref{thm: time_dependent_oscillator_sim}, we show that the dynamics of the forced oscillator with time-dependent stiffness elements can be embedded in the dynamics of a higher dimensional nonlinear oscillator system. Since we already showed that the dynamics of nonlinear oscillator systems can be simulated using \autoref{thm:nonlinear_oscillator_simulation}, it follows that the forced oscillator problem can also be solved efficiently under appropriate conditions. Our formal result is stated below.

\begin{restatable}[Time-Dependent Forced Oscillator Simulation]{theorem}{TDFOQuantumAlgorithmThm}
\label{thm:quantum_simulation_forced}
Consider the forced oscillator system described in \autoref{Def: time_dependent_forced_coupled_oscillator_system}. Assume access to parameters of the oscillator system through oracles defined in \autoref{Def:Oracles-time-dependent-oscillator-system}. Assume that the time-dependent spring elements are of the form $k_{ij}(t) = \sum_{l=1}^L A_{ij,l}\cos(\omega_{ij,l}t + \phi_{ij,l})$, and the time-dependent external forces acting on mass $m_i$, can be written as: $f_i(t) = \sum_{j=1}^L f_{ij} \cos(\mu_{ij}t +\xi_{ij})$, and assume $L \leq N$. Let $t\in [0,T]$, $k_{\max}^{(2)} = \max_{ij,l}\{\alpha_{ij,l}, \beta_{i,j}\}$, $k_{\min}^{(1)} = \min_{ij} \{\omega_{ij},\mu_{ij}\}$, $\alpha=\max_{ij}\omega_{ij}$. 


Then, there exists a quantum algorithm that simulates the dynamics of the system utilizing at most 
 \[
   G \in \mathcal{O}\left( \alpha k^2 t + \frac{(k-1)}{2}\log\left(\left[\frac{d k^{(2)}_{\max} k\left( k +1 \right)}{2k^{(1)}_{\min} (m_{\min})^{3/2}} \left(1 + \beta\frac{1-\beta^k}{(1-\beta)\epsilon} \right)t  \right] \right) \right)
    \]
where,
\[
     k \in \widetilde{\mathcal{O}}\left(\frac{\log(CT/\epsilon)}{\log(1/R_r)}\right)
     \]
queries to unitary operators encoding \(\mathbf{K}_1\), \(\mathbf{K}_2\), and the force components. Here, \(k\) is the Carlemann truncation order, \(\alpha=\max_{ij}\omega_{ij}\). Here, $R_r := \frac{4 e \beta d k^{(2)}_{\max}}{k^{(1)}_{\min} (m_{\min})^{3/2}\Delta }$, $C \leq \frac{dk^{(2)}_{\max}\beta}{2k^{(1)}_{\min}m_{\min}^{3/2}}$, and $\Delta$ is defined as follows:
\[
\Delta := \inf_{k \in [n]} \inf_{\substack{m_j \geq 0 \\ \sum_{j=1}^n m_j \geq 2}}
\left(|
\nu_k - \sum_{i=1}^n m_j \nu_j
|\right), \nu_i \in \{\omega_{ij},\mu_{ij}\}
\]
\end{restatable}


We utilize a reduction technique in which the inhomogeneous forcing term is absorbed into an extended state space, converting the problem into an equivalent homogeneous nonlinear system. This transformation is achieved by introducing auxiliary variables that encode the time-dependent driving forces as additional dynamical degrees of freedom. These auxiliary variables evolve according to known harmonic dynamics and interact linearly with the original degrees of freedom via a coupling structure determined by the Fourier components of the forcing terms.
\\

This embedding into an enlarged homogeneous system allows us to directly apply the simulation strategy developed in \autoref{thm:nonlinear_oscillator_simulation}. Moreover, under the stated non-resonance condition on the frequencies ${\omega_{ij}, \mu_{ij}}$, the system avoids small denominator issues that would otherwise degrade the convergence rate of the approximation. The parameter $\Delta$ quantifies the minimum frequency separation, ensuring that higher-order resonances are sufficiently suppressed. As a result, the algorithm provides a provably efficient method for simulating certain time-dependent forced oscillator systems on a quantum computer, with complexity that scales linearly in the simulation time $t$, the maximum coupling strength $k^{(2)}_{\max}$, and logarithmically in the target precision $1/\epsilon$.

\section{Simulating Forced Coupled Classical Oscillators}
\label{sec:linear inhomo-diff-eq}
The problem of simulating harmonic oscillators using the method of~\cite{babbush2023exponential} is at first glance challenging.  This is because energy is not conserved for a differential equation of the form $    \mathbf{M} \ddot{\mathbf{x}} = - \mathbf{K} \mathbf{x} + \mathbf{f}(t)$.  As the normalization constant for the quantum state is a function of the, conserved, energy for the method of~\cite{babbush2023exponential}, this approach cannot directly include a time-dependent external force acting on the system.

This section gives a solution to this problem.  We address this problem by constructing a larger mass-spring system that is weakly coupled to the original system and implements a Fourier series approximation to $\mathbf{f}(t)$.  We then see that as the masses of the fictitious balls increases, the back action of the system on the forcing subsystem diminishes.  We further avoid a diverging cost as seen in the original work by re-analyzing the Harmonic oscillator algorithm as a function of $\omega = k/m$ rather than bounding it in terms of the maximum value of such a ratio.




 For clarity of exposition, we begin our analysis by considering a one-dimensional example and then extend our results to $N$ dimensional systems. 
The governing differential equation of the one dimensional forced coupled oscillator system as per \autoref{Def: forced_coupled_oscillator_system} is given by: 
\begin{align}
    \label{eq one-dimensional-inhomo-fourier}
    m_0\Ddot{x} = -k_0x + \sum_{i=1}^N f_i \cos(\omega_i t + \phi_i)
\end{align}
We can now consider the time-variant force $f(t)$ from a different perspective. Specifically, we can imagine the force $f(t)$ due to forces coming from $N$ different spring-mass systems. Or, in other words, we can construct a larger coupled oscillator system with suitable choices for spring constants, masses, and initial conditions such that the dynamics of a single mass in the larger system closely resembles the dynamics of the time dependent forced oscillator system. Mathematically, we construct a higher dimensional coupled oscillator system with dynamics $\mathbf{M} \ddot{\mathbf{x}} = \mathbf{K}\mathbf{x} $ such that $\|\mathbf{P}\mathbf{x}(t) - x(t)\| \leq \epsilon$, for any $\epsilon > 0$. Here, $\mathbf{P} \in \mathbb{R}^{1 \times N}$ is a projection matrix. In the one-dimensional example, $\mathbf{P} = \bra{\mathbf{0}}$. In the next section, we explain how the stiffness matrix $\mathbf{K}$ and the mass matrix $\mathbf{M}$ can be constructed to satisfy this condition. 
\subsection{Construction of $\mathbf{M}$ and $\mathbf{K}$}
Consider a coupled oscillator system with mass $m_0$ connected to $N$ different masses with spring constants $k_{0,i}$ respectively. The mass $m_0$ is connected to the wall with a spring constant $\bar{k_0}$ and the mass $m_i$ is connected to the wall with a spring constant $k_i$. The dynamics of the coupled spring-mass oscillator system are given by the following set of equations:
\begin{equation}
\label{eq: single_mass_dynamics_forced_1D_osc_sys}
 m_0\Ddot{x}_0 = -\left(\bar{k}_0 + \sum_{i=1}^N k_{0,i} \right)x_0  + \sum_{i=1}^N k_{0,i}x_i(t)
\end{equation}
\begin{equation}
m_i\Ddot{x}_i =  k_{0,i}x_0  -k_i x_i~\forall~i>0 
\end{equation}
Our goal is to set the parameters of the larger coupled oscillator system such that $|x(t) -x_0(t)| \leq \epsilon$, $\forall \epsilon >0, t \in [0,t_s]$. To achieve desired dynamics for $m_0$, we first equate the coefficients of $x_0$ in Eq.~\eqref{eq: single_mass_dynamics_forced_1D_osc_sys} with the stiffness coefficient in Eq.~\eqref{eq one-dimensional-inhomo-fourier}. 
\begin{equation}
\label{stiffness_constraint}
    \bar{k}_0 + \sum_{i=1}^N k_{0,i} = k_0
\end{equation}
Now, since $k_0,k_{0,i} >0$, to satisfy Eq.~\eqref{stiffness_constraint}, the stiffness values, $ \bar{k}_0, k_{0,i} < k_0$. One valid solution for Eq.~\eqref{stiffness_constraint} satisfying this constraint is $\bar{k}_0 =  \sum_{i=1}^N k_{0,i} = \frac{k_0}{2}$. We further assume $k_{0,i} = k_{0,j}$, $m_i = m_j = m_f$, $\forall i \neq j$. Hence, $k_{0,i} = \frac{k_0}{2N}$. To simulate the effect of time-dependent external force, we need the sum of forces from other coupled oscillator systems to be approximately equal to the desired time-dependent external force. More formally, we need the following condition to hold:
\begin{equation}
\label{eq:force_equality}
    \sum_{i=1}^N k_{0,i}x_i(t) \approx \sum_{i=1}^N f_i \cos(\omega_i t + \phi_i)
\end{equation}
Now, we determine the other parameters in the coupled oscillator system such that Eq.~\eqref{eq:force_equality} satisfies $\forall t \in [0,t_s]$. The governing differential equation of the coupled oscillator system in the matrix form can be written as:
\begin{equation}
    \label{eq perturbed_oscillator}
    \mathbf{M}\ddot{\mathbf{x}} = \mathbf{K}_0 \mathbf{x} + \gamma \mathbf{K'}\mathbf{x}
\end{equation}
Here, $\gamma = \frac{1}{m_f}$ and the mass matrix $\mathbf{M}$ is given by:
\begin{equation}
   \mathbf{M} = \begin{bmatrix}
                        m_0 & \mathbf{0}_{N \times N} \\
                        \mathbf{0}_{N \times N} &  \mathbf{I}_{N \times N} \\
                      \end{bmatrix}
\end{equation}
and  the stiffness matrices are:

\begin{equation}
   \mathbf{K}_0 = 
\begin{bmatrix}
-k_0 & k_{0,1} & k_{0,2} & k_{0,3} & \cdots & k_{0,N}  \\
0 & -\omega_1^2 & 0 & 0 &\cdots & 0 \\
0 & 0 & -\omega_2^2 & 0 & \cdots & 0 \\
0 & 0 & 0 & -\omega_3^2 & \cdots & 0\\
\vdots & \vdots & \vdots& . & \ddots & 0\\
0 & 0 & 0 & 0& \cdots & -\omega_N^2
\end{bmatrix}
\end{equation}

\begin{equation}
  \mathbf{K'} = 
\begin{bmatrix}
0 & 0 & 0 & 0 & \cdots & 0  \\
k_{0,1} & -k_{0,1} & 0 & 0 &\cdots & 0 \\
k_{0,2} & 0 &  -k_{0,2} & 0 & \cdots & 0 \\
k_{0,3} & 0 & 0 &  -k_{0,3}  & \cdots & 0\\
\vdots & \vdots & \vdots& . & \ddots & 0\\
k_{0,N} & 0 & 0 & 0& \cdots & - k_{0,N} 
\end{bmatrix}
\end{equation}
Here, $ \mathbf{M},\mathbf{K}_0, \mathbf{K'}  \in \mathbb{R}^{N+1 \times N+1}$ and $\|\mathbf{K'} \|_F = \frac{k_0}{\sqrt{2N}}$. Note that as $\gamma \rightarrow 0$, the dynamics of mass $m_0$ is given by:
\begin{equation}
\label{eq time_dependent_oscillator}
 m_0\Ddot{x}_0 = -k_0x_0  + \sum_{i=1}^N k_{0,i}A_i \cos (\omega_i t + \phi_i)
\end{equation}
and for other masses, the equation of motion is given by:
\begin{equation}
    x_i(t) = A_i \cos(\omega_i t + \phi_i) 
\end{equation}
Here
\begin{equation}
    \label{eq A_i}
    A_i = \pm \sqrt{x_i(0)^2 + \frac{\dot{x_i}(0)^2}{\omega_i^2}}
\end{equation}
and 
\begin{equation}
    \label{eq phi_i}
    \phi_i = - \tan^{-1}\left(\frac{\dot{x}_i(0)}{\omega_i x_i(0)}\right)
\end{equation}

Now $\forall i \in {1,..N}$, if we choose initial conditions $x_i(0),\dot{x}_i(0)$ such that $f_i = k_{0,i}A_i$ then, the governing differential equation of mass $m_0$ reduces to Eq.~\eqref{eq time_dependent_oscillator}. For achieving the desired amplitudes $A_i$ and $\phi_i$, we should set the initial conditions as: 
\begin{equation}
    x_i(0) = \frac{f_i}{k_{0,i}\sqrt{1 + \tan^2(\phi_i)}} = \frac{2N f_i }{k_0\sqrt{1 + \tan^2(\phi_i)}} =  \frac{2N f_i \cos(\phi_i) }{k_0}
\end{equation}
and 
\begin{equation}
    \dot{x_i}(0) = -\frac{f_i}{k_{0,i}} \omega_i \sin(\phi_i) = -\frac{2N f_i}{k_{0}} \omega_i \sin(\phi_i)
\end{equation}
However, when $\gamma = 0 $, the stiffness matrix of the oscillator system would not be symmetric and therefore the system would not be a valid coupled oscillator. Hence, the dynamics of the classical oscillator system cannot be reduced to Schr\"{o}dinger Evolution using the method described in \cite{babbush2023exponential}. 
To resolve this problem, we apply perturbation theory to estimate sufficient value for $m_f$ such that the dynamics of mass $m_0$ is $\epsilon$ close to the dynamics of the forced oscillator system described in Eq.~\eqref{eq one-dimensional-inhomo-fourier}. Before we introduce our method, we first present upper bounds on norms of displacement vector $ \| \mathbf{x}(t) \|$ for coupled classical oscillator systems. 
\begin{lemma}
\label{Lemma: Upper_Bound_State_Simple_Oscillator}
Let $\mathbf{M} \in \mathbb{R}^{N \times N}$ be a diagonal matrix with positive entries and $\mathbf{K} \in \mathbb{R}^{N \times N}$ be a positive semidefinite matrix. Consider the dynamics of a classical oscillator system given by:
$\mathbf{M}\ddot{\mathbf{x}} = -\mathbf{K}\mathbf{x}$
Then, for all $t \in [0,T]$, the norm of the displacement vector $\mathbf{x}(t)$ satisfies:
$\|\mathbf{x}(t)\| \leq \|\mathbf{T}\|\|\sqrt{\mathbf{M}}\|\left(\|\sqrt{\mathbf{A}}\| \|\mathbf{x}(0)\| + \|{ \dot{\mathbf{x}}}(0)\|\right)$
where $\mathbf{A} = \sqrt{\mathbf{M}^{-1}} \mathbf{K} \sqrt{\mathbf{M}^{-1}}$ and $\mathbf{T} = -i \ket{0}\bra{0} \otimes \sqrt{\mathbf{M}}^{-1}\sqrt{\mathbf{A}}^{-1} + \ket{1}\bra{1} \otimes \sqrt{\mathbf{M}}^{-1}$.
    \begin{proof}
        First we follow the same strategy as in~\cite{babbush2023exponential} and make the substitution 
        \begin{equation}
            \mathbf{y} = \sqrt{\mathbf{M}}\mathbf{x}
        \end{equation}
        and write the second order ode in the form:
        \begin{equation}
        \label{eq: second-order-ode}
            \ddot{\mathbf{y}} = -\sqrt{\mathbf{M}^{-1}} \mathbf{K} \sqrt{\mathbf{M}^{-1}} \mathbf{y}
        \end{equation}
        Let  $\mathbf{A} = \sqrt{\mathbf{M}^{-1}} \mathbf{K} \sqrt{\mathbf{M}^{-1}} $ and let $\mathbf{z}(t) = [ i\sqrt{\mathbf{A}} \mathbf{y}(t) , \dot{\mathbf{y}}(t)]^T$ and $[\mathbf{x}(t),\mathbf{\dot{x}}(t)]^T  = \mathbf{T}~\mathbf{z}(t)$, where $\mathbf{T} = -i \ket{0}\bra{0} \otimes \sqrt{\mathbf{M}}^{-1}\sqrt{\mathbf{A}}^{-1} + \ket{1}\bra{1} \otimes \sqrt{\mathbf{M}}^{-1}$.
          Now, Eq.~\eqref{eq: second-order-ode} can be written in the first-order form as:
          
        \begin{equation}
        \label{eq: first-order-ode}
           \frac{d}{dt} \mathbf{z}(t) = i\hat{\mathbf{A}}\mathbf{z}(t)
        \end{equation}
          where
           \begin{equation}
            \hat{\mathbf{A}} = \begin{bmatrix}
                \mathbf{0} & \sqrt{\mathbf{A}} \\
               \sqrt{\mathbf{A}} & \mathbf{0}
            \end{bmatrix}
        \end{equation}

       The solution to Eq.~\eqref{eq: first-order-ode} can be written as: 
        \begin{equation}
            \mathbf{z}(t) = e^{i\hat{\mathbf{A}}t} \mathbf{z}(0)
        \end{equation}
        Note that $e^{i\hat{\mathbf{A}}t}$ is a unitary matrix, since $\hat{\mathbf{A}}$ is a Hermitian matrix. 
        Hence, norm of $\mathbf{z}(t)$ is bounded by:
        \begin{equation}
            \|\mathbf{z}(t)\| \leq  \|\mathbf{z}(0) \|
        \end{equation}
        and 
        \begin{align}
            \| \mathbf{x}(t) \| &\leq \|\mathbf{T}\| \|\mathbf{z}(t)\| \\ \nonumber 
            &\leq\|\mathbf{T}\| ~ \sqrt{\|\sqrt{\mathbf{A}}\mathbf{y}(0)\|^2 + \|\dot{\mathbf{y}}(0)\|^2} \\ \nonumber 
              &\leq\|\mathbf{T}\| ~ \sqrt{\|\sqrt{\mathbf{A}} \sqrt{\mathbf{M}}\mathbf{x}(0)\|^2 + \|\sqrt{\mathbf{M}}\dot{\mathbf{x}}(0)\|^2} \\ \nonumber     
              &\leq \|\mathbf{T}\| ~ \| \sqrt{\mathbf{M}}\| \sqrt{\|\sqrt{\mathbf{A}}\|^2 \| \mathbf{x}(0)\|^2 + \|\dot{\mathbf{x}}(0)\|^2} \\ \nonumber
               &\leq \|\mathbf{T}\|~ \| \sqrt{\mathbf{M}}\| \left( \|\sqrt{\mathbf{A}}\| \| \mathbf{x}(0)\| + \|\dot{\mathbf{x}}(0)\|\right)
        \end{align}
        
        \end{proof}
    
\end{lemma}

\begin{lemma}
\label{Lemma Upper_Bound_State_Norm_Inhomo}
   Let $\mathbf{M} \in \mathbb{R}^{N \times N}$ be a diagonal matrix with positive entries and $\mathbf{K} \in \mathbb{R}^{N \times N}$ be a positive semidefinite matrix. Consider the dynamics of a classical oscillator system given by:
$\mathbf{M}\ddot{\mathbf{x}} = -\mathbf{K}\mathbf{x} + \mathbf{f}(t)$.
Then, for all $t \in [0,T]$, the norm of the displacement vector $\mathbf{x}(t)$ satisfies:
$ \|\mathbf{x}(t)\| \leq  \| \mathbf{T} \| \left( \| \sqrt{\mathbf{A}} \| \|\sqrt{\mathbf{M}} \| \| \mathbf{x}(0)\| + \|\sqrt{\mathbf{M}}\| \| \dot{\mathbf{x}}(0)\| + t\max_{t \in [0,T]} \| \mathbf{\hat{F}}(t)\| \right)$
where $\mathbf{A} = \sqrt{\mathbf{M}^{-1}} \mathbf{K} \sqrt{\mathbf{M}^{-1}}$ and $\mathbf{T} = -i \ket{0}\bra{0} \otimes \sqrt{\mathbf{M}}^{-1}\sqrt{\mathbf{A}}^{-1} + \ket{1}\bra{1} \otimes \sqrt{\mathbf{M}}^{-1}$. 
    \begin{proof}
        First, we rewrite the dynamical equation in terms of $ \mathbf{y}= \sqrt{\mathbf{M}}\mathbf{x}$.
        \begin{equation}
            \ddot{\mathbf{y}} = \mathbf{A}\mathbf{y} + \sqrt{\mathbf{M}^{-1}}\mathbf{f}(t)
        \end{equation}
        Here $\mathbf{A} = \sqrt{\mathbf{M}^{-1}}\mathbf{K}\sqrt{\mathbf{M}^{-1}}$. Now we write the second order ode in first order form. Let $\mathbf{z} =[i\sqrt{\mathbf{A}}\mathbf{y},\dot{\mathbf{y}}]^T$ and $[\mathbf{x}(t),\mathbf{\dot{x}}(t)]^T  = \mathbf{T}~\mathbf{z}(t)$, where $\mathbf{T} = -i \ket{0}\bra{0} \otimes \sqrt{\mathbf{M}}^{-1}\sqrt{\mathbf{A}}^{-1} + \ket{1}\bra{1} \otimes \sqrt{\mathbf{M}}^{-1}$. After applying the coordinate transformations, the dynamical equations can be written as:
        \begin{equation}
            \label{eq first_order_inhomo}
            \dot{\mathbf{z}} = i\mathbf{\hat{A}} \mathbf{z} + \hat{\mathbf{F}}(t) 
        \end{equation}
        Here, 
        \begin{equation}
            \mathbf{\hat{A}} = \begin{bmatrix}
                        \mathbf{0} & \mathbf{\sqrt{A}} \\
                        \mathbf{\sqrt{A}} &  \mathbf{0} \\
                      \end{bmatrix}
        \end{equation}
        and 
        \begin{equation}
            \mathbf{\hat{F}}(t) = \begin{bmatrix}
                        \mathbf{0}\\
                        \sqrt{\mathbf{M}^{-1}}\mathbf{F}(t) \\
                      \end{bmatrix}
        \end{equation}
        The solution to Eq\eqref{eq first_order_inhomo} is given by:
        \begin{equation}
            \mathbf{z}(t) = e^{i\hat{\mathbf{A}}t} \mathbf{z}(0) + e^{i\hat{\mathbf{A}}t}\int_{0}^{t}e^{-i\hat{\mathbf{A}}t} \mathbf{\hat{F}}(t) dt
        \end{equation}
        and
        \begin{equation}
            [\mathbf{x}(t),  \mathbf{\dot{x}}(t)]^T = \mathbf{T} \left(e^{i\hat{\mathbf{A}}t} \mathbf{z}(0) + e^{i\hat{\mathbf{A}}t}\int_{0}^{t}e^{-i\hat{\mathbf{A}}t} \mathbf{\hat{F}}(t) dt \right)
        \end{equation}
        
        \begin{equation}
            \|\mathbf{x}(t)\| \leq \| \mathbf{T} \| \left( \|e^{i\hat{\mathbf{A}}t} \| \|\mathbf{z}(0)\|  + \|e^{i\hat{\mathbf{A}}t} \| \int_{0}^{t} \|e^{-i\hat{\mathbf{A}}t}\| \|\mathbf{\hat{F}}(t)\| dt  \right)
        \end{equation}
        Since $\hat{\mathbf{A}}$ is a Hermitian matrix, $e^{i\hat{\mathbf{A}}t}$ is unitary and $\|e^{i\hat{\mathbf{A}}t} \| = \|e^{-i\hat{\mathbf{A}}t} \| = 1$.\\
        Hence, the norm $\|\mathbf{x}(t)\|$ is bounded by:
        \begin{align}
            \|\mathbf{x}(t)\| &\leq  \| \mathbf{T} \| \left( \| \sqrt{\mathbf{A}} \| \|\sqrt{\mathbf{M}} \| \| \mathbf{x}(0)\| + \|\sqrt{\mathbf{M}}\| \| \dot{\mathbf{x}}(0)\| + t\max_{t \in [0,T]} \| \mathbf{\hat{F}}(t)\| \right)  \\ \nonumber    
            &\leq  \| \mathbf{T} \| \left( \| \sqrt{\mathbf{A}} \| \|\sqrt{\mathbf{M}} \| \| \mathbf{x}(0)\| + \|\sqrt{\mathbf{M}}\| \| \dot{\mathbf{x}}(0)\| + t \|\sqrt{\mathbf{M}}^{-1}\|\max_{t \in [0,T]} \| \mathbf{F}(t)\| \right)  \\ \nonumber    
        \end{align}
\end{proof}
\end{lemma}
We are now ready to apply these ideas to provide a proof of our main perturbation result which is stated as Theorem~\ref{thm:pertThm}. 
\begin{theorem}[Perturbation Result]
\label{thm:pertThm}

Let $\mathbf{M} \ddot{\mathbf{x}}_0 = \mathbf{K}_0 \mathbf{x}_0$, and  $\mathbf{M} \ddot{\mathbf{x}} = (\mathbf{K}_0 + \gamma \mathbf{K}') \mathbf{x}$ be governing differential equations of two coupled oscillator systems. Let $\mathbf{A} = \sqrt{\mathbf{M}^{-1}} \mathbf{K} \sqrt{\mathbf{M}^{-1}}$ , $\mathbf{T} = -i \ket{0}\bra{0} \otimes \sqrt{\mathbf{M}}^{-1}\sqrt{\mathbf{A}}^{-1} + \ket{1}\bra{1} \otimes \sqrt{\mathbf{M}}^{-1}$ and $ \mathbf{\Xi}(0) =  \left( \| \sqrt{\mathbf{A}} \| \|\sqrt{\mathbf{M}} \| \| \mathbf{x}(0)\| + \|\sqrt{\mathbf{M}}\| \| \dot{\mathbf{x}}(0)\| \right)$. Then, given any  $\epsilon >0$, for any $ \frac{1}{\gamma} \geq   \left(t \|\mathbf{K}'\| \| \mathbf{T} \| \|\sqrt{\mathbf{M}}^{-1}\|\right)  \left(1 + \frac{\mathbf{\Xi}(0)}{\epsilon}\right)$ 
,  $\|\mathbf{x}(t) -  \mathbf{x}_0(t)\| \leq \epsilon$. 
\end{theorem}
  \begin{proof}[Proof of Theorem~\ref{thm:pertThm}]
      Consider Eq \eqref{eq perturbed_dynamics}:
      \begin{equation}
      \label{eq perturbed_dynamics}
          \mathbf{M}\ddot{\mathbf{x}} = (\mathbf{K}_0 + \gamma \mathbf{K}')\mathbf{x}
      \end{equation}
      First, we apply perturbation theory to express $\mathbf{x}(t)$ at any time $t$ as a power series in $\gamma$ and $\mathbf{x}_i(t)$.
      \begin{equation}
      \label{eq perturbed_expansion_displacement}
          \mathbf{x}(t) = \mathbf{x}_0(t) + \gamma \mathbf{x}_1(t) + \gamma^2 \mathbf{x}_2(t) +...
      \end{equation}
      We further assume,
      \begin{equation}
          \mathbf{x}(t=0) = \mathbf{x}_0(t=0)
      \end{equation}
      and 
      \begin{equation}
          \label{eq initial_conditions}
          \mathbf{x}_1(t=0) = \mathbf{x}_2(t=0) = \mathbf{x}_{k \neq 0}(t=0) = \mathbf{0}
      \end{equation}
Substituting Eq \eqref{eq perturbed_expansion_displacement} in Eq \eqref{eq perturbed_dynamics} yields:
      \begin{equation}
    \mathbf{M} (\ddot{\mathbf{x}}_0 + \gamma \ddot{\mathbf{x}}_1 + \gamma^2 \ddot{\mathbf{x}}_2 +...) = (\mathbf{K}_0 + \gamma \mathbf{K}') (\mathbf{x}_0 + \gamma \mathbf{x}_1 + \gamma^2 \mathbf{x}_2 +..)
\end{equation}
This leads to a system of coupled differential equations for $\mathbf{x}_i(t)$.
The governing equation of $\mathbf{x}_0(t)$ is given by:
      \begin{equation}
          \mathbf{M} \ddot{\mathbf{x}}_0 = \mathbf{K}_0 \mathbf{x}_0
      \end{equation}
Note that this is our desired ode. 
For $\mathbf{x}_1(t)$, we get: 
\begin{align}
    \mathbf{M}\ddot{\mathbf{x}}_1 = \mathbf{K}_0 \mathbf{x}_1 + \mathbf{K}' \mathbf{x}_0 
\end{align}
\begin{align}
    \mathbf{M}\ddot{\mathbf{x}}_2 = \mathbf{K}_0 \mathbf{x}_2 + \mathbf{K}' \mathbf{x}_1 
\end{align}
Or in general,
\begin{align}
    \mathbf{M} \ddot{\mathbf{x}}_i = \mathbf{K}_0 \mathbf{x}_i + \mathbf{K}' \mathbf{x}_{i-1}
\end{align}
 The difference between $\mathbf{x}_0(t)$ and $\mathbf{x}(t)$ is bounded by:
      \begin{equation}
          \|\mathbf{x}(t) - \mathbf{x}_0(t)\| \leq \gamma\|\mathbf{x}_1(t)\|   + \gamma^2 \|\mathbf{x}_2(t)\| + \cdots
      \end{equation}
      To bound the simulation error $ \|\mathbf{x}(t) - \mathbf{x}_0(t)\| \leq \epsilon$, we bound the sum:
\begin{equation}
\label{eq error_term}
\gamma\|\mathbf{x}_1(t)\|   + \gamma^2 \|\mathbf{x}_2(t)\| + .. \leq \epsilon
\end{equation}
Now we have to bound $\|\mathbf{x}_i(t)\|$ $\forall i \in \mathbb{N}$. Note that the governing equations for $\mathbf{x}_i(t)$ are of the form: $\mathbf{M}\ddot{\mathbf{x}_i} = \mathbf{K}\mathbf{x}_i + \mathbf{F}(t)$, where $\mathbf{F}(t) = \mathbf{K}'\mathbf{x}_{i-1}(t)$. First, we start with $i=1$, and bound $\|\mathbf{x}_1(t)\|$. To bound $\|\mathbf{x}_1(t)\|$, we apply results from \hyperref[Lemma Upper_Bound_State_Norm_Inhomo]{Lemma 3}. Here, $\mathbf{F}(t) = \mathbf{K}'\mathbf{x}_0(t)$.
\begin{equation}
    \label{eq upper_bound_norm_x_1}
     \| \mathbf{x}_1(t) \| \leq \| \mathbf{T} \| \left( \| \sqrt{\mathbf{A}} \| \|\sqrt{\mathbf{M}} \| \| \mathbf{x}(0)\| + \|\sqrt{\mathbf{M}}\| \| \dot{\mathbf{x}}(0)\| + t \|\sqrt{\mathbf{M}}^{-1}\|\max_{t \in [0,T]} \| \mathbf{K}'\mathbf{x}_0(t)\| \right) 
\end{equation}
Since $\mathbf{x}_1(0) = \dot{\mathbf{x}}_1(0) = \mathbf{0}$ due to Eq \eqref{eq initial_conditions}, Eq \eqref{eq upper_bound_norm_x_1} simplifies to:
\begin{equation}
\label{eq upper_bound_norm_x_1_simplified}
     \| \mathbf{x}_1(t) \| \leq  \| \mathbf{T} \| \left( t \|\sqrt{\mathbf{M}}^{-1}\|\max_{t \in [0,T]} \| \mathbf{K}'\mathbf{x}_0(t)\| \right)
\end{equation}
And in general, we have:
\begin{equation}
\label{eq upper_bound_norm_x_i_simplified}
     \| \mathbf{x}_i(t) \| \leq  \| \mathbf{T} \| \left( t \|\sqrt{\mathbf{M}}^{-1}\|\max_{t \in [0,T]} \| \mathbf{K}'\mathbf{x}_{i-1}(t)\| \right)
\end{equation}
From \hyperref[Lemma: Upper_Bound_State_Simple_Oscillator]{Lemma 1}, we have: 
\begin{equation}
    \label{eq force_upper_bound_1}
    \|\mathbf{K}' \mathbf{x}_0(t)\| \leq \|\mathbf{K}'\|. \|\mathbf{x}_0 (t) \|  \leq   \|\mathbf{K}'\| \| \mathbf{T} \| \left( \| \sqrt{\mathbf{A}} \| \|\sqrt{\mathbf{M}} \| \| \mathbf{x}(0)\| + \|\sqrt{\mathbf{M}}\| \| \dot{\mathbf{x}}(0)\| \right) 
\end{equation}
Substituting Eq \eqref{eq force_upper_bound_1} in Eq \eqref{eq upper_bound_norm_x_1_simplified} and setting $ \mathbf{\Xi}(0) =  \left( \| \sqrt{\mathbf{A}} \| \|\sqrt{\mathbf{M}} \| \| \mathbf{x}(0)\| + \|\sqrt{\mathbf{M}}\| \| \dot{\mathbf{x}}(0)\| \right) $, we get:

\begin{equation}
\label{eq upper_bound_x_1}
     \| \mathbf{x}_1(t) \| \leq   t\|\mathbf{K}'\| \| \mathbf{T} \|^2\cdot \|\sqrt{\mathbf{M}}^{-1}\| \cdot \mathbf{\Xi}(0)
\end{equation}
Similarly, we upper bound  $\|\mathbf{x}_2(t)\|$. Since, $\mathbf{x}_2(0) = \dot{\mathbf{x}}_2(0) = \mathbf{0}$, 
\begin{align}
     \| \mathbf{x}_2(t) \| &\leq \| \mathbf{T} \| \left( t \|\sqrt{\mathbf{M}}^{-1}\|\max_{t \in [0,T]} \| \mathbf{K}'\mathbf{x}_1(t)\| \right) \nonumber\\
      &\leq t^2 \|\mathbf{K}'\|^2 \| \mathbf{T} \|^3\cdot \|\sqrt{\mathbf{M}}^{-1}\|^2 \cdot \mathbf{\Xi}(0)
      \label{eq:X2Bd}
\end{align}
Now, we apply proof by induction to bound $ \| \mathbf{x}_i(t) \| $. Specifically, let's assume that $\| \mathbf{x}_i(t) \| \leq t^i \|\mathbf{K}'\|^i \| \mathbf{T} \|^{i+1} \|\sqrt{\mathbf{M}}^{-1}\|^i \mathbf{\Xi}(0)$. Now, $\| \mathbf{x}_{i+1}(t)\|$ is upper bounded by:
\begin{align}
    \| \mathbf{x}_{i+1}(t) \| &\leq  \| \mathbf{T} \| \left( t \|\sqrt{\mathbf{M}}^{-1}\|\max_{t \in [0,T]} \| \mathbf{K}'\mathbf{x}_{i}(t)\| \right)\\ \nonumber 
    &\leq  \| \mathbf{T} \| \left( t \|\sqrt{\mathbf{M}}^{-1}\| \| \mathbf{K}' \| \max_{t \in [0,T]} \| \mathbf{x}_{i}(t)\| \right) \\ \nonumber 
    &\leq t^{i+1} \|\mathbf{K}'\|^{i+1} \| \mathbf{T} \|^{i+2} \|\sqrt{\mathbf{M}}^{-1}\|^{i+1} \mathbf{\Xi}(0)
\end{align}
This upper bound satisfies the induction hypothesis, and hence our assumption $\| \mathbf{x}_i(t) \| \leq t^i \|\mathbf{K}'\|^i \| \mathbf{T} \|^{i+1} \|\sqrt{\mathbf{M}}^{-1}\|^i \mathbf{\Xi}(0)$ is correct $\forall i \in \mathbb{N}$. 
Now, we upper bound the infinite sum in Eq \eqref{eq error_term}.
\begin{equation}
    \sum_{i=1}^{\infty} \gamma^i \|\mathbf{x}_i(t)\| \leq \sum_{i=1}^{\infty} \gamma^i t^i \|\mathbf{K}'\|^i \| \mathbf{T} \|^{i+1} \|\sqrt{\mathbf{M}}^{-1}\|^i \mathbf{\Xi}(0)
\end{equation}
The infinite sum converges when 
\begin{equation}
    \gamma \ll \frac{1}{ t\|\mathbf{K}'\|\| \mathbf{T} \| \|\sqrt{\mathbf{M}}^{-1}\|}
\end{equation}
And the converged value is given by: 
\begin{equation}
    \mathbf{\Xi}(0) \frac{\gamma t \|\mathbf{K}'\| \| \mathbf{T} \| \|\sqrt{\mathbf{M}}^{-1}\| }{1 - \gamma t \|\mathbf{K}'\| \| \mathbf{T} \| \|\sqrt{\mathbf{M}}^{-1}\|} 
\end{equation}
This sum has to be bounded by $\epsilon$, to achieve desired simulation error. Hence,
\begin{equation}
    \mathbf{\Xi}(0) \frac{\gamma t \|\mathbf{K}'\| \| \mathbf{T} \| \|\sqrt{\mathbf{M}}^{-1}\| }{1 - \gamma t \|\mathbf{K}'\| \| \mathbf{T} \| \|\sqrt{\mathbf{M}}^{-1}\|}  \leq \epsilon 
\end{equation}
After algebraic manipulations, we get:
\begin{equation}
    \frac{1}{\gamma} \geq   \left(t \|\mathbf{K}'\| \| \mathbf{T} \| \|\sqrt{\mathbf{M}}^{-1}\|\right)  \left(1 + \frac{\mathbf{\Xi}(0)}{\epsilon}\right)
\end{equation}
\end{proof}

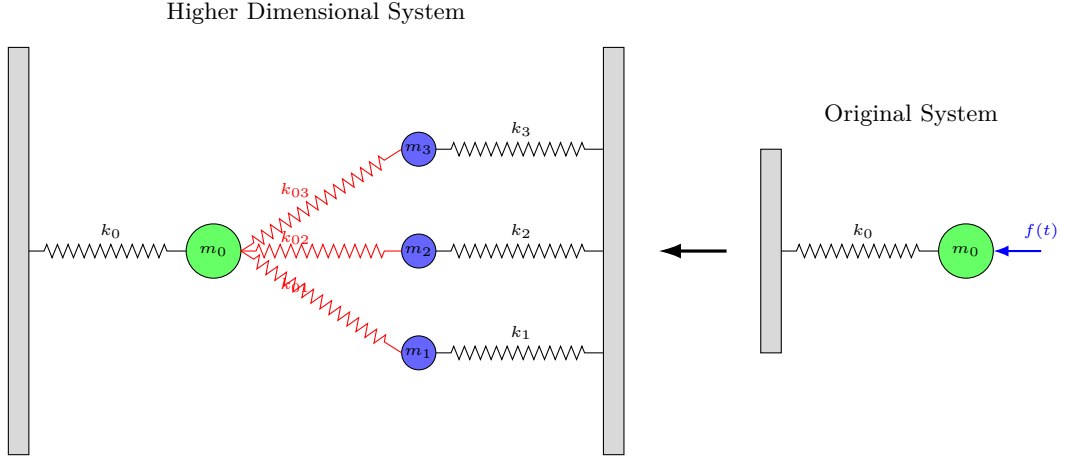
\begin{figure}
\centering
\begin{tikzpicture}[scale=0.9]
    \tikzstyle{spring}=[decorate,decoration={zigzag,pre length=0.2cm,post length=0.2cm,segment length=4}]

    \begin{scope}[xshift=11cm, yshift=1.5cm]
        \draw[fill=gray!30] (0,0) rectangle (0.3,3);
        \node[rotate=90, font=\tiny] at (0.15,1.5) {};
        
        \filldraw[fill=green!60] (3,1.5) circle (0.4);
        \node[font=\tiny] at (3,1.5) {$m_0$};
        
        \draw[spring] (0.3,1.5) -- (2.6,1.5);
        \node[font=\tiny] at (1.5,1.8) {$k_0$};
        
        \draw[-{Latex[length=2mm, width=1.5mm]}, thick, blue] (4.1,1.5) -- (3.4,1.5);
        \node[blue, font=\tiny] at (4.1,1.8) {$f(t)$};
        
        \node[font=\footnotesize] at (2.2,3.5) {Original System};
    \end{scope}

    \begin{scope}
        \draw[fill=gray!30] (0,0) rectangle (0.3,6);
        \node[rotate=90, font=\tiny] at (0.15,3) {};
        
        \draw[fill=gray!30] (8.7,0) rectangle (9,6);
        \node[rotate=90, font=\tiny] at (8.85,3) {};
        
        \filldraw[fill=green!60] (3,3) circle (0.4);
        \node[font=\tiny] at (3,3) {$m_0$};
        
        \draw[spring] (0.3,3) -- (2.6,3);
        \node[font=\tiny] at (1.5,3.3) {$k_0$};
        
        \filldraw[fill=blue!60] (6,1.5) circle (0.25);
        \node[font=\tiny] at (6,1.5) {$m_1$};
        \filldraw[fill=blue!60] (6,3) circle (0.25);
        \node[font=\tiny] at (6,3) {$m_2$};
        \filldraw[fill=blue!60] (6,4.5) circle (0.25);
        \node[font=\tiny] at (6,4.5) {$m_3$};
        
        \draw[red,spring] (3.4,3) -- (5.75,1.5);
        \node[red, font=\tiny] at (4.2,2.5) {$k_{01}$};
        \draw[red,spring] (3.4,3) -- (5.75,3);
        \node[red, font=\tiny] at (4.2,3.2) {$k_{02}$};
        \draw[red,spring] (3.4,3) -- (5.75,4.5);
        \node[red, font=\tiny] at (4.2,3.9) {$k_{03}$};
        
        \draw[spring] (6.25,1.5) -- (8.7,1.5);
        \node[font=\tiny] at (7.5,1.8) {$k_1$};
        \draw[spring] (6.25,3) -- (8.7,3);
        \node[font=\tiny] at (7.5,3.3) {$k_2$};
        \draw[spring] (6.25,4.5) -- (8.7,4.5);
        \node[font=\tiny] at (7.5,4.8) {$k_3$};
        
        \node[font=\footnotesize] at (4.5,6.5) {Higher Dimensional System};
    \end{scope}
    \draw[-{Latex[length=3mm, width=2mm]}, very thick] (10.5,3) -- (9.5,3);
    \node[font=\tiny] at (10,3.3) {};

    \node[font=\footnotesize, text width=14cm, align=center] at (6,-0.5) 
        {};
\end{tikzpicture}
\caption{Equivalence between a higher dimensional oscillator system and a simple oscillator with external force}
\label{fig:equivalent_oscillators}
\end{figure}

\begin{theorem}[One Dimensional Inhomogeneous system]
\label{thm:one_dimensional_oscillator}
  Given a dynamical system of the form $m_0\Ddot{x} = -k_0x + f(t)$, where $f(t) = \sum_{i=1}^N f_i cos(\omega_i t + \phi_i) $, there exists a higher dimensional coupled oscillator system with dynamics: $\mathbf{M}\ddot{\mathbf{x}} = \mathbf{K}\mathbf{x}$ and initial conditions $\mathbf{x}(0) = \mathbf{x}_0$ and $\dot{\mathbf{x}}(0) = \dot{\mathbf{x}}_0$ such that $\forall t \in [0,T]$, for any $\epsilon >0$,  $\|\mathbf{P}\mathbf{x}(t) -  x_0(t)\| \leq \epsilon$. 
\end{theorem}

\begin{proof}

Let $\mathbf{M} \ddot{\mathbf{x}} = \mathbf{K}\mathbf{x}$ be the governing differential equation of higher dimensional coupled oscillator system. \\
Here,
\begin{equation}
   \mathbf{M} = \begin{bmatrix}
                        m_0 & \mathbf{0}_{N \times N} \\
                        \mathbf{0}_{N \times N} &  \mathbf{I}_{N \times N} \\
                      \end{bmatrix}
\end{equation}
and let $\mathbf{K} = \mathbf{K}_0 + \frac{1}{m_f}\mathbf{K}'$

\begin{equation}
   \mathbf{K}_0 = 
\begin{bmatrix}
-k_0 & k_{0,1} & k_{0,2} & k_{0,3} & \cdots & k_{0,N}  \\
0 & -\omega_1^2  & 0 & 0 &\cdots & 0 \\
0 & 0 & -\omega_2^2  & 0 & \cdots & 0 \\
0& 0 & 0 & -\omega_3^2 & \cdots & 0\\
\vdots & \vdots & \vdots& . & \ddots & 0\\
0 & 0 & 0 & 0& \cdots & -\omega_N^2 
\end{bmatrix}
\end{equation}    
and 
\begin{equation}
   \frac{1}{m_f}\mathbf{K}' = 
\frac{1}{m_f}\begin{bmatrix}
0 & 0 & 0 & 0 & \cdots & 0  \\
k_{0,1} & - k_{0,1} & 0 & 0 &\cdots & 0 \\
k_{0,2} & 0 & - k_{0,2} & 0 & \cdots & 0 \\
k_{0,3} & 0 & 0 & - k_{0,3}& \cdots & 0\\
\vdots & \vdots & \vdots& . & \ddots & 0\\
 k_{0,N} & 0 & 0 & 0& \cdots & - k_{0,N}
\end{bmatrix}
\end{equation}  
The governing differential equation of higher dimensional oscillator is of the form: $\mathbf{M} \ddot{\mathbf{x}} = (\mathbf{K}_0 + \gamma \mathbf{K}')\mathbf{x}$, where $\gamma = \frac{1}{m_f}$.
Applying \autoref{thm:pertThm}, for any 
\begin{equation}
   \frac{1}{\gamma} = m_f \geq \left( t \|\mathbf{K}'\| \| \mathbf{T} \| \|\sqrt{\mathbf{M}}^{-1}\|\right)  \left(1 + \frac{\mathbf{\Xi}(0)}{\epsilon}\right)
\end{equation},
$\|\mathbf{x}(t) -\mathbf{x}_0(t) \| \leq \epsilon$, where the dynamics of $\mathbf{x}_0(t)$ is described by the governing equation: $\mathbf{M}\ddot{\mathbf{x}}_0 = \mathbf{K}_0 \mathbf{x}_0$.

\end{proof}

\subsection{N-Dimensional Case}
In this section, we consider the general $N$ dimensional forced coupled oscillator system defined in \autoref{Def: forced_coupled_oscillator_system}. Mathematically, the dynamics of the system can be expressed as:

\begin{equation}
\label{eq-inhomo}
 \mathbf{M}\ddot{\mathbf{x}} =  -\mathbf{K} \mathbf{x} + \mathbf{f}(t)
\end{equation}
Given initial conditions $\mathbf{x}(t=0) = \mathbf{x}_0 \in \mathbb{R}^N, \dot{\mathbf{x}}(t=0) = \dot{\mathbf{x}}_0 \in \mathbb{R}^N$ and $\mathbf{f}(t) \in \mathbb{R}^N$.  The Forced Quantum Oscillator problem is then to find a sequence of gate and query operations of minimal length  that prepares a quantum state within distance $\epsilon$ of  $\bm{\ket{\psi(t)}} \propto [\mathbf{x}(t),\dot{\mathbf{x}}(t)]^T$.
\\

To solve the problem, we show that the dynamics of the non-conservative classical system can be contained in the dynamics of a higher-dimensional conservative dynamical system. Mathematically, we construct a higher dimensional coupled oscillator system with dynamics $\mathbf{M}'\ddot{\mathbf{z}} = \mathbf{K}'\mathbf{z}$, such that a subspace of $\mathbf{z}(t)$ evolves according to Eq.~\eqref{eq-inhomo}. We show the explicit construction of $\mathbf{K}'$ and $\mathbf{M}'$ in the Theorem below. 


\begin{restatable}[Construction of larger Coupled Oscillator System]{theorem}{consthm}
\label{thm:pertThm N Dimensions}
 Consider a forced coupled oscillator system defined in \autoref{Def: forced_coupled_oscillator_system}. Let the governing differential equation of the oscillator system be:  
  \[
  \mathbf{M}\ddot{\mathbf{x}} = -\mathbf{K}\mathbf{x} + \mathbf{f}(t)
  \] 
  Then, there exists a higher dimensional coupled oscillator system with dynamics: 
  \[
  \mathbf{M}'\ddot{\mathbf{z}} = -\mathbf{K}'\mathbf{z}
  \]
 and initial conditions $\mathbf{z(0)} , \dot{\mathbf{z}}(0)$ such that $\forall t \in [0,t_s]$, for any $\epsilon >0$,  $\|\mathbf{P}\mathbf{z}(t) -  \mathbf{x}(t)\| \leq \epsilon$. Here, $\mathbf{z} \in \mathbb{R}^{N(l+1)}$,  $\mathbf{M}' \in \mathbb{R}^{N(l+1) \times N(l+1)}$,  and $\mathbf{K}' \in \mathbb{R}^{N(l+1)\times N(l+1)}$ are the mass, and stiffness matrices and the rectangular matrix $\mathbf{P} = [\mathbf{I}_{N\times N}, \mathbf{0}_{N \times Nl}] \in \mathbb{R}^{N \times N(l+1)}$. The initial displacement vector of the higher dimensional oscillator system is given by:
 \[
 \mathbf{z}(0) = [\mathbf{x}(0), \mathbf{y}_1(0), \mathbf{y}_2(0) \cdots \mathbf{y}_N(0) ]^T,
 \]
 where 
 \[
 \mathbf{y}_i(0) = [y_{i,1}(0), \cdots y_{i,l}(0)]^T,
 \]
 and $ y_{i,j}(0) =  \frac{2 l f_{ij}}{k_{ii}\sqrt{1+\tan^2 (\phi_{ij}) }} $. The initial velocity vector, of the high dimensional oscillator system is given by: $\dot{\mathbf{z}}(0) = [\dot{\mathbf{x}}(0), \dot{\mathbf{y}}_1(0), \dot{\mathbf{y}}_2(0) \cdots \dot{\mathbf{y}}_N(0) ]^T $, where $\dot{\mathbf{y}}_i(0) = [\dot{y}_{i,1}(0), \cdots \dot{y}_{i,l}(0)]^T$ , and $ \dot{y}_{i,j}(0) =  -{ \frac{2l f_{ij}}{k_{ii}}} \omega_{ij}\sin(\phi_{ij})$.



 \end{restatable}
\begin{proof}
Let $\mathbf{G} \in \mathbb{R}^{N \times N}$ be the stiffness matrix of the original system governed by the differential equation: $\mathbf{M}\ddot{\mathbf{x}} = -\mathbf{K} \mathbf{x} + \mathbf{f}(t)$. The elements of $\mathbf{G}(i,j)$ are values of spring constants $k_{ij}$ of stiffness elements that connects $i$ th and $j$ th masses. \\

We construct a higher-dimensional system by introducing $N l$ auxiliary masses of mass $m_f$ and coupling them with masses in the original system. The auxillary masses are indexed from $N+1$ to $N(1 + l)$. The auxillary masses indexed from $N+1$ to $N+l$ are connected to mass $m_1$, with springs of stiffness constant $\frac{k_{11}}{2l}$. These auxillary masses are connected to the wall with springs of stiffness constant $m_f \omega_{1i}^2$. We continue this procedure, and connect mass $m_2$ with the next set of $l$ auxillary masses, and so on and so forth. In addition to coupling masses $m_j$ with auxillary masses, we set the stiffness constant of spring connecting the mass to the wall as $\frac{k_{jj}}{2}$. The stiffness matrix for the higher dimensional system can be written as:\\

\begin{align}
 \mathbf{G}' = \begin{bmatrix}
        {\mathbf{G}} & \mathbf{G}_{12} & \mathbf{G}_{13} & \cdots & \mathbf{G}_{1N}\\
        \mathbf{G}_{12}^T & \mathbf{G}_{22} & \mathbf{0} & \cdots & \mathbf{0}\\
        \mathbf{G}_{13}^T & \mathbf{0} & \mathbf{G}_{33} & \ddots & \mathbf{0} \\
        \vdots & \vdots & \ddots & \ddots & \vdots \\
        \mathbf{G}_{1N}^T & \mathbf{0} & \cdots & \cdots & \mathbf{G}_{NN}
 \end{bmatrix} 
\end{align}
The submatrices are:
\begin{align}
{\mathbf{G}} =   \begin{bmatrix}
                    \frac{k_{11}}{2} & k_{12} & k_{13}&k_{14} & \cdots & k_{1N} \\
                    k_{21} & \frac{k_{22}}{2} & k_{23}&k_{24} & \cdots & k_{2N}\\
                    k_{31} & k_{32} & \frac{k_{33}}{2}&k_{34}  & \cdots & k_{2N}\\
                    k_{41} &  k_{42} & k_{43} & \frac{k_{44}}{2}&\cdots & k_{4N}\\
                    \vdots & \vdots & \vdots & \vdots &\ddots & \vdots \\
                    k_{N1}  & k_{N2} & k_{N3}& k_{N4}& \cdots & \frac{k_{NN}}{2}\\                 
                  \end{bmatrix}
\end{align}

and 
\begin{align}
    \mathbf{G}_{12} = \begin{bmatrix}
              \frac{k_{11}}{2l} & \frac{k_{11}}{2l} & \cdots  \frac{k_{11}}{2l}\\
              {0} &  {0}  & \cdots {0}\\
              \vdots &\vdots &\vdots \\ 
              {0} &  {0}  & \cdots {0}
             \end{bmatrix}
\end{align}
and 
\begin{align}
    \mathbf{G}_{13} = \begin{bmatrix}              
            {0} &  {0}  & \cdots {0}\\
              \frac{k_{22}}{2l} & \frac{k_{22}}{2l} & \cdots  \frac{k_{22}}{2l}\\
              \vdots &\vdots &\vdots \\ 
              {0} &  {0}  & \cdots {0}
             \end{bmatrix}
\end{align}
and in general, 

\begin{align}
    \mathbf{G}_{1j} = \frac{k_{j-1,j-1}}{2l}\begin{bmatrix}              
              \ket{j-2} &  \ket{j-2}  & \cdots \ket{j-2}\\            
             \end{bmatrix}
\end{align}
and the diagonal block matrices, $\mathbf{G}_{ii} \in \mathbb{R}^{l \times l}$ are:
\begin{equation}
    \mathbf{G}_{ii} = \begin{bmatrix}
m_f \omega_{i1}^2 & 0 & 0 & \cdots & 0 \\
0 & m_f \omega_{i2}^2 & 0 & \cdots & 0 \\
0 & 0 & m_f \omega_{i3}^2 & \cdots & 0 \\
\vdots & \vdots & \vdots & \ddots & \vdots \\
0 & 0 & 0 & \cdots & m_f \omega_{il}^2
\end{bmatrix}
\end{equation}
The mass matrix of the higher-dimensional oscillator system is given by:
\begin{equation}
    \mathbf{M}' = \begin{bmatrix}
                    \mathbf{M}_{N \times N} & \mathbf{0}\\
                    \mathbf{0} & m_f\mathbf{I}_{Nl \times Nl}
                  \end{bmatrix}
\end{equation}

Let the displacement vector of higher dimensional system be denoted by $\mathbf{z} = [\mathbf{x}, \mathbf{y}_1 , \mathbf{y}_2 , ...\mathbf{y}_N]^T$, where $\mathbf{x} \in \mathbb{R}^N$ and $\mathbf{y}_i \in \mathbb{R}^l, \forall i\in [N]$. For the coupled oscillator system defined by stiffness matrix $\mathbf{G}'$, the governing set of differential equations is given by:
\begin{equation}
    \mathbf{M}' \ddot{\mathbf{z}} = -\mathbf{K}' \mathbf{z}
\end{equation}
Here, the matrix $\mathbf{K}'$ is constructed from the stiffness matrix $\mathbf{G}'$. This construction is explained in \autoref{Def: coupled_oscillator_system}. 

\begin{align}
 \mathbf{K}' = \begin{bmatrix}
        {\mathbf{K}} & \mathbf{K}_{12} & \mathbf{K}_{13} & \cdots & \mathbf{K}_{1N}\\
        \mathbf{K}_{12}^T & \mathbf{K}_{22} & \mathbf{0} & \cdots & \mathbf{0}\\
        \mathbf{K}_{13}^T & \mathbf{0} & \mathbf{K}_{33} & \ddots & \mathbf{0} \\
        \vdots & \vdots & \ddots & \ddots & \vdots \\
        \mathbf{K}_{1N}^T & \mathbf{0} & \cdots & \cdots & \mathbf{K}_{NN}
 \end{bmatrix} 
\end{align}
where $\mathbf{K}$ is of the form 
\begin{align}
{\mathbf{K}} =  - \begin{bmatrix}
                    - ({k_{11}} + \sum_{i=2} k_{1i}) & k_{12} & k_{13} & \cdots & k_{1N} \\
                    k_{21} & -({k_{22}} + \sum_{i=1} k_{2i}) & k_{23} & \cdots & k_{2N}\\
                    k_{31} & k_{32} & -({k_{22}} + \sum_{i=3} k_{3i}) & \cdots & k_{2N}\\
                    k_{41} &  k_{42} & k_{43} & \cdots & k_{4N}\\
                    \vdots & \vdots & \vdots & \ddots & \vdots \\
                    k_{N1}  & k_{N2} & \vdots & \cdots & -({k_{NN}} + \sum_{i=2} k_{Ni})\\                 
                  \end{bmatrix}
\end{align}
where we take the $\mathbf{K}_{ij}$ sub-matrices to be be defined as  
\begin{align}
    \mathbf{K}_{12} = - \mathbf{G}_{12} = -\begin{bmatrix}
              \frac{k_{11}}{2l} & \frac{k_{11}}{2l} & \cdots  \frac{k_{11}}{2l}\\
              \mathbf{0} &  \mathbf{0}  & \cdots \mathbf{0}\\
              \vdots &\vdots &\vdots \\ 
              \mathbf{0} &  \mathbf{0}  & \cdots \mathbf{0}
             \end{bmatrix}
\end{align}
and 
\begin{align}
    \mathbf{K}_{13} = - \mathbf{G}_{13} = -\begin{bmatrix}              
              \mathbf{0} &  \mathbf{0}  & \cdots \mathbf{0}\\
              \frac{k_{22}}{2l} & \frac{k_{22}}{2l} & \cdots  \frac{k_{22}}{2l}\\
              \vdots &\vdots &\vdots \\ 
              \mathbf{0} &  \mathbf{0}  & \cdots \mathbf{0}
             \end{bmatrix}
\end{align}
and in general, 
\begin{align}
    \mathbf{K}_{1j} = - \mathbf{G}_{1j} = -\frac{k_{j-1,j-1}}{2l}\begin{bmatrix}              
              \ket{j-2} &  \ket{j-2}  & \cdots \ket{j-2}\\            
             \end{bmatrix}
\end{align}

$\mathbf{K}_{22}$, $\mathbf{K}_{33}$ and in general, $\mathbf{K}_{ii}$ are diagonal matrices, of the form:
\begin{align}
   \mathbf{K}_{ii} = 
-\begin{bmatrix}
-\frac{k_{i-1,i-1}}{2l} -m_f\omega_{i-1,1}^2  & 0 & 0 & \cdots & 0  \\
0 & -\frac{k_{i-1,i-1}}{2l}-m_f\omega_{i-1,2}^2 & 0 & \cdots & 0 \\
0 & 0 & -\frac{k_{i-1,i-1}}{2l}-m_f\omega_{i-1,3}^2 & \cdots & 0 \\
\vdots & \vdots & \vdots & \ddots & \vdots \\
0 & 0 & 0 & \cdots & -\frac{k_{i-1,i-1}}{2l} -m_f\omega_{i-1,N}^2
\end{bmatrix}
\nonumber
\end{align}

Now we regroup terms in the governing differential equation: $\mathbf{M}' \ddot{\mathbf{z}} = \mathbf{K}'\mathbf{z}$ and write it in the form:
\begin{equation}
    \label{eq:Mzddot}
    \hat{\mathbf{M}}\ddot{\mathbf{z}} = -(\mathbf{K}'_0 + \gamma \hat{\mathbf{K}})\mathbf{z}
\end{equation}
Here, $\gamma = \frac{1}{m_f}$. The matrix $\hat{\mathbf{M}}$ is given by:

\begin{equation}
    \hat{\mathbf{M}} = \begin{bmatrix}
                        \mathbf{M}_{N \times N} & \mathbf{0}\\
                        \mathbf{0} & \mathbf{I}_{Nl \times Nl}
                        \end{bmatrix}
\end{equation}

and, $\mathbf{K}'_0$ is given by:
\begin{align}
 \mathbf{K}'_0 = \begin{bmatrix}
       \mathbf{K} & \mathbf{K}_{12} & \mathbf{K}_{13} & \cdots & \mathbf{K}_{1N}\\
        \mathbf{0} & \mathbf{K}_{22}' & \mathbf{0} & \cdots & \mathbf{0}\\
        \mathbf{0} & \mathbf{0} & \mathbf{K}_{33}' & \ddots & \mathbf{0} \\
        \vdots & \vdots & \ddots & \ddots & \vdots \\
        \mathbf{0} & \mathbf{0} & \cdots & \cdots & \mathbf{K}_{NN}'
 \end{bmatrix} 
\end{align}
The sub matrices $\mathbf{K}_{ii}'$ contains only $\omega_{i,j}^2$ terms. 
\begin{align}
   \mathbf{K}_{ii}' = 
-\begin{bmatrix}
-\omega_{i-1,1}^2  & 0 & 0 & \cdots & 0  \\
0 & -\omega_{i-1,2}^2 & 0 & \cdots & 0 \\
0 & 0 & -\omega_{i-1,3}^2 & \cdots & 0 \\
\vdots & \vdots & \vdots & \ddots & \vdots \\
0 & 0 & 0 & \cdots & -\omega_{i-1,N}^2
\end{bmatrix}
\nonumber
\end{align}
and 
\begin{align}
 \hat{\mathbf{K}} = \begin{bmatrix}
        \mathbf{0} & \mathbf{0} & \mathbf{0} & \cdots & \mathbf{0}\\
        \mathbf{K}_{12}^T & \hat{\mathbf{K}}_{22} & \mathbf{0} & \cdots & \mathbf{0}\\
        \mathbf{K}_{13}^T & \mathbf{0} & \hat{\mathbf{K}}_{33} & \ddots & \mathbf{0} \\
        \vdots & \vdots & \ddots & \ddots & \vdots \\
        \mathbf{K}_{1N}^T & \mathbf{0} & \cdots & \cdots & \hat{\mathbf{K}}_{NN}
 \end{bmatrix} 
\end{align}

The sub-matrices $\hat{\mathbf{K}}_{ii}$ are given by:
\begin{align}
   \hat{\mathbf{K}}_{ii} = 
-\begin{bmatrix}
-\frac{k_{i-1,i-1}}{2l}  & 0 & 0 & \cdots & 0  \\
0 & -\frac{k_{i-1,i-1}}{2l}& 0 & \cdots & 0 \\
0 & 0 & -\frac{k_{i-1,i-1}}{2l} & \cdots & 0 \\
\vdots & \vdots & \vdots & \ddots & \vdots \\
0 & 0 & 0 & \cdots & -\frac{k_{i-1,i-1}}{2l}
\end{bmatrix}
\nonumber
\end{align}

Applying \autoref{thm:pertThm}, for $
    \frac{1}{\gamma} = m_f \geq \left( t_s \|\mathbf{K}'\| \| \mathbf{T} \| \|\sqrt{\mathbf{M}}^{-1}\|\right)  \left(1 + \frac{\mathbf{\Xi}(0)}{\epsilon}\right)
$, the displacement vector of the higher dimensional oscillator system $\mathbf{z}(t)$ is $\epsilon$ close to $\mathbf{z}_0(t)$, $\forall t \in [0,t_s]$. More formally, $\| \mathbf{z}_0(t) - \mathbf{z}(t)\| \leq \epsilon$, and $\mathbf{z}_0(t) = [\mathbf{x}, \mathbf{y}_1, \mathbf{y}_2 \cdots \mathbf{y}_N]^T$ obeys the following differential equation:
\begin{equation}
    \hat{\mathbf{M}} \ddot{\mathbf{z}}_0 = -\mathbf{K}_0' \mathbf{z}_0
\end{equation}
Note that $\mathbf{K}_0'$ is an upper triangular block matrix, and hence the governing differential equation can be seperated into two components. The first component, is of the form:
\begin{align}
\label{gov-eq-Y}
    \mathbf{M}\ddot{\mathbf{x}} = \mathbf{K}\mathbf{x} + \mathbf{K}_{12} \mathbf{y}_1 + \mathbf{K}_{13} \mathbf{y}_2 +\cdots+\mathbf{K}_{1N}\mathbf{y}_N
\end{align}
The second component determines the trajectory of $\mathbf{y}_i(t)$, $\forall i \in [N]$. Specifically, $\mathbf{y}_i(t)$ obeys the following differential equation. 
\begin{align}
    \ddot{\mathbf{y}}_i = \mathbf{K}_{i+1,i+1}'\mathbf{y}_i
\end{align}
Note that, since $\mathbf{K}_{i+1,i+1}'$ is a diagonal matrix, the solution $\mathbf{y}_i(t)$ is given by: 
\begin{equation}
\mathbf{y}_i(t)  = \begin{bmatrix}
                    A_{i1}\cos(\omega_{i1}t + \phi_{i1})\\
                    A_{i2}\cos(\omega_{i2}t + \phi_{i2})\\
                    A_{i3}\cos(\omega_{i3}t + \phi_{i3})\\
                    A_{i4}\cos(\omega_{i4}t + \phi_{i4})\\
                    \vdots\\
                    A_{il}\cos(\omega_{i4}t + \phi_{i1})
\end{bmatrix}
\end{equation}

Now, if we consider the terms in Eq.~\eqref{gov-eq-Y}, the action of the matrix $\mathbf{K}_{1i}$ on $\mathbf{y}_{i-1}$ creates a vector with each term containing sum of time dependent functions $A_{ij}\cos(\omega_{ij}t + \phi_{ij})$. They can be written explicitly as:
\begin{equation}
    \mathbf{K}_{12} \mathbf{y}_1 = \begin{bmatrix}
                    \frac{k_{11}}{2l} \sum_j (A_{1j}\cos(\omega_{1j}t + \phi_{1j}))\\
                    0\\
                    0\\
                    0\\
                    \vdots\\
                    0
\end{bmatrix}
\end{equation}
and 
\begin{equation}
    \mathbf{K}_{13} \mathbf{y}_2 = \begin{bmatrix}
                    0\\
                    \frac{k_{11}}{2l} \sum_j (A_{2j}\cos(\omega_{2j}t + \phi_{2j}))\\
                    0\\
                    0\\
                    \vdots\\
                    0
\end{bmatrix}
\end{equation}
Summing up these terms, we get: 
\begin{equation}
    \mathbf{K}_{12}\mathbf{y}_1 + \mathbf{K}_{13}\mathbf{y}_2 + .. \mathbf{K}_{1N-1}\mathbf{y}_N
     = \begin{bmatrix}
                    \frac{k_{11}}{2l} \sum_j (A_{1j}\cos(\omega_{1j}t + \phi_{1j}))\\
                    \frac{k_{22}}{2l} \sum_j (A_{2j}\cos(\omega_{2j}t + \phi_{2j}))\\
                    \frac{k_{33}}{2l}\sum_j (A_{3j}\cos(\omega_{3j}t + \phi_{3j}))\\
                    \frac{k_{44}}{2l} \sum_j (A_{4j}\cos(\omega_{4j}t + \phi_{4j}))\\
                    \vdots\\
                    \frac{k_{NN}}{2l} \sum_j (A_{Nj}\cos(\omega_{Nj}t + \phi_{Nj}))
\end{bmatrix}
\end{equation}

Here, the amplitudes $A_{ij}$ and $\phi_{ij}$ can be controlled by setting the initial conditions. Specifically, the positions and velocities of the auxiliary masses and the desired amplitudes and phase factors are related by:

\begin{equation}
y_{ij}(0) = \frac{A_{ij}}{\sqrt{1+\tan^2 (\phi_{ij}) }}     
\end{equation}
and the initial velocities are given by:
\begin{equation}
\dot{y}_{ij}(0) = -{A_{ij}} \omega_{ij} \sin(\phi_{ij}) 
\end{equation}
Here, $y_{ij}(0)$ the initial position of $j^{\text{th}}$ auxiliary mass connected to $i^{\text{th}}$ mass $m_i$. If we set the following initial conditions for the auxiliary masses,
\begin{equation}
y_{ij}(0) = \frac{2 l f_{ij}}{k_{ii}\sqrt{1+\tan^2 (\phi_{ij}) }}     
\end{equation},
\begin{equation}
\dot{y}_{ij}(0) = -{ \frac{2l f_{ij}}{k_{ii}}} \omega_{ij} \sin(\phi_{ij}) 
\end{equation}
Eq.~\eqref{gov-eq-Y} reduces to:
\begin{align}
    \mathbf{M}\ddot{\mathbf{x}} = -\mathbf{K}\mathbf{x} + \mathbf{f}(t)
\end{align}
Hence,
\begin{equation}
    \label{eq:projector_higher_dim_sys_action}
    \mathbf{P} \mathbf{z}_0(t) = \mathbf{x}(t)
\end{equation}
, $\forall t \in [0, t_s]$. Applying norm inequalities, and using the fact that $\| \mathbf{P}\| = 1$, we can upper bound  $ \| \mathbf{P} (\mathbf{z} - \mathbf{z}_0)  \| $ by $\epsilon$. Specifically, 
\begin{align}
\label{eq:projector_higher_dim_sys_norm_ineq}
    \| \mathbf{P} (\mathbf{z} - \mathbf{z}_0)  \| &\leq  \| \mathbf{P} \| \|\mathbf{z}(t) - \mathbf{z}_0(t) \|   \\ \nonumber
                                                  &\leq  \|\mathbf{P}\| \epsilon \\ \nonumber
                                                  &\leq \epsilon 
\end{align}
Substituting Eq.~\eqref{eq:projector_higher_dim_sys_action}, in Eq.~\eqref{eq:projector_higher_dim_sys_norm_ineq}, we get $ \| \mathbf{P}\mathbf{z}(t) - \mathbf{x}(t)  \| \leq \epsilon$.
\end{proof}

\subsection{Reduction to Quantum Evolution}
\label{sec: redn_to_quantum_evolution}
Applying \autoref{thm:pertThm N Dimensions}, given an inhomogeneous dynamical system of the form $\mathbf{M}\ddot{\mathbf{x}} = -\mathbf{K}\mathbf{x} + \mathbf{f}(t)$, we can construct a higher dimensional oscillator system with dynamics $\mathbf{M}'\ddot{\mathbf{z}} = -\mathbf{K}'\mathbf{z}$ such that a subspace of that system evolves according to $\mathbf{M}\ddot{\mathbf{x}} = -\mathbf{K}\mathbf{x} + \mathbf{f}(t)$, $\forall t \in [0,t_s]$. The construction of $\mathbf{M}'$ and $\mathbf{K}'$ is described in \autoref{thm:pertThm N Dimensions}. Now, we are in a position to reduce the problem to Schr\"{o}dinger Evolution using techniques inspired from \cite{babbush2023exponential,costa2019quantum}. Specifically, let $\mathbf{A} = \sqrt{\mathbf{M}}^{-1}\mathbf{K}'\sqrt{\mathbf{M}}^{-1} $ and let $\mathbf{B}$ be any $N\times M$ matrix such that $\mathbf{A} = \mathbf{B}\mathbf{B}^{\dagger}$ and let $\mathbf{w} = \sqrt{\mathbf{M}'} \mathbf{z}$. Let $\mathbf{H}$ be the block Hamiltonian:
\begin{equation}
\label{eq: ham}
    \mathbf{H} = -\begin{bmatrix}
        \mathbf{0} & \mathbf{B}\\
        \mathbf{B}^{\dagger} & \mathbf{0}\\
    \end{bmatrix}
\end{equation}
Then, the Schr\"{o}dinger Equation induced by $\mathbf{H}$ is:
\begin{equation}
    \label{eq-Schr\"{o}dinger}
    \dot{\ket{\psi}} = -i\mathbf{H} \ket{\psi}
\end{equation}
we can show by direct substitution that:
\begin{equation}
  \ket{\psi} \propto  \begin{bmatrix}
                  \dot{\mathbf{w}}(t)\\
                  i \mathbf{B}^{\dagger} \mathbf{w}(t)
                 \end{bmatrix}
\end{equation}
is a solution to Eq \eqref{eq-Schr\"{o}dinger}. \\
The matrix $\mathbf{B}$ is defined as:
\begin{equation}
\label{eq: B_definition}
\mathbf{B} \, |j, k\rangle = 
\begin{cases} 
\sqrt{ \frac{\mathbf{G}'(j,j)}{m_j}} \, |j\rangle, & \text{if } j = k \\
\sqrt{ \frac{\mathbf{G}'(j,k)}{m_j}} (|j\rangle - |k\rangle), & \text{if } j < k
\end{cases}
\end{equation}

Note that $f_j(t) = \sum_k f_{jk}cos(\omega_{jk}t)$, and $\omega_{jk}$ is the frequency associated with fourier expansion of time dependent force acting on masses. In our indexing scheme, $j \in [N]$ correspond to the actual masses $m_j$, and any $j>N$ correspond to auxillary masses added to the system to introduce time dependence. Specifically, for each mass $m_j$, we add $l = 2^r -1$ additional auxillary masses of mass $m_f$. These masses are indexed from $N+1$ to $Nl$.  Hence, $\forall j > N, \sqrt{\frac{\mathbf{G}'(j,j)}{m_j}} = \omega_{  \lfloor{(j-N -1)/l}\rfloor,(j-N) \mod l}$. We assume a query based access to the $\omega_{jk},k_{jk}$ and $m_j$ through oracles as per~\autoref{Def:Oracles-forced-oscillator-system}. We also assume access to Oracle $\mathcal{S}$ which allows us to map:\\

\begin{equation}
    \ket{j,l} \rightarrow \ket{j,g(j,l)}
\end{equation}
Here, $g(j,l)$ is the column index of the $l$ th nonzero entry in the $j$ th row of the stiffness matrix $\mathbf{G}'$. We also assume access to an oracle $O_{G'}$ which marks nonzero elements of the matrix $\mathbf{G}'$.
\begin{equation}
O_{\mathbf{G}'} \ket{j}\ket{k}\ket{0}  = 
\begin{cases} 
\ket{j}\ket{k}\ket{1}, & \text{if } \mathbf{G}'(j,k) \neq 0\\
\ket{j}\ket{k}\ket{0}, & \text{if } \mathbf{G}'(j,k) = 0\\
\end{cases}
\end{equation}


\begin{table}[h!]
\centering
\renewcommand{\arraystretch}{1.4} 
\[
\begin{array}{|c|c|c|c|c|c|c|c|}
\hline
{j} & {k} & {} & i & g_1(i,j,k) & g_2(i,j,k) & g_3(i,j,k)\\ \hline
{j < N} & {k < N} & {j = k} & 0 & {\frac{\overline{G(j,k)}}{2}} & \bar{m}_j & 1 \\ \hline
j < N & k < N & j \neq k & 1 & \overline{G(j,k)} & \bar{m}_j & 1 \\ \hline
{j < N} & {k < N} & {j = k} & 2 & \overline{\frac{G(j,j)}{2l}} & \bar{m}_j & 1 \\ \hline
j \geq N & k \geq N & j \neq k & 3 & \overline{\frac{G(j,k)}{2l}} & \bar{m}_f & 1 \\ \hline
j \geq N & k \geq N & j = k & 4 & \bar{\omega}_{jk} & 1 & 0 \\ \hline
j \geq N & k \geq N & j \neq k & 5 & 0 & 1 & 0 \\ \hline
\end{array}
\]

\caption{Output States of the Gadget Oracle for different values of $j$, $k$, and $i$. The table shows how Register 1, and Register 2, and Register 3 are computed for each scenario.}
\label{tab:logical_operations}
\end{table}
From these Oracles, we construct a Gadget Oracle that gives direct access to elements of the matrix $\mathbf{G}'$. 
Now, we explain the construction of a unitary $\mathbf{\mathcal{U}}_B$ that is a block encoding of $\mathbf{B}$ and then explain the simulation method. 
\begin{lemma}
    Let $p = n + r $, $\epsilon' > 0,s = \log(1/\epsilon')$, $m_{min} <  m_j < m_f < m_{max}$, $k_{max} \geq k_{jk}$, $\forall j,k \in [N]$ and $\alpha = \max_{jk} \left({\frac{k_{jk}}{m_j}} , \omega_{jk}^2\right)$. Consider the padded, $d$ sparse $p \times p^2$ matrix  $\mathbf{B}$ from our construction discussed above. Then, there exists a unitary $\mathbf{\mathcal{U}_B}$ acting on $2p + s + 2$ qubits that provides a block encoding $\mathbf{B}$ with an additive error $\epsilon'$ as follows:

  \[    \left\| (\bra{0}^{\otimes p} \otimes \bra{0}^{\otimes s+2} \otimes \mathbf{I}_p)\mathbf{\mathcal{U}_B} (\ket{0}^{\otimes s+2}\mathbf{I}_p \otimes \mathbf{I}_p) - \frac{1}{\sqrt{2 \alpha d}} \mathbf{B} \right\| \leq \epsilon'   \] 
    
\end{lemma}

\begin{proof}
    
First, we show how to block encode $\mathbf{B}^{\dagger}$, and take the conjugate transpose to implement a unitary circuit encoding $\mathbf{B}$. From Eq~\eqref{eq: B_definition}, we have:

\begin{equation}
    \mathbf{B}^{\dagger}\ket{j} = \left( \sum_{k \geq j} \sqrt{\frac{\mathbf{G}'(j,k)}{m_j}} \ket{j}\ket{k} \right) - \left( \sum_{k < j} \sqrt{\frac{\mathbf{G}'(j,k)}{m_j}} \ket{k}\ket{j} \right)
\end{equation}
The steps to block encode $ \mathbf{B}^{\dagger}$ are:
\begin{enumerate}
\item Apply a unitary that performs the map $\ket{0}^{\otimes n + r} \rightarrow \frac{1}{\sqrt{d}} \sum_{l=1}^{d} \ket{l}$. 
\item Apply the oracle $\mathcal{S}$ for the positions of non-zero entries of $\mathbf{G}'$ to map $\ket{\ell}$ to $\ket{g(j, \ell)}$ 
\item Apply the Gadget Oracle to compute $\overline{m}_j$ and $\overline{G}_{jk}$, i.e., perform the map\\
$\ket{j, 0} \rightarrow \ket{j, \overline{m}_j}$, $\forall j < N$ \\
$\ket{j, 0} \rightarrow \ket{j, \overline{m}_f}$, $\forall j \geq N$ \\
and \\
$\ket{j, k, 0} \rightarrow \ket{j, k, \overline{k}_{jk}}$, $\forall k < N$\\
$\ket{j, k, 0} \rightarrow \ket{j, k, \overline{\omega}_{jk}}$, $\forall j=k \geq N$ \\
\item Prepare the equal superposition state of $s$ ancilla qubits, i.e., $\frac{1}{2^{s/2}} \sum_{x=1}^{2^s} \ket{x}$.
\item When $\text{Register 3} = 1 $, Using coherent arithmetic, compute the square of $x$ and multiplications needed for the inequality test
\begin{equation}
\frac{k_{\text{max}} \overline{k}_{jk}}{2^{r_k} } \leq \frac{x^2}{2^{2s}}\alpha
\frac{{m}_{\text{max}} \overline{m_j}}{2^{r_m} }, \tag{A7}
\end{equation}
where $r_k$ and $r_m$ are the numbers of bits of $\overline{k}_{jk}$ and $\overline{m_j}$, respectively. This gives the factor in the amplitude approximately $\sqrt{\frac{k_{jk}}{m_j \alpha}}$.
\item When $\text{Register 3} = 0 $, use bit representation of the value in the Register 1, to prepare amplitude $\frac{\omega_{jk}}{\sqrt{\alpha}}$ or $0$.

\item Reverse the computation of the arithmetic for Eq.~(A7) and the oracles for $\overline{m}_j$, $\bar{\omega}_{jk}$ and $\overline{k}_{jk}$. This transforms the working registers back to an all-zero state.

\item Perform the map $\ket{j} \ket{k} \ket{0} \rightarrow \ket{k} \ket{j} \ket{1}$ if $k < j$ or leave the state $\ket{j} \ket{k}\ket{0}$ unchanged otherwise. 

\item Apply $HZ$ on the ancilla qubit of the previous step to implement $\ket{0} \rightarrow \frac{1}{\sqrt{2}}(\ket{0} + \ket{1})$ and $\ket{1} \rightarrow \frac{1}{\sqrt{2}}(\ket{0} -\ket{1}) $
\end{enumerate}

The sequence of steps define $\mathcal{U_B}^{\dagger}$

\end{proof}

Now, we consider block encoding of the hamiltonian $\mathbf{H}$ which is defined in Eq.~\eqref{eq: ham}. The Hamiltonian $\mathbf{H}$ can be block encoded by applying Lemma 9 from \cite{babbush2023exponential}. In our case, the block encoding constant is different, and applying Lemma 9, using $O(1)$ queries to controlled versions of $\mathbf{\mathcal{U}}_B$ and $\mathcal{O}(p)$ two qubit gates, we can construct a unitary $\mathbf{\mathcal{U}_{H}}$ such that: 

\begin{equation}
\label{eq: block_encode_H}
  \left\| (\bra{0}^{\otimes s+3} \otimes \mathbf{I}) \mathbf{\mathcal{U}_{H}} (\ket{0}^{\otimes s+3} \otimes \mathbf{I}) - \frac{1}{\sqrt{2\alpha d}} \mathbf{H} \right\| \leq \epsilon'
\end{equation}

\QuantumAlgorithmThm*

\begin{proof}
Applying \autoref{thm:pertThm N Dimensions}, for the forced oscillator system with dynamics $\mathbf{M}\ddot{\mathbf{x}} = -\mathbf{K}\mathbf{x} + \mathbf{f}(t)$, first we construct a higher dimensional coupled classical oscillator system with dynamics $\mathbf{M}' \dot{\mathbf{z}} = -\mathbf{K}' \mathbf{z}$ such that $\| \mathbf{P}\mathbf{z}(t) - \mathbf{x}(t) \| \leq \epsilon, \forall t \in [0,t_s]$. Here, $\mathbf{M}' \in \mathbb{R}^{N(l+1)\times N(l+1)} $ is the mass matrix of the higher dimensional oscillator system and $\mathbf{K}'\in \mathbb{R}^{N(l+1)\times N(l+1)} $ is the force matrix. The initial displacement vector of the higher dimensional dynamical system as per \autoref{thm:pertThm N Dimensions} is given by:
\begin{equation}
    \mathbf{z}(0) = [\mathbf{x}(0),\mathbf{y}(0)]^T
\end{equation}
where
\begin{equation}
    \mathbf{y}(0) = [\mathbf{y}_1(0), \mathbf{y}_2(0),\mathbf{y}_3(0)\cdots
    \mathbf{y}_N(0) ]^T
\end{equation}
and
\begin{equation}
    \mathbf{y}_i(0) = [y_{i,1}(0), y_{i,2}(0),y_{i,3}(0)\cdots
    y_{i,l}(0) ]^T
\end{equation}
\begin{equation}
    y_{i,j}(0) = \frac{2l f_{ij} \cos(\phi_{ij})}{k_{ii}}
\end{equation}
The initial velocity vector is:
\begin{equation}
    \dot{\mathbf{y}}(0) = [\dot{\mathbf{y}}_1(0), \dot{\mathbf{y}}_2(0),\dot{\mathbf{y}}_3(0)\cdots
    \dot{\mathbf{y}}_N(0) ]^T
\end{equation}
where
\begin{equation}
    \dot{\mathbf{y}}_i(0) = [\dot{y}_{i,1}(0), \dot{y}_{i,2}(0),\dot{y}_{i,3}(0)\cdots
    \dot{y}_{i,l}(0) ]^T
\end{equation}
\begin{equation}
    \dot{y}_{i,j}(0) = \frac{-2l \omega_{ij} f_{ij} \sin(\phi_{ij})}{k_{ii}}
\end{equation}

From \autoref{sec: redn_to_quantum_evolution}, we know that the dynamics of the higher dimensional coupled oscillator system can be reduced to Schr\"{o}dinger Evolution of a quantum system, wherein the amplitudes of the quantum state $\ket{\bm{\psi}(t)} = [\dot{\mathbf{w}}, -i\mathbf{B}^{\dagger} \mathbf{w}]^T$, $\mathbf{w} = \sqrt{\mathbf{M}'}\mathbf{z}$ encodes the momenta and displacement vector of the higher dimensional oscillator system. To simulate time evolution of the quantum system, first, we prepare the following initial state $  \ket{\psi(0)}$:
\begin{equation}
    \ket{\bm{\psi}(0)}  =  \frac{1}{\sqrt{2E}} \begin{pmatrix}
                                                                         \mathbf{\sqrt{M}} \dot{\mathbf{x}}(0)\\
                                                                         \sqrt{m_f}\mathbf{I}_{Nl \times Nl} \dot{\mathbf{y}}(0) \\ 
                                                                          i\bm{\mu}(0) \\
                                                                      \end{pmatrix}
\end{equation}

Here, the vector $\bm{\mu}(0) = -\mathbf{B}^{\dagger}\mathbf{w(0)}$ contains terms proportional to displacements of the masses in the coupled oscillator system.
The constant $E$ is equal to the total initial energy of the oscillator system, which is given by:
\begin{align}
    E &= \frac{1}{2} \left( \dot{\mathbf{x}}^T(0) \mathbf{M}\dot{\mathbf{x}}(0) + m_f \dot{\mathbf{y}}^T(0)\dot{\mathbf{y}}(0) + \bm{\mu}^T(0) \bm{\mu}(0)\right) \\ \nonumber
    &= \frac{1}{2} ( \sum_i m_i \dot{x}_i^2(0)  + m_f \sum_i \dot{\mathbf{y}}_i^T(0)\dot{\mathbf{y}}_i(0)  + \sum_i k_{ii} x_i^2(0) + \\
    &+ \sum_{i < j} k_{ij}(x_i(0) -x_j(0))^2  + \sum_{ij}m_f\omega_{ij}^2 y_{i,j}^2(0) + \sum_{ij}\frac{k_{ii}}{2l}(x_i(0) - y_{i,j}(0) )^2)    
\end{align}
More compactly, the initial total energy can be written as:
\begin{equation}
    E = E_{sys}  + E_{aux} + E_{int},
\end{equation}
where
\begin{align}
    E_{sys} &= \frac{1}{2}( \sum_i m_i \dot{x}_i^2(0) + \sum_i k_{ii} x_i^2(0) +  \sum_{i < j} k_{ij}(x_i(0) -x_j(0))^2   ) \\ \nonumber 
    E_{aux} &= \frac{1}{2} (  m_f \sum_i \dot{\mathbf{y}}_i^T(0)\dot{\mathbf{y}}_i(0) + \sum_{ij}m_f\omega_{ij}^2 y_{i,j}^2(0) ) \\ \nonumber
    E_{int} &= \sum_{ij}\frac{k_{ii}}{2l}(x_i(0) - y_{i,j}(0) )^2
\end{align}
Here, $E_{sys}$ is the total initial energy of the oscillator system which we wish to simulate, $E_{aux}$ is the total initial energy of the auxillary system which exerts desired time dependent forces, and $E_{int}$ is the potential energy due to interactions between these two systems. 
The initial state $\ket{\bm{\psi}(0)}$ can be prepared using the following procedure:
First, using the state preparation oracle $\mathbf{\mathcal{W}}$, we prepare the following quantum state:
\begin{align}
    \ket{0}\mathbf{\ket{0}}  \rightarrow \frac{1}{\beta} \left( \ket{0}\ket{\dot{\mathbf{x}}(0)}  + \ket{1}\ket{\dot{\mathbf{y}}(0)}+\ket{2}\ket{\mathbf{x}(0)} + \ket{3}\ket{\mathbf{y}(0)}  \right)
\end{align}
Then, we apply Lemma 10 from \cite{babbush2023exponential} to prepare the quantum state $\ket{\bm{\psi}(0)}$. After preparing the initial state, we apply Hamiltonian Simulation techniques from \cite{low2019hamiltonian} to implement a unitary circuit approximating the matrix exponential: $e^{-it \mathbf{H}}$ and prepare the quantum state $\ket{\psi(t)} =  e^{-it \mathbf{H}} \ket{\psi(0)}$. The construction of unitary circuit require $\mathcal{O} (t \lambda + \log(1/\epsilon)) $ calls to the unitary $\mathbf{\mathcal{U}_{H}}$. Here, $\lambda$ is the block encoding constant, which from Eq.~\eqref{eq: block_encode_H}, is given by, $\lambda = \sqrt{2\alpha d}$, and $\alpha := \max_{ij}\left(\omega_{ij}^2,{\frac{k_{ij}}{m_j}} \right)$. 
\end{proof}
The prepared quantum state $\ket{\mathbf{\psi}(t)}$ encodes the velocities and displacements of the forced oscillator system. Our choice of encoding enables easy estimation of the system's kinetic energies or potential energies at any time $t$. As an example application, consider the problem of estimating the kinetic energies of a subset of the masses $V \subseteq N$ in the oscillator system. 

\begin{proposition}
    Given same inputs as \autoref{prblm 1}, and an oracle for 
 the subset of oscillators, $V \subseteq N$, define $\hat{k}_V(t)$ as:
 \[|\hat{k}_V(t) - K_V(t)/{E_{sys}}| \leq \epsilon.  \]
 Then, there exists a quantum algorithm that can estimate $\hat{k}_V(t) \in \mathbb{R}$ using $\mathcal{O}\left(\frac{f_{\max}^2l^6 t^2}{\epsilon^2}\right)$ uses of Oracle for $V$, and the quantum circuit that prepares $\ket{\psi(t)}$ and its inverse and controlled version. Here $K_V(t) = \frac{1}{2}\sum_{i\in V} m_i \dot{x}_i^2$  is the kinetic energy of $V$ at time t and $E_{sys}$ is the initial total energy of the oscillator system. 
\end{proposition}

\begin{proof}
Let $E = E_{sys} + E_{int} + E_{aux}$ be the total energy of the oscillator system. Recall that the oscillator system involves $N$ masses of the original forced oscillator system, and $Nl$ additional masses added to introduce time-dependent external forces. Let $K_V(t)/E \approx K_V(t)/E$ be an estimate of the kinetic energy of the oscillators divided by the total energy of the system.  More precisely, assume $\bar{k}_V(t)$ satisfies the following equation:
\begin{equation}
\label{eq:ke_estimate}
   | \bar{k}_V(t) - K_V(t)/E | \leq \epsilon_1
\end{equation}
Applying Theorem 5 from \cite{babbush2023exponential}, we can obtain an estimate $\bar{k}_V(t) \approx K_V(t)/E$ with success probability $1-\delta$ using $\mathcal{O}(\log(1/\delta)/\epsilon_1)$ uses of $V$ and the quantum circuit that prepares $\ket{\psi(t)}$. To obtain $\epsilon$ precise estimate of $\hat{k}_V(t)$, we need to estimate $\bar{k}_V(t)$ with a high precision since the total energy of the oscillator system, $E$ can be large relative to $E_{sys}$. First, note that:
\begin{equation}
    \hat{k}_V(t) = \frac{E}{E_{sys}}\bar{k}_V(t)
\end{equation}
Now, multiply both sides of the \autoref{eq:ke_estimate} with $\frac{E}{E_{sys}}$. This yields:
\begin{equation}
\label{eq:kinetic_energy_bound}
   | \hat{k}_V(t) - K_V(t)/E_{sys} | \leq  \frac{E \epsilon_1}{E_{sys}}
\end{equation}
\begin{align}
    \frac{E}{E_{sys}} &= \frac{E_{sys} + E_{aux} + E_{int}}{E_{sys}} \\ \nonumber
    &= 1 + \frac{E_{aux} + E_{int}}{E_{sys}}
\end{align}
Since $E_{aux} \propto Nm_f \omega_{max}^2 f_{\max}^2 l^3 \in \mathcal{O}(Nf_{\max}^2 l^3t/\epsilon_2) $, $E_{int}$ negligible compared to $E_{aux}$ and $E_{sys} \geq \mathcal{O}(N (m_{min} \dot{x}^2(0)+ k_{min}x^2(0))$, 
\begin{equation}
    \frac{E}{E_{sys}} \in \mathcal{O}(f_{\max}^2 l^3t/\epsilon_2)
\end{equation}
Substituting in \autoref{eq:kinetic_energy_bound}, we get 
\begin{equation}
    | \hat{k}_V(t) - K_V(t)/E_{sys} | \leq  \frac{E \epsilon_1}{E_{sys}} \leq \mathcal{O} (f_{\max}^2 l^3t \epsilon_1/\epsilon_2)
\end{equation}
To obtain the desired final precision $\epsilon$, it is sufficient to choose $\epsilon_2 = \sqrt{\epsilon_1}$, and $\epsilon_1 \in \mathcal{O}(\frac{\epsilon^2}{f_{\max}^2l^6 t^2})$. Hence the final cost of obtaining the kinetic energy estimate is $\mathcal{O}(\frac{f_{\max}^2l^6 t^2}{\epsilon^2})$.
\end{proof}

Note that the computational cost of estimating kinetic energies of the oscillators is quadratic in $t/{\epsilon}$, while for conservative oscillator systems, it is linear in $1/{\epsilon}$. The reason behind this increase in cost is the presence of auxillary oscillators present in the system, added to simulate time dependent external forces. 
\section{Nonlinear Schr\"{o}dinger Equation}
\label{sec:nonlinear_Schrodinger_eq}

In this section, we focus on \autoref{nonlinear_Schr\"{o}dinger}, where the goal is to simulate the dynamics of a quantum system evolving according to the nonlinear Schr\"{o}dinger equation defined in \autoref{Def: nonlinear_Schr\"{o}dinger_eq}. Recall that the nonlinear Schr\"{o}dinger Equation is of the form: 
\begin{align}
\label{eq: nonlinear_Schr\"{o}dinger}
    \dot{\ket{\bm{\psi}} } = -i\mathbf{H}_1 \ket{\bm{\psi}} + \mathbf{H}_2 \ket{\bm{\psi}} \otimes \ket{\bm{\psi}}
\end{align}

Here, $\mathbf{H}_1 \in \mathbb{C}^{N \times N}$ is a Hermitian matrix, and $\mathbf{H}_2 \in \mathbb{C}^{N \times N^2}$ is a rectangular matrix. Because of the presence of the nonlinear term $\mathbf{H}_2 \ket{\bm{\psi}} \otimes \ket{\bm{\psi}}$ in Eq.~\eqref{eq: nonlinear_Schr\"{o}dinger}, we cannot apply Hamiltonian Simulation methods directly to simulate the dynamics of a quantum system evolving according to Eq.~\eqref{eq: nonlinear_Schr\"{o}dinger}. Although we can apply Carlemann Linearization techniques to embed the nonlinear dynamics into a higher-dimensional linear system, the resulting time evolution is not guaranteed to be unitary. While techniques like LCHS \cite{an2023linear} can simulate non-unitary dynamics, the operator derived from Carlemann Linearization may not be positive semidefinite, making it unsuitable for direct application of LCHS methods. To address this challenge, we introduce a novel approach to symmetrize the operator obtained from Carlemann Linearization. After symmetrization, we apply Hamiltonian Simulation methods to simulate the dynamics of symmetrized system and thereby simulate nonlinear quantum dynamics. Before detailing our method, we first establish the conditions that $\mathbf{H}_2$ must meet to ensure that the norm of the wavefunction $\braket{ \bm{\psi} (t) | \bm{\psi} (t)}$ remains non-increasing throughout the time evolution. This condition is necessary to prove the convergence of the Truncated Carlemann system. 

\begin{lemma}
\label{lemma: norm_nonlinear_Schr\"{o}dinger_eq}
    Let the nonlinear Schr\"{o}dinger Equation of the form: 
    \begin{equation}
         \dot{\ket{\bm{\psi}} } = -i\mathbf{H}_1 \ket{\bm{\psi}} + \mathbf{H}_2 \ket{\bm{\psi}} \otimes \ket{\bm{\psi}},
    \end{equation} where $\mathbf{H}_1 \in \mathbb{C}^{N \times N}$ is a Hermitian matrix and $\mathbf{H}_2 \in \mathbb{C}^{N \times N^2}$ is a rectangular matrix and let the norm of initial conditions, be $\braket{ \bm{\psi} (0) | \bm{\psi} (0)} = \beta$. Let $ \Upsilon(t) = \Re(\bra{\bm{\psi}(t) }\mathbf{H}_2 \ket{\bm{\psi}(t)} \otimes \ket{\bm{\psi}(t)})  \leq \bar{\Upsilon} \leq 0$ , $\forall t \in [0,T]$. Then the norm of the wavefunction at any time $t \in [0,T]$ is upper bounded by:
    \[
    \braket{ \bm{\psi} (t) | \bm{\psi} (t)} \leq  \beta
    \]

\end{lemma}
\begin{proof}
    \begin{align}
        \frac{d}{dt} \braket{ \bm{\psi} (t) | \bm{\psi} (t)} &= \bra{\bm{\psi}}(-i\mathbf{H}_1 \ket{\bm{\psi}} + \mathbf{H}_2 \ket{\bm{\psi}} \otimes \ket{\bm{\psi}} )
        + ( \bra{\bm{\psi}}i\mathbf{H}_1  + \bra{\bm{\psi}} \otimes \bra{\bm{\psi}} \mathbf{H}_2^{\dagger} )\ket{\bm{\psi}}\\
        &= \bra{\bm{\psi}} \mathbf{H}_2 \ket{\bm{\psi}} \otimes \ket{\bm{\psi}} + \bra{\bm{\psi}} \otimes \bra{\bm{\psi}} \mathbf{H}_2^{\dagger}\ket{\bm{\psi}} \\ \nonumber
        &= 2\Re( \bra{\bm{\psi}(t)} \mathbf{H}_2 \ket{\bm{\psi}(t)} \otimes \ket{\bm{\psi}(t)})\\ 
    \end{align} 
    Since $\Re( \bra{\bm{\psi}(t)} \mathbf{H}_2 \ket{\bm{\psi}(t)} \otimes \ket{\bm{\psi}(t)}) \leq \bar{\Upsilon} $ by assumption, 
    \begin{align}
        \frac{d}{dt} \braket{ \bm{\psi} (t) | \bm{\psi} (t)} \leq 2 \bar{\Upsilon}
    \end{align}
    Hence, $\braket{ \bm{\psi} (t) | \bm{\psi} (t)} \leq  \braket{ \bm{\psi} (0) | \bm{\psi} (0)}$
\end{proof}

\begin{table}
    \centering
    \begin{tabular}{|c|c|c|c|}
        \hline
        \textbf{Condition} & \textbf{Truncation Error Bound} & \textbf{Lower Bound for $k$}  & \textbf{Time Period} \\
        \hline
        $\mu(-iH_1) < 0$ & $\mathcal{\epsilon}(t) \leq \frac{\| \mathbf{H}_2 \|^k}{\mu(-i\mathbf{H}_1)^k} (1 - e^{\mu(-iH_1)t})^k$ & $k \geq \frac{\log(\frac{1}{\epsilon})}{\log\left(\frac{ \mu(-iH_1)}{\| \mathbf{H}_2\|(1 - e^{\mu(-iH_1)t})}\right)}$ & $\forall t \in [0,\infty)$\\
        \hline
        $\mu(-iH_1) = 0$ & $\mathcal{\epsilon}(t) \leq (\|\mathbf{H}_2\|t)^k$ & $k \geq \frac{\log(\frac{1}{\epsilon})}{\log(\frac{1}{\| \mathbf{H}_2 \|t})}$ & $\forall t \in [0, \frac{1}{\|\mathbf{H}_2\|} ]$\\
        \hline
        $\mu(-iH_1) > 0$ & $\mathcal{\epsilon}(t) \leq \frac{\| \mathbf{H}_2 \|^k}{\mu(-i\mathbf{H}_1)^k} (1 - e^{\mu(-iH_1)t})^k$  & $k \geq \frac{\log(\frac{1}{\epsilon})}{\log\left(\frac{ \mu(-iH_1)}{\| \mathbf{H}_2\|(1 - e^{\mu(-iH_1)t})}\right)}$ & $0 < t < \frac{1}{\mu(-i\mathbf{H}_1)} \ln \left( \beta \right)
$\\
        \hline
    \end{tabular}
    \caption{Summary of conditions and truncation error bounds for Carlemann Linearization. Here the restriction on time gives the maximum time that the simulation can be performed for before Carlemann $\beta =  ( 1 + \frac{\mu(-i\mathbf{H}_1)}{\|\mathbf{H}_2\|}) $}
    \label{tab:carlemann_truncation}
\end{table}

\subsection{Carlemann Embedding}
We can apply Carlemann linearization techniques \cite{carleman1932application} to embed the finite-dimensional nonlinear Schr\"{o}dinger Evolution into an infinite-dimensional linear dynamical system.  The basic idea behind Carlemann can be understood easily in the scalar case.  It stems from the observation that $\partial_t x(t) = ax+bx^2$ can be re-written as the linear-looking equation $\partial_t x(t) = ax+bw$ for $w=x^2$.  The differential equation for $w$, however, is of the form $\partial_t w(t) = 2bx(t)\partial_t x(t)=2baw(t) +2bx(t)w(t)$, which is nonlinear itself.  The idea is that if the nonlinearity is sufficiently small then the remainder term neglected by ignoring a higher-order nonlinear term vanishes.  This allows certain weakly nonlinear differential equations to be simulated on quantum computers, which only naturally can simulate unitary and in turn linear dynamics~\cite{liu2021efficient}.   This idea can also be easily generalized to the case of vector-valued functions, wherein tensor products replacing ordinary multiplication: $\partial_t \mathbf{x}(t) = \mathbf{A} \mathbf{x}(t) + \mathbf{B} \mathbf{x}(t)\otimes \mathbf{x}(t)$.

The Carlemann linearized system of differential equations for Eq \eqref{eq: nonlinear_Schr\"{o}dinger} can be written as:

\begin{align}
\label{Carlemann_Linearized_EOM}
\frac{d}{dt}
\begin{bmatrix}
    \bm{\ket{\psi}} \\
    \bm{\ket{\psi}}^{\otimes 2} \\
    \bm{\ket{\psi}}^{\otimes 3} \\
    \bm{\ket{\psi}}^{\otimes 4} \\
    \vdots
\end{bmatrix} =     -i \begin{bmatrix}
        \mathbf{H}_1 & i\mathbf{H}_2 & \mathbf{0} & \mathbf{0} & \cdots\\
        \mathbf{0} & \mathbf{H}_1 \oplus^1 \mathbf{H}_1 & i\mathbf{H}_2 \oplus^1 \mathbf{H}_2 & \mathbf{0} & \cdots\\
        \mathbf{0} & \mathbf{0} & \mathbf{H}_1 \oplus^2 \mathbf{H}_1  & i\mathbf{H}_2 \oplus^2 \mathbf{H}_2 & \cdots\\
        \mathbf{0} & \mathbf{0} & \mathbf{0} & \mathbf{H}_1 \oplus^3 \mathbf{H}_1 & i\mathbf{H}_2 \oplus^3 \mathbf{H}_2 & \cdots\\
        \vdots & \vdots & \vdots & \vdots & \vdots & \ddots
    \end{bmatrix}
\begin{bmatrix}
    \bm{\ket{\psi}} \\
    \bm{\ket{\psi}}^{\otimes 2} \\
    \bm{\ket{\psi}}^{\otimes 3} \\
    \bm{\ket{\psi}}^{\otimes 4} \\
    \vdots
\end{bmatrix}
\end{align}

Where, $\oplus^i$ is notational shorthand for a padding operation defined as follows:
\begin{align}
    \mathbf{H}_1 \oplus^1 \mathbf{H}_1 = I \otimes \mathbf{H}_1 + \mathbf{H}_1 \otimes I
\end{align},
\begin{align}
    \mathbf{H}_1 \oplus^2 \mathbf{H}_1 = I \otimes \mathbf{H}_1 \otimes I + \mathbf{H}_1 \otimes I \otimes I + I \otimes I \otimes \mathbf{H}_1,\end{align}
and 
\begin{align}
    \mathbf{H}_1 \oplus^i \mathbf{H}_1 = \sum_{j=1}^{i+1} I^{\otimes(i+1-j)} \otimes \mathbf{H}_1 \otimes I^{\otimes(j-1)}.
\end{align}

Now, we define $\ket{\mathbf{w}_k}$ to encode $\ket{\bm{\psi}}^{\otimes k}$  padded with zero vector such that each $ \ket{\mathbf{w}_i} \in \mathbb{C}^{N^{k}}$.
\begin{align}
    \ket{\mathbf{w}_1} =  \bm{\ket{0}}^{\otimes k-1}\bm{\ket{\psi}} = \begin{bmatrix}
            \bm{\ket{\psi}} \\
            \mathbf{0}
          \end{bmatrix}
\end{align}

\begin{align}
      \ket{\mathbf{w}_2}  = \bm{\ket{0}}^{\otimes k-2}\bm{\ket{\psi}}^{\otimes 2} = \begin{bmatrix}
            \bm{\ket{\psi}}^{\otimes 2} \\
            \mathbf{0}
          \end{bmatrix}
\end{align}
and in general,
\begin{align}
      \ket{\mathbf{w}_i}  = \bm{\ket{0}}^{\otimes k-i}\bm{\ket{\psi}}^{\otimes i} = \begin{bmatrix}
            \bm{\ket{\psi}}^{\otimes i} \\
            \mathbf{0}
          \end{bmatrix}
\end{align}
Now we define $\ket{\mathbf{p}}$:
\begin{align}
    \ket{\mathbf{p}} = \sum_i \ket{i-1} \otimes \ket{\mathbf{w}_i}  
\end{align}
Here,  $\ket{\mathbf{p}} \in \mathbb{C}^{k.N^{k}}$
Now, we can write padded version of Eq.~\eqref{Carlemann_Linearized_EOM}  as:
\begin{align}
\label{DE_After_Carlemann}
    \frac{d}{dt} \ket{ \mathbf{p}(t)}  = -i\mathbf{Q} \ket{\mathbf{p}(t)}
\end{align}
Where 
\begin{align}
\mathbf{Q} = \sum_{i=0}^{k-1} \ket{i} \bra{i} \otimes \mathbf{A}_{i+1} + \ket{i} \bra{i+1} \otimes \mathbf{B}_{i+1} = \begin{bmatrix}
    \mathbf{A}_1 & \mathbf{B}_1 & \mathbf{0} & \mathbf{0} & \cdots \\
    \mathbf{0} & \mathbf{A}_2 & \mathbf{B}_2 & \mathbf{0} & \cdots \\
    \mathbf{0} & \mathbf{0} & \mathbf{A}_3 & \mathbf{B}_3 & \cdots \\
    \mathbf{0} & \mathbf{0} & \mathbf{0} & \mathbf{A}_4 & \mathbf{B}_4 & \cdots \\
    \vdots & \vdots & \vdots & \vdots & \vdots & \ddots
\end{bmatrix}
\end{align}

Here, $\mathbf{A}_k, \mathbf{B}_k$ are matrices constructed from $\mathbf{H}_1, \mathbf{H}_2$, with additional padded zeros. Specifically, the matrix $\mathbf{A}_1 = \left(\ket{0}\bra{0} \right)^{\otimes k-1} \otimes -i \mathbf{H}_1 $, and $\mathbf{A}_2 =\left(\ket{0}\bra{0} \right)^{\otimes k-2} \otimes -i (\mathbf{H}_1 \oplus^1 \mathbf{H}_1 )$, and in general, $\mathbf{A}_i =\left(\ket{0}\bra{0} \right)^{\otimes k-i} \otimes -i (\mathbf{H}_1 \oplus^{i-1} \mathbf{H}_1 )$. The padded matrices $\mathbf{B}_k$ are given by:
$\mathbf{B}_1 = \left(\ket{0}\bra{0} \right)^{\otimes k-2} \otimes (\ket{0} \otimes \mathbf{H}_2)$,  and $\mathbf{B}_2  = \left(\ket{0}\bra{0} \right)^{\otimes k-3} \otimes \left(\ket{0} \otimes [\mathbf{H}_2 \oplus^1 \mathbf{H}_2] \right)$, and in general, $\mathbf{B}_i = \left(\ket{0}\bra{0} \right)^{\otimes k-i} \otimes \left(\ket{0} \otimes [\mathbf{H}_2 \oplus^{i-2} \mathbf{H}_2] \right) $.

\subsection{Error from Carlemann Linearization}
The convergence conditions for Carleman linearization, as outlined in \cite{forets2017explicit}, are summarized in Table~\ref{tab:carlemann_truncation}. In our setting, the logarithmic norm $\mu(-i\mathbf{H}_1) = 0$, which implies that convergence is restricted to short time horizons. Specifically, the time horizon scales as $t \in \mathcal{O}(1/\|\mathbf{H}_2\|)$.
In this work, we derive a qualitatively similar error bound for the nonlinear Schrödinger equation. However, the bound presented here applies to the entire truncated Carleman system, rather than being limited to just the first Carleman block. 

\begin{theorem}
\label{Thm: Carlemann_Truncation_Order}
    Consider the nonlinear Schr\"{o}dinger Equation of the form: 
    \begin{equation}
         \dot{\ket{\bm{\psi}} } = -i\mathbf{H}_1 \ket{\bm{\psi}} + \mathbf{H}_2 \ket{\bm{\psi}} \otimes \ket{\bm{\psi}},
    \end{equation} where $\mathbf{H}_1 \in \mathbb{C}^{N \times N}$ is a Hermitian matrix  and $\mathbf{H}_2 \in \mathbb{C}^{N \times N^2}$ is a rectangular matrix. Let $d = \max (d_r,d_c)$ where $d_r$ is the maximum number of nonzero elements in any row of the matrix $\mathbf{H}_2$, and $d_c$ is the maximum number of nonzero elements in any coloumn of the matrix $\mathbf{H}_2$. Let $k$ be the truncation order of the Carlemann Linearized system of differential equations. Let $\braket{\bm{\psi}(0)|\bm{\psi}(0)} \leq \beta $ and let, $ \Re(\bra{\bm{\psi}(t) }\mathbf{H}_2 \ket{\bm{\psi}(t)} \otimes \ket{\bm{\psi}(t)}) \leq 0$ , $ \forall t \in [0,t_s]$, where $t_s \in \mathcal{O}(1/d\max_{ij}(|\mathbf{H}_2(i,j)|))$. Then, for a desired error $\epsilon$, the Carlemann Truncation order is lower bounded by:

\begin{align}
 k\geq   \frac{
    \log\!\Bigl(\frac{\|\mathbf H_{2}\|}{2\,d\,\bigl(\max_{i,j}|\mathbf H_{2}(i,j)|\bigr)\,\epsilon}\Bigr)
}{
    \log\!\Bigl(\frac1\beta\Bigr)\;-\;2\,d\,\bigl(\max_{i,j}|\mathbf H_{2}(i,j)|\bigr)\,t
}
\end{align}

    \begin{proof}
        Let the truncated system of differential equations after Carlemann Linearization be of the form:
        
\begin{align}
    \label{Carlemann_Linearized_EOM_truncated}
    \frac{d}{dt}
    \begin{bmatrix}
        \widetilde{\bm{\ket{\psi}}} \\
        \widetilde{\bm{\ket{\psi}}}^{\otimes 2} \\
        \widetilde{\bm{\ket{\psi}}}^{\otimes 3}\\
        \vdots \\
        \widetilde{\bm{\ket{\psi}}}^{\otimes k} \\
    \end{bmatrix} 
    =
    -i 
    \begin{bmatrix}
        \mathbf{H}_1 & i\mathbf{H}_2 & \mathbf{0} & \mathbf{0} & \cdots \\
        \mathbf{0} & \mathbf{H}_1 \oplus^1 \mathbf{H}_1 & i\mathbf{H}_2 \oplus^1 \mathbf{H}_2 & \mathbf{0} & \cdots \\
        \mathbf{0} & \mathbf{0} & \mathbf{H}_1 \oplus^2 \mathbf{H}_1 & i\mathbf{H}_2 \oplus^2 \mathbf{H}_2 & \cdots \\
        \mathbf{0} & \mathbf{0} & \mathbf{0} & \ddots & \vdots \\
        \vdots & \vdots & \vdots & \vdots & \mathbf{H}_1 \oplus^{k-1} \mathbf{H}_1 \\
    \end{bmatrix}
    \begin{bmatrix}
        \widetilde{\bm{\ket{\psi}}} \\
        \widetilde{\bm{\ket{\psi}}}^{\otimes 2} \\
        \widetilde{\bm{\ket{\psi}}}^{\otimes 3}\\
        \vdots \\
        \widetilde{\bm{\ket{\psi}}}^{\otimes k} \\
    \end{bmatrix} 
\end{align}

Comparing Eq \eqref{Carlemann_Linearized_EOM_truncated} and Eq \eqref{Carlemann_Linearized_EOM}, the dynamics of error vector can be written as:

\begin{align}
    \label{Carlemann_Linearized_EOM_truncated_error}
    \frac{d}{dt}
    \begin{bmatrix}
        \mathcal{E}_1 \\
        \mathcal{E}_2 \\
        \mathcal{E}_3\\
        \vdots \\
        \mathcal{E}_k \\
    \end{bmatrix} 
    =&
    -i 
    \begin{bmatrix}
        \mathbf{H}_1 & i\mathbf{H}_2 & \mathbf{0} & \mathbf{0} & \cdots \\
        \mathbf{0} & \mathbf{H}_1 \oplus^1 \mathbf{H}_1 & i\mathbf{H}_2 \oplus^1 \mathbf{H}_2 & \mathbf{0} & \cdots \\
        \mathbf{0} & \mathbf{0} & \mathbf{H}_1 \oplus^2 \mathbf{H}_1 & i\mathbf{H}_2 \oplus^2 \mathbf{H}_2 & \cdots \\
        \mathbf{0} & \mathbf{0} & \mathbf{0} & \ddots & \vdots \\
        \vdots & \vdots & \vdots & \vdots & \mathbf{H}_1 \oplus^{k-1} \mathbf{H}_1 \\
    \end{bmatrix}
     \begin{bmatrix}
        \mathcal{E}_1 \\
        \mathcal{E}_2 \\
        \mathcal{E}_3\\
        \vdots \\
        \mathcal{E}_k \\
    \end{bmatrix} 
    \nonumber\\ &+ 
    \begin{bmatrix}
        \mathbf{0} \\
        \mathbf{0}\\
        \mathbf{0}\\
        \vdots \\
        -\mathbf{H}_2 \oplus^{k-1} \mathbf{H}_2  \bm{\ket{\psi}}^{\otimes k + 1}\\
    \end{bmatrix} 
\end{align}
here, $\mathcal{E}_j = \widetilde{\bm{\ket{\psi}}}^{\otimes j} -  \bm{\ket{\psi}}^{\otimes j} $.  Next this equation for the evolution of the error can be compactly re-written, for  $\mathbf{\mathcal{E}} = [\mathcal{E}_1,...\mathcal{E}_k]^T $, as

\begin{equation}
    \dot{\mathbf{\mathcal{E}}} = \mathbf{A}_k\mathbf{\mathcal{E}} + \mathbf{b}(t)
\end{equation}
Here, $\mathbf{A}_k$ is the truncated Carlemann Matrix.  Since $\mathbf{\mathcal{E}}(0) = \mathbf{0}$,
\begin{equation}
    \mathbf{\mathcal{E}}(t) = \int_0^t e^{\mathbf{A}_k (t-s)} b(s) ds
\end{equation}
\begin{equation}
    \|  \mathbf{\mathcal{E}}(t)\| \leq \int_0^t \|e^{\mathbf{A}_k (t-s)} \| \| b(s) \| ds
\end{equation}
Now, we apply the exponential property of the logarithmic norm,
\begin{equation}
     \|e^{\mathbf{A}_k (t-s)} \|  \leq e^{\mu(\mathbf{A}_k)(t-s)}
\end{equation}
For the induced 2 norm, \\
\begin{equation}
    \mu(\mathbf{A}_k) = \lambda_{\max}( A_k + A_k^{\dagger})
\end{equation}
Since $R_k = A_k + A_k^{\dagger}$ is a Hermitian hollow matrix, the sum of eigenvalues is equal to zero, and this is possible only if either the eigenvalues are all zeros or if the set of all eigenvalues contains positive and corresponding negative elements. The eigenvalues are all zero only if $A_k + A_k^{\dagger}$ is a zero matrix, and hence, in general, the matrix $A_k + A_k^{\dagger}$ is indefinite. We can apply the Gershgorin circle theorem to bound the maximum eigenvalue of the matrix $A_k + A_k^{\dagger}$. Since all the diagonal elements are zero, 
the maximum eigenvalue is given by:
\begin{align}
\lambda_{\max}( A_k + A_k^{\dagger}) = \max_i ( \sum_{j \neq i} |R_k(i,j)| ) \leq k( \|\mathbf{H}_2 \|_{\infty} + \|\mathbf{H}_2\|_{1}) 
\end{align}
Let $d = \max (d_r,d_c)$ where $d_r$ is the maximum number of nonzero elements in any row of the matrix $\mathbf{H}_2$, and $d_c$ is the maximum number of nonzero elements in any coloumn of the matrix $\mathbf{H}_2$. It is straightforward to show that $\|\mathbf{H}_2 \|_{\infty} \leq d_r \max_{ij}(|\mathbf{H}_2(i,j)|) $, and $\|\mathbf{H}_2 \|_{1} \leq d_c \max_{ij}(|\mathbf{H}_2(i,j)|) $. Hence, $\lambda_{\max}( A_k + A_k^{\dagger}) \leq  2kd \max_{ij}(|\mathbf{H}_2(i,j)|) $. From \autoref{lemma: norm_nonlinear_Schr\"{o}dinger_eq}, $ \|b(s) \| \leq k\| \mathbf{H}_2\|(\braket{\psi(0)|\psi(0)} )^k$. Hence, $ \|  \mathbf{\mathcal{E}}(t)\|$ is bounded by:
\begin{align*}
    \|  \mathbf{\mathcal{E}}(t)\| &\leq k \beta^k \| \mathbf{H}_2\|\int_0^t e^{2kd \max_{ij}(|\mathbf{H}_2(i,j)|)(t-s) } ds  \\ \nonumber 
&\text{Let } a = 2kd \max_{ij}(|\mathbf{H}_2(i,j)|)   \text{and } \beta = \braket{\psi(0)|\psi(0)} \\
&= k \beta^k \|\mathbf{H}_2\|\int_0^t e^{a(t-s) } ds \\
&= k \beta^k\|\mathbf{H}_2\|\int_0^t e^{at - as } ds \\
&= k \beta^k\|\mathbf{H}_2\|e^{at}\int_0^t e^{-sa} ds \\
&= k \beta^k\|\mathbf{H}_2\|e^{at} \left(\frac{e^{-at} - 1}{-a}\right) \\
&= k\beta^k\|\mathbf{H}_2\|\frac{(e^{at} - 1)}{a} \\
&= \beta^k\|\mathbf{H}_2\|\frac{ e^{2kd \max_{ij}(|\mathbf{H}_2(i,j)|)t} -1 }{2d\max_{ij}(|\mathbf{H}_2(i,j)|)}
\end{align*}

for $t>>1$, 
\begin{equation}
     \|  \mathbf{\mathcal{E}}(t)\| \leq  \|\mathbf{H}_2\|\frac{ (e^{2d \max_{ij}(|\mathbf{H}_2(i,j)|)t} \beta)^k }{2d\max_{ij}(|\mathbf{H}_2(i,j)|)}
\end{equation}
This limit converges only when $(e^{2d \max_{ij}(|\mathbf{H}_2(i,j)|)t} \beta) < 1$. This can happen when $t \in \mathcal{O}(1/d \max_{ij}(|\mathbf{H}_2(i,j)|)), \beta = \mathcal{O}(1)$ or when $\beta < e^{-2d \max_{ij}(|\mathbf{H}_2(i,j)|)t }$. 
When this condition is satisfied, the truncation order $k$ to achieve the desired error $\epsilon$ is given by:
\begin{align}
    k &\geq 
\frac{
    \log\!\Bigl(\frac{\|\mathbf H_{2}\|}{2\,d\,\bigl(\max_{i,j}|\mathbf H_{2}(i,j)|\bigr)\,\epsilon}\Bigr)
}{
    \log\!\Bigl(\frac1\beta\Bigr)\;-\;2\,d\,\bigl(\max_{i,j}|\mathbf H_{2}(i,j)|\bigr)\,t
}
\end{align}

\end{proof}
\end{theorem}


\subsection{Symmetrization}

Note that $\mathbf{Q}$ is not Hermitian, and hence we cannot implement $e^{-i\mathbf{Q}t}$ using Hamiltonian Simulation methods. Therefore, we perform additional transformations to symmetrize the matrix $\mathbf{Q}$. 
Our main idea here is to define a modified system of differential equations $ \frac{d \ket{\mathbf{\hat{p}}}}{dt} = -i\mathbf{\hat{Q}(\eta)}\ket{\mathbf{\hat{p}}} $ such that $\mathbf{\hat{Q}(\eta)}$ is a Hermitian matrix and, as $\eta \geq \eta_0$, $\|\mathbf{D}\ket{\mathbf{\hat{p}}(t)} -  \ket{\mathbf{p}(t)}\| \leq \epsilon$. Here, $\mathbf{D} = \operatorname{diag}(\eta^{k-1} \mathbf{I}, \eta^{k-2} \mathbf{I}, \ldots, \mathbf{I})$ is a diagonal matrix. We define $\ket{\mathbf{\hat{p}}} = \sum_{i=1}^k \ket{i-1}\otimes \ket{\mathbf{\hat{w}}_i}$ where $\ket{\mathbf{\hat{w}}_i} = \frac{\ket{\mathbf{w}_i}}{\eta^{k-i}}$ and the matrix $\mathbf{\hat{Q}(\eta)}$ is defined as:

\begin{align}
    \mathbf{\hat{Q}(\eta)} = 
     \begin{bmatrix}
                                \mathbf{A_1} & \frac{\mathbf{B}_1}{\eta} & \mathbf{0} & \mathbf{0} & \cdots \\
                                \frac{\mathbf{B}_1^{\dagger}}{\eta} & \mathbf{A_2} &  \frac{\mathbf{B}_2}{\eta} & \mathbf{0} & \cdots \\
                                \mathbf{0} & \frac{\mathbf{B}_2^{\dagger}}{\eta} & \mathbf{A_3} & \frac{\mathbf{B}_3}{\eta} & \cdots \\
                                \mathbf{0} & \mathbf{0} & \frac{\mathbf{B}_3^{\dagger}}{\eta} & \mathbf{A_4} & \cdots \\
                                \vdots & \vdots & \vdots & \vdots & \ddots
                                \end{bmatrix}
\end{align}
More compactly,  the truncated $\mathbf{\hat{Q}}(\eta)$ can be written as:
\begin{equation}
   \mathbf{\hat{Q}(\eta)} =  \left( \sum_{i=0}^{k-1} \ket{i} \bra{i} \otimes \mathbf{A}_{i+1} + \ket{i} \bra{i+1} \otimes \frac{\mathbf{B}_{i+1}}{\eta} + \ket{i+1} \bra{i} \otimes \frac{\mathbf{B}_{i+1}^{\dagger}}{\eta} \right) ,
\end{equation}
with $\mathbf{B}_{k} = \mathbf{0}$.

To understand how this transformation helps symmetrize and solve the original ode, consider truncation at $k=2$.\\
Here, the original ode is:
\begin{align}
\label{eq_second_order_carlemann}
\frac{d}{dt}
\begin{bmatrix}
        \mathbf{w}_1\\
      \mathbf{w}_2\\
\end{bmatrix} = -i\begin{bmatrix}
       {\mathbf{A_1}} & \mathbf{B_1} \\
       \mathbf{0} & \mathbf{A_2}
    \end{bmatrix}
    \begin{bmatrix}
       \mathbf{w}_1\\
      \mathbf{w}_2\\  
    \end{bmatrix} 
\end{align}
The symmetrized system of ode is defined as:
\begin{align}
\label{eq:symmetrized_2dof_ode}
\frac{d}{dt}
\begin{bmatrix}
        \frac{\mathbf{w_1}}{\eta}\\
       \mathbf{w_2}\\
\end{bmatrix} = -i\begin{bmatrix}
       \mathbf{A_1} & \frac{\mathbf{B_1}}{\eta} \\
       \frac{\mathbf{B_1}^{\dagger}}{\eta} & \mathbf{A_2} 
    \end{bmatrix}
    \begin{bmatrix}
        \frac{\mathbf{w_1}}{\eta}\\
       \mathbf{w_2}\\
\end{bmatrix}
\end{align}

The system of odes now can be written as:
\begin{align}
    \label{eq:symmetrized_2dof_ode_1}
   \frac{1}{\eta}  \frac{d}{dt} \mathbf{w_1}= -i\frac{1}{\eta}(\mathbf{A_1} \mathbf{w_1} + \mathbf{B_1} \mathbf{w_2})
\end{align}
\begin{align}
    \label{eq:symmetrized_2dof_ode_2}
    \frac{d}{dt} \mathbf{w_2} = \frac{\mathbf{B_1^{\dagger} }\mathbf{w_1}}{\eta^2} + \mathbf{A_2} \mathbf{w_2}
\end{align}
As $\eta \rightarrow \infty$, Eq \eqref{eq:symmetrized_2dof_ode_1} converges to :

\begin{align}
\lim_{\eta \rightarrow \infty}   \frac{1}{\eta}  \frac{d}{dt} \mathbf{w_1}= -i\frac{1}{\eta}(\mathbf{A_1} \mathbf{w_1} + \mathbf{B_1} \mathbf{w_2}) \equiv \frac{d}{dt} \mathbf{w_1}= -i(\mathbf{A_1} \mathbf{w_1} + \mathbf{B_1} \mathbf{w_2})
\end{align}

and Eq \eqref{eq:symmetrized_2dof_ode_2} converges to:

\begin{align}
\lim_{\eta \rightarrow \infty}   \frac{d}{dt} \mathbf{w_2} = \frac{\mathbf{B_1^{\dagger} }\mathbf{w_1}}{\eta^2} + \mathbf{A_2} \mathbf{w_2} \equiv \frac{d}{dt} \mathbf{w_2} = \mathbf{A_2} \mathbf{w_2} 
\end{align}
Now, we formalize this intuition  and apply perturbation theory to derive sufficient value for $\eta$ such that $\| \mathbf{D}\ket{\mathbf{\hat{p}}(t)} -  \ket{\mathbf{p}(t)}\| \leq \epsilon$, $\forall \epsilon > 0$. 
\begin{lemma}
\label{lemma: main_eta_bound}
Let $\ket{\mathbf{p}(t)} = \sum_i \ket{i-1}\otimes \ket{\mathbf{w}_i}$, where $\ket{\mathbf{w}_i(t)}  = \bm{\ket{0}}^{\otimes k-i}\bm{\ket{\psi(t)}}^{\otimes i} $, let $\beta = \braket{\psi(0)|\psi(0)}$ and let $\ket{ \dot{\mathbf{p}}(t)} = -i\mathbf{Q}\ket{\mathbf{p}(t)}$. Here, $\mathbf{Q} =  \left( \sum_{i=0}^{k-1} \ket{i} \bra{i} \otimes \mathbf{A}_{i+1} + \ket{i} \bra{i+1} \otimes {\mathbf{B}_{i+1}} \right) $ is a non-Hermitian matrix containing Hermitian block matrices $\mathbf{A}_i$ on the diagonal. Let $\ket{\hat{\mathbf{p}}(t)} = \sum_i \ket{i-1}\otimes \ket{\hat{\mathbf{w}}_i}$, where $\ket{\hat{\mathbf{w}}_i} = \frac{1}{\eta^{k-i}}\ket{{\mathbf{w}}_i}$ and let $\ket{ \dot{ \hat{\mathbf{p}}}(t)} = -i\mathbf{\hat{Q}}(\eta)\ket{\hat{\mathbf{p}}(t)}$, where 
\[\mathbf{\hat{Q}(\eta)} =  \left( \sum_{i=0}^{k-1} \ket{i} \bra{i} \otimes \mathbf{A}_{i+1} + \ket{i} \bra{i+1} \otimes \frac{\mathbf{B}_{i+1}}{\eta} + \ket{i} \bra{i-1} \otimes \frac{\mathbf{B}_{i}^{\dagger}}{\eta} \right) \]
Then, there for all values of $\eta \geq \sqrt{\|\hat{\mathbf{H}} \| \left(1 + \beta\frac{1-\beta^k}{(1-\beta)\epsilon} \right)t_s}$, $\| \mathbf{D}\ket{\mathbf{\hat{p}}(t)} -  \ket{\mathbf{p}(t)}\| \leq \epsilon$, $\forall \epsilon > 0, t \in [0,t_s]$. Here, $\mathbf{\hat{H}} = \left( \sum_{i=0}^{k-1} \ket{i}  \ket{i} \bra{i+1} \otimes {\mathbf{B}_{i+1}} \right)  $, and $\mathbf{D} = \operatorname{diag}(\eta^{k-1} \mathbf{I}, \eta^{k-2} \mathbf{I}, \ldots, \mathbf{I}) $.

\end{lemma}

\begin{proof}
Let $\mathbf{B}_0 =\mathbf{B}_k =  \mathbf{0}$. Now, let us write the equation of dynamics of $\ket{ \dot{ \hat{\mathbf{p}}}(t)} $. 
    \begin{align}
        \ket{ \dot{ \hat{\mathbf{p}}}(t)} &= \sum_j \ket{j-1}\otimes \ket{ \dot{\hat{\mathbf{w}}}_j } = -i\sum_j \ket{j-1} \otimes \left( \mathbf{A}_j \ket{\hat{\mathbf{w}}_j} + \frac{1}{\eta}\mathbf{B}_j \ket{\hat{\mathbf{w}}_{j+1}} + \frac{1}{\eta}\mathbf{B}_{j-1}^{\dagger} \ket{\hat{\mathbf{w}}_{j-1}} \right)  \nonumber \\
        &= \sum_j \ket{j-1}\otimes \frac{\ket{\dot{{\mathbf{w}}}_j}}{\eta^{k-j}} =  -i\sum_j \ket{j-1} \otimes \left( \frac{\mathbf{A}_j \ket{\mathbf{w}_j}}{\eta^{k-j}}+ \frac{1}{\eta}\frac{\mathbf{B}_j \ket{\mathbf{w}_{j+1}}}{\eta^{k-j-1}} + \frac{1}{\eta}\mathbf{B}_{j-1}^{\dagger} \frac{\ket{ {\mathbf{w}}_{j-1} }}{ \eta^{k-j+1}} \right) \nonumber \\ 
        &= \sum_j \frac{1}{\eta^{k-j}}\ket{j-1}\otimes {\ket{\dot{{\mathbf{w}}}_j}} =  -i\sum_j \ket{j-1} \otimes \left( \frac{\mathbf{A}_j \ket{\mathbf{w}_j}}{\eta^{k-j}}+ \frac{1}{\eta}\frac{\mathbf{B}_j \ket{\mathbf{w}_{j+1}}}{\eta^{k-j-1}} + \frac{1}{\eta}\mathbf{B}_{j-1}^{\dagger} \frac{\ket{ {\mathbf{w}}_{j-1} }}{ \eta^{k-j+1}} \right) \nonumber \\ 
        &= \sum_j \frac{1}{\eta^{k-j}}\ket{j-1}\otimes {\ket{\dot{{\mathbf{w}}}_j}} =  -i\sum_j \frac{1}{\eta^{k-j}} \ket{j-1} \otimes \left( {\mathbf{A}_j \ket{\mathbf{w}_j}}+ {\mathbf{B}_j \ket{\mathbf{w}_{j+1}}} + \frac{1}{\eta^2}\mathbf{B}_{j-1}^{\dagger} {\ket{ {\mathbf{w}}_{j-1} }} \right) \\ \nonumber
    \end{align}
After factoring out $\eta^{k-j}$, we can see that the time evolution of $\ket{{{\mathbf{w}}}_j(t)}$ obeys:
\begin{align}
\label{eq: symmetrized_ode_after_simplification}
     \sum_j \ket{j-1}\otimes {\ket{\dot{{\mathbf{w}}}_j}} &=  -i\sum_j  \ket{j-1} \otimes \left( {\mathbf{A}_j \ket{\mathbf{w}_j}}+ {\mathbf{B}_j \ket{\mathbf{w}_{j+1}}} + \frac{1}{\eta^2}\mathbf{B}_{j-1}^{\dagger} {\ket{ {\mathbf{w}}_{j-1} }} \right) \\ \nonumber 
     \sum_j \ket{j-1}\otimes {\ket{\dot{{\mathbf{w}}}_j}} &= -i\sum_l \left( \ket{l}\bra{l} \otimes {\mathbf{A}_{l+1} }+ \ket{l}\bra{l+1} \otimes{\mathbf{B}_{l+1} } + \frac{1}{\eta^2}\ket{l}\bra{l-1} \otimes  \mathbf{B}_l^{\dagger} \right) \sum_j \ket{j-1}\otimes {\ket{{{\mathbf{w}}}_j}}\\ \nonumber 
\end{align}
Or, more compactly, Eq.~\eqref{eq: symmetrized_ode_after_simplification} can be written as:
\begin{equation}
    \ket{ \dot{\widetilde{\mathbf{p}}(t)}} = (-i\mathbf{H} + \frac{1}{\eta^2}\hat{\mathbf{H}} )\ket{{\widetilde{\mathbf{p}}(t)}} 
\end{equation}
where, $\ket{{\widetilde{\mathbf{p}}(t)}} =  \sum_j \ket{j-1}\otimes {\ket{\dot{{\mathbf{w}}}_j}}$, $\mathbf{H} = \sum_l \left( \ket{l}\bra{l} \otimes {\mathbf{A}_{l+1} }+ \ket{l}\bra{l+1} \otimes{\mathbf{B}_{l+1} } \right)$, and $\hat{\mathbf{H}} = \sum_l \ket{l}\bra{l-1} \otimes  \mathbf{B}_l^{\dagger} $. 
Applying \autoref{eta_suff_val_lemma}, for any $ \eta \geq \sqrt{\|\hat{\mathbf{H}} \| \left(1 + \frac{ \|{\mathbf{p}}(0) \|}{\epsilon} \right)t_s}$, $\| \mathbf{p}(t) - \hat{\mathbf{p}}(t) \| \leq \epsilon$, $\forall t \in [0,t_s]$. Since $\mathbf{D}\ket{\hat{\mathbf{p}}(t)} = \ket{{\widetilde{\mathbf{p}}(t)}}$, this implies, $\| \mathbf{D}\ket{\mathbf{\hat{p}}(t)} -  \ket{\mathbf{p}(t)}\| \leq \epsilon \text{,}\forall \eta \geq \sqrt{\|\hat{\mathbf{H}} \| \left(1 + \frac{ \|{\mathbf{p}}(0) \|}{\epsilon} \right)t_s}$. Recall that $\ket{\mathbf{p}(0)}  = \sum_i \ket{i-1} \otimes \ket{\mathbf{w}_i(0)} $, where $\ket{\mathbf{w}_i(0)}  = \bm{\ket{0}}^{\otimes k-i}\bm{\ket{\psi(0)}}^{\otimes i} $. Hence, $\braket{\mathbf{p}(0)|\mathbf{p}(0)} = \sum_i \braket{\psi(0)|\psi(0)}^i = k$, if $\braket{\psi(0)|\psi(0)} = 1$. If $\braket{\psi(0)|\psi(0)} = \beta$, then $\braket{\mathbf{p}(0)|\mathbf{p}(0)} = \beta\frac{1-\beta^k}{1 - \beta} $. Hence, $\eta \geq \sqrt{\|\hat{\mathbf{H}} \| \left(1 + \beta\frac{1-\beta^k}{(1-\beta)\epsilon} \right)t_s}$
\end{proof}

\begin{lemma}
\label{eta_suff_val_lemma}
    Let $\mathbf{H}$ be a block diagonal matrix consisting of Hermitian sub blocks, $\mathbf{H}_{ii} = A_i$, and let $\mathbf{\hat{H}}$ be a sparse, non-symmetric matrix. Let $ \dot{\hat{\mathbf{p}}}(t) = (-i\mathbf{H} + \frac{1}{\eta^2} \hat{\mathbf{H}})\hat{\mathbf{p}}(t)$ and $\dot{{\mathbf{p}}}(t) = -i\mathbf{H}{\mathbf{p}}(t)$ be a set of two differential equations. For any $ \eta \geq \sqrt{\|\hat{\mathbf{H}} \| \left(1 + \frac{ \|{\mathbf{p}}(0) \|}{\epsilon} \right)t_s}$, $\| \mathbf{p}(t) - \hat{\mathbf{p}}(t) \| \leq \epsilon$, $\forall t \in [0,t_s]$. 
\end{lemma}

\begin{proof}
Consider
\begin{equation}
\label{perturbed_ode}
    \dot{\hat{\mathbf{p}}} = (-i\mathbf{H} + \lambda \hat{\mathbf{H}})\hat{\mathbf{p}}
\end{equation}
  The solution to Eq.~\eqref{perturbed_ode} can be written as:
  \begin{equation}
      \hat{\mathbf{p}}(t) =  \hat{\mathbf{p}}_0(t) + \lambda \hat{\mathbf{p}}_1(t) + \lambda^2 \hat{\mathbf{p}}_2(t) + ..
  \end{equation}
  and, 
   \begin{equation}
     \| \hat{\mathbf{p}}(t) - \hat{\mathbf{p}}_0(t) \| \leq \lambda \| \hat{\mathbf{p}}_1(t) \| +  \lambda^2  \| \hat{\mathbf{p}}_2(t)\| + ..
  \end{equation}
  \begin{equation}
   \frac{d}{dt} \hat{\mathbf{p}}_0(t) = -i \mathbf{H} \hat{\mathbf{p}}_0(t)
  \end{equation}

\begin{equation}
   \frac{d}{dt} \hat{\mathbf{p}}_1(t) = -i \mathbf{H} \hat{\mathbf{p}}_1(t) + \hat{\mathbf{H}} \hat{\mathbf{p}}_0(t)
\end{equation}

\begin{equation}
   \frac{d}{dt} \hat{\mathbf{p}}_k(t) = -i \mathbf{H} \hat{\mathbf{p}}_k(t) + \hat{\mathbf{H}} \hat{\mathbf{p}}_{k-1}(t)
\end{equation}

\begin{equation}
    \| \hat{\mathbf{p}}_0(t) \| \leq \| e^{-i\mathbf{H}t}\hat{\mathbf{p}}_0(0) \| \leq \| \hat{\mathbf{p}}_0(0)\|
\end{equation}
  Now, we upper bound $\| \hat{\mathbf{p}}_1(t) \|$
  \begin{equation}
    \hat{\mathbf{p}}_1(t) = e^{-i \mathbf{H}t}\hat{\mathbf{p}}_1(0) + e^{-i \mathbf{H}t} \int_0^T   e^{i \mathbf{H}t} \hat{\mathbf{H}} e^{-i \mathbf{H}t} \hat{\mathbf{p}}_0(0)dt
  \end{equation}
  \begin{equation}
      \|  \hat{\mathbf{p}}_1(t)\| \leq \| \hat{\mathbf{p}}_1(0)  \| + \| \hat{\mathbf{H}} \hat{\mathbf{p}}_0(0)\|t
  \end{equation}
  Which implies,
   \begin{equation}
      \|  \hat{\mathbf{p}}_1(t)\| \leq \| \hat{\mathbf{H}}\| \|\hat{\mathbf{p}}_0(0)\|t
  \end{equation}
 Similarly, it can be shown that:
  \begin{equation}
      \|  \hat{\mathbf{p}}_2(t)\| \leq \| \hat{\mathbf{H}}\|^2 \| \hat{\mathbf{p}}_0(0)\|t^2
  \end{equation}
and 
\begin{equation}
      \|  \hat{\mathbf{p}}_k(t)\| \leq  \| \hat{\mathbf{H}}\|^k \| \hat{\mathbf{p}}_0(0)\|t^k
\end{equation}

Hence, 
\begin{align}
     \| \hat{\mathbf{p}}(t) - \hat{\mathbf{p}}_0(t) \| \leq \| \hat{\mathbf{p}}_0(0) \| (\lambda \| \hat{\mathbf{H}} \|t +  \lambda^2 \| \hat{\mathbf{H}}\|^2 t^2 + .. ) \leq \epsilon
\end{align}
Hence, 
\begin{equation}
    \frac{1}{\lambda  } \geq  \| \hat{\mathbf{H}} \|(1 + \frac{\|\hat{\mathbf{p}}_0(0)\| }{\epsilon})t
\end{equation}
Hence, 
\begin{equation}
    \eta \geq \sqrt{\|\hat{\mathbf{H}} \|\left(1 + \frac{\|\hat{\mathbf{p}}_0(0)\|}{\epsilon}\right)t}
\end{equation}
which implies:
\begin{equation}
    \eta \geq \sqrt{\|\hat{\mathbf{H}} \|\left(1 + \frac{\|{\mathbf{p}}(0)\|}{\epsilon}\right)t}
\end{equation}

\end{proof}

\begin{lemma}
\label{upper_bound_H_hat}
    Upper bound for $\|\hat{\mathbf{H}}  \|$ for a given Carlemann truncation order $k$.
\end{lemma}

\begin{proof}

Recall that:\\
\begin{equation}
    \hat{\mathbf{H}} = \hat{\mathbf{H}}_1 + \hat{\mathbf{H}}_2 +..\hat{\mathbf{H}}_k
\end{equation}

Here,

\begin{align}
\hat{\mathbf{H}}_1 = 
i \ket{\mathbf{1}} \bra{\mathbf{0}} \otimes \mathbf{B}_1^{\dagger} 
\end{align}

, 

\begin{align}
\hat{\mathbf{H}}_2 = 
i \ket{\mathbf{2}} \bra{\mathbf{1}} \otimes \mathbf{B}_2^{\dagger} 
\end{align}


and $\hat{\mathbf{H}}_k = i \ket{\mathbf{k}} \bra{\mathbf{k-1}} \otimes \mathbf{B}_k^{\dagger}$.\\

\begin{align}
\|\hat{\mathbf{H}}\| &\leq \|\hat{\mathbf{H}}_1\| + \|\hat{\mathbf{H}}_2\| + \cdots \\
&\leq  \left( \|\mathbf{B}_1^{\dagger} \| + \|\mathbf{B}_2^{\dagger}\| + \cdots \|\mathbf{B}_k^{\dagger}\| \right)\\
& \leq \left( \|\mathbf{H}_2 \| + 2\|\mathbf{H}_2\| + \cdots k\|\mathbf{H}_2\| \right) \\
& \leq  \frac{\|\mathbf{H}_2 \| k\left( k +1 \right)}{2}  \\
\end{align}

\end{proof}
\subsection{Block Encodings}

In this section, we describe the construction of block encodings for the symmetrized Carlemann operator:
\[
\mathbf{Q}(\eta) =  \sum_{i=0}^{k-1} \left( \ket{i} \bra{i} \otimes \mathbf{A}_{i+1} + \ket{i} \bra{i+1} \otimes \frac{\mathbf{B}_{i+1}}{\eta} + \ket{i+1} \bra{i} \otimes \frac{\mathbf{B}_{i+1}^\dagger}{\eta} \right)
\]
To construct the block encoding for $\mathbf{Q}(\eta)$, we first construct block encodings for the submatrices \(\mathbf{A}_{i+1}\), \(\mathbf{B}_{i+1}\), and \(\mathbf{B}_{i+1}^\dagger\) that compose \(\mathbf{Q}(\eta)\). These submatrices depend on the matrices \(\mathbf{H}_1\) and \(\mathbf{H}_2\). Thus, as a preliminary step, we create block encodings for \(\mathbf{H}_1\) and \(\mathbf{H}_2\). Finally, we combine the submatrices using tensor products with the appropriate projection matrices to obtain the desired block encoding of \(\mathbf{Q}(\eta)\).
 
\begin{lemma}[Block Encoding $\mathbf{H}_1 \in \mathbb{C}^{N \times N}$]
\label{lemma:block_encoding_H1}
Let $\epsilon'>0$, and let $\alpha = d\max_{jk}\{|\mathbf{H}_1(j,k)|,|\mathbf{H}_2(j,k)|\}$, where $d$ is the maximum number of non zero elements in any row or column of the matrices $\mathbf{H}_1$ and $\mathbf{H}_2$. Assume we have Oracle access to elements in the matrix $\mathbf{H}_1$ as per \autoref{Def:Oracles-nonlinear_Schr\"{o}dinger_eq}. Then, there exists a unitary matrix, $\mathbf{\mathcal{U}}(\mathbf{H}_1)$ such that:
\begin{equation}
  \Big\|\mathbf{H}_1 -  \alpha(\bra{0}^{\otimes a} \otimes \mathbf{I}_N )\mathbf{\mathcal{U}}(\mathbf{H}_1)(\ket{0}^{\otimes a} \otimes \mathbf{I}_N)  
   \Big\| \leq \epsilon'
\end{equation} 

\end{lemma}

\begin{proof}
    We can construct the Block encoding by applying Lemma 48 from \cite{gilyen2019quantum}. Here, $a = \log(N) + 3$.
\end{proof}

\begin{lemma}[Block Encoding $\mathbf{H}_2 \in \mathbb{C}^{N \times N^2}$] 
\label{Block_Encoding_H_2}
Let $\epsilon'>0$, and let $\alpha = d\max_{jk}\{|\mathbf{H}_1(j,k)|,|\mathbf{H}_2(j,k)|\}$, where $d$ is the maximum number of non zero elements in any row or column of the matrices $\mathbf{H}_1$ and $\mathbf{H}_2$. Assume we have Oracle access to elements in the matrix $\mathbf{H}_1$ as per \autoref{Def:Oracles-nonlinear_Schr\"{o}dinger_eq}. Let $\mathbf{H}_2' = \ket{0}^{\otimes n} \otimes \mathbf{H}_2$. Then, there exist a  unitary matrix, $\mathbf{\mathcal{U}}(\mathbf{H}_2')$ such that:
\begin{equation}
  \Big\|\mathbf{H}_2' -  \alpha (\bra{0}^{\otimes a} \otimes \mathbf{I}_N )\mathbf{\mathcal{U}}(\mathbf{H}_2')(\ket{0}^{\otimes a} \otimes \mathbf{I}_N)  
   \Big\| \leq \epsilon'
\end{equation}    
\end{lemma}
\begin{proof}
     We can construct the Block encoding by applying Lemma 48 from \cite{gilyen2019quantum}. Here, $a = 2\log(N) + 3$.
\end{proof}

\begin{lemma}[Block Encoding $\mathbf{C}_j^i = \mathbf{I}^{\otimes i+1-j} \otimes \mathbf{H}_1 \otimes \mathbf{I}^{\otimes j-1}$] Let $\mathbf{\mathcal{U}}(\mathbf{H}_1) $ be $(\alpha,a,\epsilon')$ block encoding of $\mathbf{H}_1$ obtained from \autoref{lemma:block_encoding_H1}. Then, there exists an $(\alpha,a,\epsilon_1)$ block encoding of $\mathbf{C}_j^i$. 
    
\end{lemma}

\begin{proof}
     Let $\mathbf{\mathcal{U}}(\mathbf{H}_1)$ be an $(\alpha,a,\epsilon_1)$ block encoding of $\mathbf{H}_1$. Let $\mathcal{S}$ be a unitary circuit that performs the following swap operation:
     \[ \mathcal{S}(\ket{0}^{\otimes a} \otimes \mathbf{I}_N \otimes \mathbf{I}^{\otimes j-1} \otimes \mathbf{I}^{\otimes i+1-j}) = (\mathbf{I}^{\otimes i+1-j}) \otimes \ket{0}^{\otimes a} \otimes \mathbf{I}_N \otimes \mathbf{I}^{\otimes j-1} \]
     
     Then, consider: $ \mathcal{S}^{\dagger} [ \mathbf{I}^{\otimes i+1-j} \otimes \mathbf{\mathcal{U}}(\mathbf{H}_1) \otimes \mathbf{I}^{\otimes j-1}] \mathcal{S}$
     \begin{align}
     &(\bra{0}^{\otimes a} \otimes \mathbf{I}_N \otimes \mathbf{I}^{\otimes j-1} \otimes \mathbf{I}^{\otimes i+1-j}) \mathcal{S}^{\dagger} [\mathbf{I}^{\otimes i+1-j} \otimes  \mathbf{\mathcal{U}}(\mathbf{H}_1) 
      \otimes \mathbf{I}^{\otimes j-1}] \mathcal{S} (\ket{0}^{\otimes a} \otimes \mathbf{I}_N \otimes \mathbf{I}^{\otimes j-1} \otimes \mathbf{I}^{\otimes i+1-j}) \nonumber \\ \nonumber
      &\qquad=   (\bra{0}^{\otimes a} \otimes \mathbf{I}_N \otimes \mathbf{I}^{\otimes j-1} \otimes \mathbf{I}^{\otimes i+1-j}) \mathcal{S}^{\dagger} (\mathbf{I}^{\otimes i+1-j} \otimes \mathbf{\mathcal{U}}(\mathbf{H}_1)(\ket{0}^{\otimes a} \otimes \mathbf{I}_N) \otimes \mathbf{I}^{\otimes j-1} ) \\ \nonumber
     &\qquad=  (\mathbf{I}^{\otimes i+1-j} \otimes (\bra{0}^{\otimes a} \otimes \mathbf{I}_N) \mathbf{\mathcal{U}}(\mathbf{H}_1)(\ket{0}^{\otimes a} \otimes \mathbf{I}_N) \otimes \mathbf{I}^{\otimes j-1} ) \\ \nonumber
     &\qquad=     \frac{(\mathbf{I}^{\otimes i+1-j} \otimes \mathbf{H}_1 \otimes \mathbf{I}^{\otimes j-1}) }{\alpha}          \\ \nonumber
     \end{align}

\end{proof}

\begin{lemma}[Block Encoding $\mathbf{D}_j^i = \ket{0} \otimes \mathbf{I}^{\otimes i+1-j} \otimes \mathbf{H}_2 \otimes \mathbf{I}^{\otimes j-1}$]. Let $\mathbf{\mathcal{U}}(\mathbf{H}_2) $ be $(\alpha,a,\epsilon')$ block encoding of $\mathbf{H}_2$ obtained from \autoref{Block_Encoding_H_2}. Then, there exists an $(\alpha,a,\epsilon')$ block encoding of $\mathbf{D}_j^i$.
    
\end{lemma}

\begin{proof}
     Let $\mathbf{\mathcal{U}}(\mathbf{H}_2')$ be an $(\alpha,a,\epsilon_1)$ block encoding of $\ket{0} \otimes \mathbf{H}_2$ from \autoref{Block_Encoding_H_2}. Let $\mathcal{S}$ be a unitary circuit that performs the following swap operation:
     \[ \mathcal{S}(\ket{0}^{\otimes a} \otimes \mathbf{I}_N \otimes \mathbf{I}^{\otimes j-1} \otimes \mathbf{I}^{\otimes i+1-j}) = (\mathbf{I}^{\otimes i+1-j}) \otimes \ket{0}^{\otimes a} \otimes \mathbf{I}_N \otimes \mathbf{I}^{\otimes j-1} \]
     
     Then, consider: $ \mathbf{I}^{\otimes i+1-j} \otimes \mathbf{\mathcal{U}}(\mathbf{H}_2') \otimes \mathbf{I}^{\otimes j-1}$
     \begin{align}
     &(\bra{0}^{\otimes a} \otimes \mathbf{I}_N \otimes \mathbf{I}^{\otimes j-1} \otimes \mathbf{I}^{\otimes i+1-j}) \mathcal{S}^{\dagger} [\mathbf{I}^{\otimes i+1-j} \otimes  \mathbf{\mathcal{U}}(\mathbf{H}_2') 
      \otimes \mathbf{I}^{\otimes j-1}] \mathcal{S} (\ket{0}^{\otimes a} \otimes \mathbf{I}_N \otimes \mathbf{I}^{\otimes j-1} \otimes \mathbf{I}^{\otimes i+1-j}) \nonumber \\ \nonumber
      &\qquad=   (\bra{0}^{\otimes a} \otimes \mathbf{I}_N \otimes \mathbf{I}^{\otimes j-1} \otimes \mathbf{I}^{\otimes i+1-j}) \mathcal{S}^{\dagger} (\mathbf{I}^{\otimes i+1-j} \otimes \mathbf{\mathcal{U}}(\mathbf{H}_2')(\ket{0}^{\otimes a} \otimes \mathbf{I}_N) \otimes \mathbf{I}^{\otimes j-1} ) \\ \nonumber
     &\qquad=  (\mathbf{I}^{\otimes i+1-j} \otimes (\bra{0}^{\otimes a} \otimes \mathbf{I}_N) \mathbf{\mathcal{U}}(\mathbf{H}_2')(\ket{0}^{\otimes a} \otimes \mathbf{I}_N) \otimes \mathbf{I}^{\otimes j-1} ) \\ \nonumber
     &\qquad=     \frac{(\mathbf{I}^{\otimes i+1-j} \otimes \mathbf{H}_2' \otimes \mathbf{I}^{\otimes j-1}) }{\alpha}          \\ \nonumber
       &\qquad=     \frac{(\mathbf{I}^{\otimes i+1-j} \otimes \ket{0} \otimes \mathbf{H}_2 \otimes \mathbf{I}^{\otimes j-1}) }{\alpha}     \nonumber     \\ 
     \end{align}
     Finally, we apply a SWAP operation to get:
     \begin{equation}
         \frac{\text{SWAP}(\mathbf{I}^{\otimes i+1-j} \otimes \ket{0} \otimes \mathbf{H}_2 \otimes \mathbf{I}^{\otimes j-1}) }{\alpha} =  \frac{\ket{0}\otimes (\mathbf{I}^{\otimes i+1-j}  \otimes \mathbf{H}_2 \otimes \mathbf{I}^{\otimes j-1}) }{\alpha}
     \end{equation}

\end{proof}

\textcolor{black}{}
\begin{lemma}[Block Encoding $ \ket{0} \otimes (\oplus^i (\mathbf{H}_2))$] There exists a unitary block encoding of the superoperator $ \ket{0} \otimes (\oplus^i (\mathbf{H}_2))$, $\mathcal{U}(\ket{0} \otimes (\oplus^i (\mathbf{H}_2)) )$ such that $(\bra{0}^{\otimes a} \otimes \mathbf{I})\mathcal{U}(\ket{0} \otimes (\oplus^i (\mathbf{H}_2)) )(\ket{0}^{\otimes a} \otimes \mathbf{I}) = \frac{ \ket{0} \otimes (\oplus^i (\mathbf{H}_2))}{\alpha k} $
\end{lemma}
\begin{proof}
    Recall that $\oplus^i (\mathbf{H}_2) = \sum_{j=1}^{i+1} \mathbf{D}_j^i $, where the matrix $\mathbf{D}_j^i = \ket{0} \otimes \mathbf{I}^{\otimes i+1-j} \otimes \mathbf{H}_2 \otimes \mathbf{I}^{\otimes j-1}$. We can block encode the operator $\oplus^i (\mathbf{H}_2)$
     by applying LCU based techniques \cite{childs2012hamiltonian}, similar to the construction in \autoref{Block_encoding_oplus_H_1}. Specifically, we define the following PREP and SELECT oracles:
      \[
    \text{PREP}\ket{0}= \frac{1}{\sqrt{k}}\sum_{j=0}^{i} \ket{j}+ \sum_{i+1}^{k-1} \omega_{j-i} \ket{j}
    \]

Here, $\omega_{j-i}$ is $k-i$ th root of unity, and $\sum_{i+1}^{k-1} \omega_{j-i} =0$
\[
\text{SELECT}\ket{j}\ket{\psi} = 
 \left\{
\begin{aligned}
    &\ket{j}\mathbf{\mathcal{U}}(\ket{0} \otimes \mathbf{D}_{j+1}^i)\ket{\psi}, && \text{if } j \leq i, \\
    &\ket{j}I\ket{\psi}, && \text{if } j > i
\end{aligned}
\right.
\]

    From these unitaries, we can build the following quantum circuit to implement the linear combination of unitaries:
    \[
    \frac{\ket{0}\otimes( \oplus^i (\mathbf{H}_2))}{\alpha k} = (\bra{0}^{\otimes a} \otimes \mathbf{I})\text{PREP}^{\dagger} \cdot \text{SELECT} \cdot \text{PREP} (\ket{0}^{\otimes a} \otimes \mathbf{I})
    \]
     
\end{proof}

\begin{lemma}[Block Encoding $\oplus^i (\mathbf{H}_1)$]
  \label{Block_encoding_oplus_H_1}
   There exists a unitary block encoding of the superoperator $\oplus^i (\mathbf{H}_1)$, $\mathcal{U}(\oplus^i (\mathbf{H}_1))$ such that $(\bra{0}^{\otimes a} \otimes \mathbf{I})\mathcal{U}(\oplus^i (\mathbf{H}_1))(\ket{0}^{\otimes a} \otimes \mathbf{I}) = \frac{\oplus^i (\mathbf{H}_1)}{\alpha k} $
\end{lemma}
\begin{proof}
    Recall that $\oplus^i (\mathbf{H}_1) = \sum_{j=1}^{i+1} \mathbf{C}_j^i$, where the matrix $\mathbf{C}_j^i = \mathbf{I}^{\otimes i+1-j} \otimes \mathbf{H}_1 \otimes \mathbf{I}^{\otimes j-1}$. We can apply LCU-based techniques to block encode the sum $\sum_{j=1}^{i+1} \mathbf{C}_j^i$. Note that for different values of $i$, we sum $i+1$ of matrices, and this means an LCU circuit for directly implementing the sum $\oplus^i (\mathbf{H}_1)$ can result in a block encoding constant proportional to $i+1$. This is problematic, as we need a common factor $\forall \oplus^i (\mathbf{H}_1) $. Hence, we define the PREP, and SELECT unitaries, for implementing the linear combination of unitaries in the following manner: 
    \[
    \text{PREP}\ket{0}= \frac{1}{\sqrt{k}}\sum_{j=0}^{i} ( \ket{j}+ \sum_{j=i+1}^{k-1} \omega_{j-i} \ket{j})
    \]

Here, $\omega_{j-i}$ is $k-i$ th root of unity, and $\sum_{i+1}^{k-1} \omega_{j-i} =0$
\[
\text{SELECT}\ket{j}\ket{\psi} = 
 \left\{
\begin{aligned}
    &\ket{j}\mathbf{\mathcal{U}}(\mathbf{A}_{j+1}^i)\ket{\psi}, && \text{if } j \leq i, \\
    &\ket{j}I\ket{\psi}, && \text{if } j > i
\end{aligned}
\right.
\]

    From these unitaries, we can build the following quantum circuit to implement the linear combination of unitaries:
    \[
    \frac{\oplus^i (\mathbf{H}_1)}{\alpha k} = (\bra{0}^{\otimes a} \otimes \mathbf{I})\text{PREP}^{\dagger} \cdot \text{SELECT} \cdot \text{PREP} (\ket{0}^{\otimes a} \otimes \mathbf{I})
    \]
\end{proof}

\begin{lemma}[Block Encoding the matrix $ \mathbf{A}_i =\left(\ket{0}\bra{0} \right)^{\otimes k-i} \otimes -i (\oplus^{i-1} (\mathbf{H}_1 ))$] There exists a unitary block encoding for the matrix $\mathbf{A}_i$, $\mathcal{U}(\mathbf{A}_i)$ such that $  (\bra{0}^{\otimes a} \otimes \mathbf{I}) \mathcal{U}(\mathbf{A}_i) (\ket{0}^{\otimes a} \otimes \mathbf{I}) = \frac{\mathbf{A}_i}{ {2} \alpha k } $
\label{block_encoding_Ai}
\end{lemma}
\begin{proof}
    The matrix $\mathbf{A}_i$ can be written as a product of two matrices $\mathbf{E}_i = \mathbf{I}^{\otimes{k-i}} \otimes  (\oplus^{i-1}( -i\mathbf{H}_1 ))$ and $ \mathbf{F}_i = \left(\ket{0}\bra{0} \right)^{\otimes k-i} \otimes \mathbf{I}$. Hence, we can create a block encoding for the matrix $\mathbf{A}_i$ by constructing a block encoding for the matrix product $\mathbf{E}_i \mathbf{F}_i$. The block encoding for the matrix product can be constructed by applying Lemma 53 from \cite{gilyen2019quantum}. Specifically,  $ \mathcal{U}(\mathbf{E}_i \mathbf{F}_i) = (\mathbf{I} \otimes \mathcal{U}(\mathbf{E}_i) )(\mathbf{I} \otimes \mathcal{U}(\mathbf{F}_i))$. The block encoding for the matrix $\mathbf{E}_i$, $\mathcal{U}(\mathbf{E}_i)  = 
    \mathbf{I}^{\otimes{k-i}} \otimes \mathcal{U} ((\oplus^{i-1}( -i\mathbf{H}_1 )))$, where $\mathcal{U} ((\oplus^{i-1}( -i\mathbf{H}_1 )))$ can be constructed by applying \autoref{Block_encoding_oplus_H_1} and $\mathcal{U}(\mathbf{F}_i) =  \mathcal{U}(\left(\ket{0}\bra{0} \right)^{\otimes k-i} )\otimes \mathbf{I}$. The block encoding for the projector $(\ket{0}\bra{0})^{\otimes k-i} $, $\mathcal{U}(\left(\ket{0}\bra{0} \right)^{\otimes k-i})$ can be constructed as shown below:

    \begin{equation}
        \left(\ket{0}\bra{0} \right)^{\otimes k-i} = \mathbf{U}_1 + \mathbf{U}_2
    \end{equation}
    where, $\mathbf{U}_1$, and $\mathbf{U}_2$ are unitary matrices given by:
    \begin{equation}
        \mathbf{U}_1 = \left(\ket{0}\bra{0} \right)^{\otimes k-i} + \sum_{j \neq 0} \left(\ket{j}\bra{j} \right)^{\otimes k-i}
    \end{equation}
    \begin{equation}
        \mathbf{U}_2 = \left(\ket{0}\bra{0} \right)^{\otimes k-i}  -\sum_{j \neq 0} \left(\ket{j}\bra{j} \right)^{\otimes k-i}
    \end{equation}
    Applying LCU lemma, the block encoding for the projector is given by:
    \begin{equation}
    (\bra{0} \otimes \mathbf{I})\mathcal{U}(\left(\ket{0}\bra{0} \right)^{\otimes k-i})(\ket{0} \otimes \mathbf{I}) = \frac{1}{2}(\ket{0}\bra{0} )^{\otimes k-i}
    \end{equation}
    As per Lemma 53, \cite{gilyen2019quantum}, the block encoding constant for the matrix product $\mathbf{E}_i \mathbf{F}_i$ is the product of block encoding constants $\alpha k$ and $2$. Hence, the block encoding constant for $\mathbf{A}_i$ is $2 \alpha k$. Thus,  $  (\bra{0}^{\otimes a} \otimes \mathbf{I}) \mathcal{U}(\mathbf{A}_i) (\ket{0}^{\otimes a} \otimes \mathbf{I}) = \frac{\mathbf{A}_i}{ {2} \alpha k } $.

\end{proof}
\begin{lemma}[Block Encoding the matrix 
$\mathbf{B}_i = (\ket{0}\bra{0})^{\otimes k-i} \otimes (\ket{0} \otimes (\oplus^{i-2} \mathbf{H}_2)$] 
There exists a unitary block encoding for the matrix $\mathbf{B}_i$, $\mathcal{U}(\mathbf{B}_i)$ such that $  (\bra{0}^{\otimes a} \otimes \mathbf{I}) \mathcal{U}(\mathbf{B}_i) (\ket{0}^{\otimes a} \otimes \mathbf{I}) = \frac{\mathbf{B}_i}{ {2} \alpha k } $
\end{lemma}

\begin{proof}
    The matrix $\mathbf{B}_i$ can be written as a product of two matrices $\mathbf{M}_i = \mathbf{I}^{\otimes{k-i}} \otimes (\ket{0} \otimes (\oplus^{i-2} \mathbf{H}_2))$ and $ \mathbf{N}_i = \left(\ket{0}\bra{0} \right)^{\otimes k-i} \otimes \mathbf{I}$. Hence, we can create a block encoding for the matrix $\mathbf{B}_i$ by constructing a block encoding for the matrix product $\mathbf{M}_i \mathbf{N}_i$. The block encoding for the matrix product can be constructed by applying Lemma 53 from \cite{gilyen2019quantum}. Specifically,  $ \mathcal{U}(\mathbf{M}_i \mathbf{N}_i) = (\mathbf{I} \otimes \mathcal{U}(\mathbf{M}_i) )(\mathbf{I} \otimes \mathcal{U}(\mathbf{N}_i))$. The block encoding for the matrix $\mathbf{M}_i$, $\mathcal{U}(\mathbf{M}_i)  = 
    \mathbf{I}^{\otimes{k-i}} \otimes \mathcal{U} ((\oplus^{i-1}( -i\mathbf{H}_1 )))$, where $\mathcal{U} ((\oplus^{i-1}( -i\mathbf{H}_1 )))$ can be constructed by applying \autoref{Block_encoding_oplus_H_1} and $\mathcal{U}(\mathbf{N}_i) =  \mathcal{U}(\left(\ket{0}\bra{0} \right)^{\otimes k-i} )\otimes \mathbf{I}$.

    As per Lemma 53, \cite{gilyen2019quantum}, the block encoding constant for the matrix product $\mathbf{M}_i \mathbf{N}_i$ is the product of block encoding constants $\alpha k$ and $2$. Hence, the block encoding constant for $\mathbf{A}_i$ is $2 \alpha k$. Thus,  $  (\bra{0}^{\otimes a} \otimes \mathbf{I}) \mathcal{U}(\mathbf{B}_i) (\ket{0}^{\otimes a} \otimes \mathbf{I}) = \frac{\mathbf{B}_i}{ {2} \alpha k } $.

\end{proof}

\begin{lemma}[Block Encoding $\mathbf{Q}(\eta)$] There exists unitary matrix $\mathcal{U}(\mathbf{Q}(\eta))$ such that  $  (\bra{0}^{\otimes a} \otimes \mathbf{I}) \mathcal{U}( \mathbf{Q}(\eta) ) (\ket{0}^{\otimes a} \otimes \mathbf{I}) = \frac{\mathbf{Q}(\eta)}{ {6} \alpha k^2 } $
\label{block_encoding_q_eta}
\end{lemma}

\begin{proof}
We can implement $\mathcal{U}(\mathbf{Q}(\eta))$ by first creating unitary block encodings for the submatrices $\ket{i} \bra{i} \otimes \mathbf{A}_{i+1}$, 
$\ket{i} \bra{i+1} \otimes \frac{\mathbf{B}_{i+1}}{\eta}$, and $\ket{i+1} \bra{i} \otimes \frac{\mathbf{B}_{i+1}^{\dagger}}{\eta}$ and then the sum can be implemented using an LCU circuit. 
\begin{equation}
    \ket{i} \bra{i} \otimes \mathbf{A}_{i+1} =  (\ket{i} \bra{i} \otimes \mathbf{I} )(\mathbf{I} \otimes \mathbf{A}_{i+1})
\end{equation}
The block encoding for $\ket{i} \bra{i} \otimes \mathbf{A}_{i+1}$, $\mathcal{U}(\ket{i} \bra{i} \otimes \mathbf{A}_{i+1})$ can be constructed by applying Lemma 53 from \cite{gilyen2019quantum}. Specifically,
\begin{equation}
    \mathcal{U}(\ket{i} \bra{i} \otimes \mathbf{A}_{i+1}) = \left(\mathbf{I} \otimes \mathcal{U}(\ket{i} \bra{i} \otimes \mathbf{I})\right)\left(\mathbf{I} \otimes \mathcal{U}(\mathbf{I} \otimes \mathbf{A}_{i+1})\right)
\end{equation}
and,
\begin{equation}
    \mathcal{U}(\ket{i} \bra{i} \otimes \mathbf{I}) = \mathcal{U}(\ket{i} \bra{i}) \otimes \mathbf{I}
\end{equation}
\begin{equation}
    \mathcal{U}(\mathbf{I} \otimes \mathbf{A}_{i+1}) =   \mathbf{I} \otimes \mathcal{U}(\mathbf{A}_{i+1})
\end{equation}
We can implement block encoding for the projector $\ket{i}\bra{i}$ using techniques from \cite{gilyen2019quantum}, and the block encoding for the matrix $\mathbf{A}_{i+1}$, $\mathcal{U}(\mathbf{A}_{i+1})$ can be constructed using 
\autoref{block_encoding_Ai}. The block encoding constant for the matrix $\ket{i} \bra{i} \otimes \mathbf{A}_{i+1}$ is $2\alpha k$, and hence,
\begin{equation}
     (\ket{i} \bra{i} \otimes \mathbf{A}_{i+1}) = \frac{1}{2\alpha k}(\bra{0} \otimes \mathbf{I}) \mathcal{U}(\ket{i} \bra{i} \otimes \mathbf{A}_{i+1}) (\ket{0} \otimes \mathbf{I})
\end{equation}
Similarly, we can create block encodings for the following submatrices:

\begin{equation}
    \left(\ket{i} \bra{i+1} \otimes \frac{\mathbf{B}_{i+1}}{\eta} \right) = \frac{1}{2\alpha k} \left(\bra{0} \otimes \mathbf{I}\right) \mathcal{U} \left(\ket{i} \bra{i+1} \otimes \frac{\mathbf{B}_{i+1}}{\eta} \right) \left(\ket{0} \otimes \mathbf{I}\right)
\end{equation}

\begin{equation}
    \left(\ket{i+1} \bra{i} \otimes \frac{\mathbf{B}_{i+1}^{\dagger}}{\eta} \right) = \frac{1}{2\alpha k} \left(\bra{0} \otimes \mathbf{I}\right) \mathcal{U} \left(\ket{i+1} \bra{i} \otimes \frac{\mathbf{B}_{i+1}^{\dagger}}{\eta} \right) \left(\ket{0} \otimes \mathbf{I}\right)
\end{equation}
Now we can use an LCU circuit to sum these $3k$ submatrices, resulting in a unitary matrix block encoding $\mathbf{Q}(\eta)$.

\begin{equation}
    \frac{\mathbf{Q}(\eta)}{6\alpha k^2} =  (\bra{0}^{\otimes a} \otimes \mathbf{I}) \mathcal{U}(\mathbf{Q}(\eta)) (\ket{0}^{\otimes a} \otimes \mathbf{I})
\end{equation}

\end{proof}


\subsection{Reduction to Quantum Evolution}

We are now ready to combine these ideas and develop quantum algorithms for simulating the nonlinear Schr\"{o}dinger Equation defined in Eq.~\eqref{eq: nonlinear_Schr\"{o}dinger}. We first apply Carlemann Linearization techniques to linearize the nonlinear Schr\"{o}dinger Equation and then apply symmetrization techniques to symmetrize the resultant higher dimensional linear differential equation. The dynamical system described by symmetrized linear differential equations can be efficiently simulated by Hamiltonian simulation methods.


\begin{theorem}
\label{thm: nonlinear_Schr\"{o}dinger_subnormalized_state}
    Consider a nonlinear Schr\"{o}dinger Equation of the form: 
    \begin{equation}
         \dot{\ket{\bm{\psi}} } = -i\mathbf{H}_1 \ket{\bm{\psi}} + \mathbf{H}_2 \ket{\bm{\psi}} \otimes \ket{\bm{\psi}},
    \end{equation} where $\mathbf{H}_1 \in \mathbb{C}^{N \times N}$ is a Hermitian matrix and $\mathbf{H}_2 \in \mathbb{C}^{N \times N^2}$ is a rectangular matrix. Assume access to oracles defined in \autoref{Def:Oracles-nonlinear_Schr\"{o}dinger_eq}, and let $\beta = \braket{\psi(0)|\psi(0)}$. Let the $k$ th order Carlemann linearized system be $\ket{\dot{\mathbf{p}}(t)} = -i\mathbf{Q}\ket{\mathbf{p}(t)}$. Here,  $\ket{\mathbf{p}} = \sum_i \ket{i-1} \otimes \ket{\mathbf{w}_i}$ and  $\ket{\mathbf{w}_i}  = \bm{\ket{0}}^{\otimes k-i}\bm{\ket{\psi}}^{\otimes i}$. Let $\alpha = d\max_{jk}\{|\mathbf{H}_1(j,k)|,|\mathbf{H}_2(j,k)|\}$, where $d$ is the maximum number of non zero elements in any row or column of the matrices $\mathbf{H}_1$ and $\mathbf{H}_2$ and let $\mathcal{U}(\mathbf{Q}(\eta))$ be the unitary matrix encoding the symmetrized Carlemann Operator from \autoref{block_encoding_q_eta}. There exists a quantum algorithm that can prepare the state $\ket{\bm{\phi}(t)}$ such that $ \|\mathbf{T}\ket{\bm{\phi}(t)} - \ket{\mathbf{p}(t)} \| \leq \epsilon$. Here, $\mathbf{T} = \aleph\operatorname{diag}(\eta^{k-1} \mathbf{I}, \eta^{k-2} \mathbf{I}, \ldots, \mathbf{I}) $, where $\aleph = \frac{1}{\eta^k}\sqrt{\frac{\beta \eta^{2} \left(1 - (\beta \eta^2)^k\right)}{1 - \beta \eta^2}} $. 
    
    \begin{enumerate}
        \item The algorithm require $\mathcal{O}(1)$ calls to the initial state preparation Unitary $\mathcal{W}$, and $G$ calls to the unitary circuit implementing $\mathcal{U}(\mathbf{Q}(\eta))$ where 
    \begin{align}
       G =  \mathcal{O}\left( \alpha k^2 t + \frac{(k-1)}{2}\log\left(\left[\frac{\|\mathbf{H}_2 \| k\left( k +1 \right)}{2} \left(1 + \beta\frac{1-\beta^k}{(1-\beta)\epsilon} \right)t  \right] \right) \right)
\end{align}  
    \item When $\beta\eta^2 <1$, the probability of projecting onto the first Carleman block is given by
\[
p_1 \in \mathcal{O}\left( \frac{\braket{\psi(t)|\psi(t)}}{\beta} (1 - \beta \eta^2) \right)
\]
    \end{enumerate}
    
\end{theorem}

\begin{proof}

  We apply symmetrization techniques to construct the dynamical system of the form $\ket{\dot{\hat{\mathbf{p}}}(t)} = -i\hat{\mathbf{Q}}(\eta) \ket{\mathbf{\hat{p}}(t)}$ such that $\|\mathbf{D} \hat{\mathbf{p}}(t)-  \mathbf{{p}}(t)\| \leq \epsilon$ for specific values of $\eta$. Here,  $\ket{ \hat{\mathbf{p}}} = \sum_i \ket{i-1} \otimes \ket{ \hat{\mathbf{w}_i}}$ and  $\ket{ \hat{\mathbf{w}_i}}  = \frac{1}{\eta^{k-i}}\bm{\ket{0}}^{\otimes k-i}\bm{\ket{\psi}}^{\otimes i}$.
  The sufficient value of $\eta$ is given by \autoref{lemma: main_eta_bound}. 
\begin{align}
    \eta \geq  \sqrt{\|\hat{\mathbf{H}} \| \left(1 + \beta\frac{1-\beta^k}{(1-\beta)\epsilon} \right)t_s}
\end{align}
From \autoref{upper_bound_H_hat}, the norm $\|\hat{\mathbf{H}} \|$ is upper bounded by:
\begin{align}
  \|\hat{\mathbf{H}} \|   \leq \frac{\|\mathbf{H}_2 \| k\left( k +1 \right)}{2} 
\end{align}
Hence, the sufficient value for $\eta$ is given by:
\begin{align}
  \eta  \geq  \sqrt{\frac{\|\mathbf{H}_2 \| k\left( k +1 \right)}{2} \left(1 + \beta\frac{1-\beta^k}{(1-\beta)\epsilon} \right)t}
\end{align}

Our quantum state encoding is defined as: 

\begin{align}
    \ket{\mathbf{\phi}(t)} = \frac{1}{\aleph} \left( \ket{0} \otimes \ket{\hat{\mathbf{w}}_1(t)}  + \ket{1} \otimes  \ket{\hat{\mathbf{w}}_2(t)} + .. \ket{k-1} \otimes \ket{\hat{\mathbf{w}}_{k}(t)} \right)
\end{align}
 Here $\aleph$ is the normalization constant and $\ket{\mathbf{\hat{w}}_i} = \ket{\frac{\mathbf{w}_i}{\eta^{k-i}}}$.  Recall that $\ket{\mathbf{w}_i(t)}$ is a vector containing $\bm{\ket{\psi}}^{\otimes i}$ padded with zeros. The normalization constant $\aleph$ is given by:
\begin{align}
\label{normalization_constant}
    \aleph  &= \sqrt{\braket{\hat{p}(0)|\hat{p}(0)}} \\ \nonumber
            &= \sqrt{ \left(\sum_{j=1}^k \bra{j-1} \otimes \bra{ \hat{\mathbf{w}_j}(0)} \right)  \left( \sum_{i=1}^k \ket{i-1} \otimes \ket{ \hat{\mathbf{w}_i}(0)} \right) } \\ \nonumber
            &= \sqrt{ \sum_{i=1}^k \braket{ \hat{\mathbf{w}_i}(0)| \hat{\mathbf{w}_i}(0) } } \\ \nonumber
            &= \sqrt{\sum_{i=1}^k \frac{ \braket{\psi(0)|\psi(0)}^i }{\eta^{2(k-i)} } } = \sqrt{\sum_{i=1}^k \frac{\beta^i}{\eta^{2(k-i)}}}  \text{   }\left(\text{ Substituting } \beta = \braket{\psi(0)|\psi(0)} \right) \\ \nonumber 
            &= \frac{1}{\eta^k}\sqrt{\frac{\beta \eta^{2} \left(1 - (\beta \eta^2)^k\right)}{1 - \beta \eta^2}}
\end{align}

When $\beta \eta^2 <<1$, 
\begin{equation}
    \aleph \in \mathcal{O} \left(\frac{1}{\eta^k} \sqrt{\frac{ \beta \eta^2}{1-\beta \eta^2}} \right)
\end{equation}
When $\beta \eta^2 >> 1$, 
\begin{equation}
    \aleph \in \mathcal{O} \left(\sqrt{\beta^{k}} \right)
\end{equation}

  Our Algorithm would require high-precision encoding of the initial state and evolved quantum state at time $t$. Since QSVT based Hamiltonian simulation algorithms\cite{gilyen2019quantum} scales as $\mathcal{O}(\tau + \log(\frac{1}{\epsilon_1}))$, this is not a big overhead, because the cost would increase only by log factors. Specifically, our desired error in preparing final quantum state would be proportional to:
\begin{align}
    \epsilon_1 = \mathcal{O}\left({ \left[\frac{\|\mathbf{H}_2 \| k\left( k +1 \right)}{2} \left(1 + \beta\frac{1-\beta^k}{(1-\beta)\epsilon} \right)t \right]^{\frac{1-k}{2}} } \right)
\end{align}
Applying QSVT based Hamiltonian simulation algorithms, we can compute the desired quantum state using $G$ calls to the unitary matrix $\mathcal{U}(\mathbf{Q}(\eta))$ after preparing the initial state $\ket{\mathbf{\phi}(0)}$ using $\mathcal{O}(1)$ calls to the state preparation unitary $\mathcal{W}$. 
\begin{align}
    G = \mathcal{O}\left( \alpha k^2 t + \frac{(k-1)}{2}\log\left(\left[\frac{\|\mathbf{H}_2 \| k\left( k +1 \right)}{2} \left(1 + \beta\frac{1-\beta^k}{(1-\beta)\epsilon} \right)t  \right] \right) \right)
\end{align}

Since $\ket{\mathbf{\phi}(t)} = \frac{\hat{\mathbf{p}}(t)}{\aleph}$, and $\|\mathbf{D}\hat{\mathbf{p}}(t) - \mathbf{p}(t)\| \leq \epsilon$, it implies that $\|\aleph\mathbf{D}\ket{\mathbf{\phi}(t)} - \mathbf{p}(t)\| \leq \epsilon$. More succintly, we can write $\|\mathbf{T}\ket{\mathbf{\phi}(t)} - \mathbf{p}(t)\| \leq \epsilon$, where $\mathbf{T} = \aleph \mathbf{D}$.

The probability of projecting onto the first Carleman block is given by
\[
p_1 = \bra{\boldsymbol{\phi}(t)} \mathbf{P} \ket{\boldsymbol{\phi}(t)},
\]
where the projection matrix is defined as \( \mathbf{P} = \ket{0}\bra{0} \otimes \mathbf{I} \). Substituting the definition of \( \boldsymbol{\phi}(t) \), we obtain:
\begin{align}
    p_1 &= \bra{\boldsymbol{\phi}(t)} \mathbf{P} \ket{\boldsymbol{\phi}(t)} \\
        &= \frac{ \braket{\psi(t) | \psi(t)} }{ \eta^{2(k-1)} \aleph^2 }
\end{align}

Substituting the expression for the normalization constant \( \aleph^2 \) from \autoref{normalization_constant}, where
\[
\aleph^2 = \frac{1}{\eta^{2k}} \cdot \frac{\beta \eta^2 (1 - (\beta \eta^2)^k)}{1 - \beta \eta^2},
\]
we get:
\begin{align}
p_1 &= \frac{ \braket{\psi(t)|\psi(t)} }{ \eta^{2(k-1)} \cdot \left( \frac{1}{\eta^{2k}} \cdot \frac{\beta \eta^2 (1 - (\beta \eta^2)^k)}{1 - \beta \eta^2} \right) } \\
    &= \frac{ \braket{\psi(t)|\psi(t)} }{ \frac{1}{\eta^2} \cdot \frac{\beta \eta^2 (1 - (\beta \eta^2)^k)}{1 - \beta \eta^2} } \\
    &= \braket{\psi(t)|\psi(t)} \cdot \frac{1 - \beta \eta^2}{\beta \left(1 - (\beta \eta^2)^k\right)}
\end{align}

\noindent
When \( \beta \eta^2 < 1 \), the geometric factor \( (\beta \eta^2)^k \) vanishes exponentially with \( k \), and we obtain the asymptotic scaling:
\begin{equation}
p_1 \in \mathcal{O}\left( \frac{\braket{\psi(t)|\psi(t)}}{\beta} (1 - \beta \eta^2) \right)
\end{equation}
Note that one can always rescale the nonlinear Schrödinger equation to ensure that \( \beta \eta^2 < 1 \) holds.
\end{proof}
The query complexity in \autoref{thm: nonlinear_Schr\"{o}dinger_subnormalized_state} is a function of Carlemann Truncation order, which depends on the strength of nonlinearity and properties of the hermitian matrix $\mathbf{H}_1$. We now analyze two distinct regimes in which the approximation error can be systematically reduced by increasing the truncation order.
\NSEQuantumAlgorithmThm*

\begin{proof}
   Recall that the nonlinear Schr\"{o}dinger equation is given by:
\begin{equation}
         \dot{\ket{\bm{\psi}} } = -i\mathbf{H}_1 \ket{\bm{\psi}} + \mathbf{H}_2 \ket{\bm{\psi}} \otimes \ket{\bm{\psi}}
\end{equation}, where $\mathbf{H}_1 \in \mathbb{C}^{N \times N}$ is a Hermitian matrix  and $\mathbf{H}_2 \in \mathbb{C}^{N \times N^2}$ is a rectangular matrix. When $Re( \bra{\bm{\psi}(t) }\mathbf{H}_2 \ket{\bm{\psi}(t)} \otimes \ket{\bm{\psi}(t)}) = 0, \forall t \geq 0 $, the norm of the solution $\braket{\psi(t)|\psi(t)}$ is bounded, due to \autoref{lemma: norm_nonlinear_Schr\"{o}dinger_eq}. In this case, from \cite{forets2017explicit}, the truncation order is given by:
\[
  k \geq \mathcal{O}\left(\frac{\log(1/\epsilon)}{\log(1/\|\mathbf{H}_2\|)}\left(1 + \frac{\log(t)}{\log(1/\|\mathbf{H}_2\|)} \right) \right)
\]
Now, we can apply \autoref{thm: nonlinear_Schr\"{o}dinger_subnormalized_state}, and prepare the quantum state using $G$ queries to the unitary matrix $\mathbf{Q}(\eta)$, which can be constructed using a constant number of queries to the oracles accessing elements of the matrices $\mathbf{H}_1$ and $\mathbf{H_2}$. 
\[
G = \mathcal{O}\left( \alpha k^2 t + \frac{(k-1)}{2}\log\left(\left[\frac{\|\mathbf{H}_2 \| k\left( k +1 \right)}{2} \left(1 + \beta\frac{(1-\beta^k)}{(1-\beta)\epsilon} \right)t  \right] \right) \right)
\]

After preparing the quantum state $\ket{\mathbf{\phi}(t)}$:
\begin{align}
    \ket{\mathbf{\phi}(t)} = \frac{1}{\aleph} \left( \ket{0} \otimes \ket{\hat{\mathbf{w}}_1(t)}  + \ket{1} \otimes  \ket{\hat{\mathbf{w}}_2(t)} + .. \ket{k-1} \otimes \ket{\hat{\mathbf{w}}_{k}(t)} \right)
\end{align}
we measure the first $\log(k)$ qubits and $k-1$ registers, to prepare the final state
\begin{equation}
    \ket{\widetilde{\phi(t)}} = \frac{\ket{j}\otimes \ket{\widetilde{\psi(t)}}}{ \braket{\widetilde{\psi(t)}|\widetilde{\psi(t)} }  } 
\end{equation}

\end{proof}

\NSEQuantumAlgorithmThmResonanceCond*
\begin{proof}
   Recall that the nonlinear Schr\"{o}dinger equation is given by:
\begin{equation}
         \dot{\ket{\bm{\psi}} } = -i\mathbf{H}_1 \ket{\bm{\psi}} + \mathbf{H}_2 \ket{\bm{\psi}} \otimes \ket{\bm{\psi}},
\end{equation} where $\mathbf{H}_1 \in \mathbb{C}^{N \times N}$ is a Hermitian matrix  and $\mathbf{H}_2 \in \mathbb{C}^{N \times N^2}$ is a rectangular matrix. When $Re( \bra{\bm{\psi}(t) }\mathbf{H}_2 \ket{\bm{\psi}(t)} \otimes \ket{\bm{\psi}(t)}) = 0, \forall t \geq 0 $, the norm of the solution $\braket{\psi(t)|\psi(t)}$ is bounded, due to \autoref{lemma: norm_nonlinear_Schr\"{o}dinger_eq}. In this case, from Theorem 1.1 in \cite{wu2024quantum}, the truncation order is given by:

\begin{align}
    k &\in \frac{W\left( \frac{\epsilon}{CT} \cdot R_r \log(R_r) \right)}{\log(R_r)} \\ \nonumber
\end{align}
  Here, $W(.)$ is the Lambert W function, the parameter, $C = \|\mathbf{H}_2\|\beta^2$, and $R_r = \frac{4 e \beta \|\mathbf{H}_2\|}{\Delta}$, where the parameter, $\Delta$ is defined as follows:
     \[
\Delta := \inf_{k \in [n]} \inf_{\substack{m_j \geq 0 \\ \sum_{j=1}^n m_j \geq 2}}
\left(|
\lambda_k - \sum_{i=1}^n m_j \lambda_j
|\right)
\]
Simplifying the expression, we get:
\begin{align}
    k &\in \widetilde{\mathcal{O}} \left( \frac{\log(\frac{CT}{R_r\epsilon})}{\log(1/R_r)} \right)
\end{align}


   
Now, we can apply \autoref{thm: nonlinear_Schr\"{o}dinger_subnormalized_state}, and prepare the quantum state using $G$ queries to the unitary matrix $\mathbf{Q}(\eta)$, which can be constructed using a constant number of queries to the oracles accessing elements of the matrices $\mathbf{H}_1$ and $\mathbf{H_2}$. 
\[
G = \mathcal{O}\left( \alpha k^2 t + \frac{(k-1)}{2}\log\left(\left[\frac{\|\mathbf{H}_2 \| k\left( k +1 \right)}{2} \left(1 + \beta\frac{(1-\beta^k)}{(1-\beta)\epsilon} \right)t  \right] \right) \right)
\]

After preparing the quantum state $\ket{\mathbf{\phi}(t)}$:
\begin{align}
    \ket{\mathbf{\phi}(t)} = \frac{1}{\aleph} \left( \ket{0} \otimes \ket{\hat{\mathbf{w}}_1(t)}  + \ket{1} \otimes  \ket{\hat{\mathbf{w}}_2(t)} + .. \ket{k-1} \otimes \ket{\hat{\mathbf{w}}_{k}(t)} \right)
\end{align}
we measure the first $\log(k)$ qubits and $k-1$ registers, to prepare the final state
\begin{equation}
    \ket{\widetilde{\phi(t)}} = \frac{\ket{j}\otimes \ket{\widetilde{\psi(t)}}}{ \braket{\widetilde{\psi(t)}|\widetilde{\psi(t)} }  } 
\end{equation}
\end{proof}

\begin{algorithm}[H]
\caption{Quantum Algorithm for Nonlinear Schrödinger Equation}
\label{alg:nonlinear_Schr\"{o}dinger}
\begin{algorithmic}[1]
\Require
\State Oracle Access to Hermitian matrix $\mathbf{H}_1 \in \mathbb{C}^{N \times N}$ (\autoref{Def:Oracles-nonlinear_Schr\"{o}dinger_eq})
\State Oracle Access to Rectangular matrix $\mathbf{H}_2 \in \mathbb{C}^{N \times N^2}$ (\autoref{Def:Oracles-nonlinear_Schr\"{o}dinger_eq})
\State Initial state $\ket{\psi(0)}$ with $\braket{\psi(0)|\psi(0)} = \beta$
\State Time parameter $t$
\State Error tolerance $\epsilon$
\State Maximum number of nonzero elements $d$ in any row/column of $\mathbf{H}_1, \mathbf{H}_2$

\Ensure Quantum state $\ket{\widetilde{\psi}(t)}$ approximating the solution

\Function{NonlinearSchr\"{o}dingerSimulation}{$\mathbf{H}_1, \mathbf{H}_2, \ket{\psi(0)}, t, \epsilon$}

\State // Step 1: Determine Carlemann truncation order
\State $k \gets \mathcal{O}\left(\frac{\log(1/\epsilon)}{\log(1/\|\mathbf{H}_2\|)}\left(1 + \frac{\log(t)}{\log(1/\|\mathbf{H}_2\|)} \right) \right)$ or $\widetilde{\mathcal{O}} \left( \frac{\log\left(\frac{\| \mathbf{H}_2\|\beta^2T}{\epsilon}\right)}{\log(\frac{1}{R_r})}\right)$

\State // Step 2: Compute symmetrization parameter
\State $\eta \gets \sqrt{\frac{\|\mathbf{H}_2\| k(k+1)}{2}\left(1 + \beta\frac{1-\beta^k}{(1-\beta)\epsilon}\right)t}$

\State // Step 3: Initialize extended state vector
\State $\ket{\mathbf{p}} \gets \sum_{i=1}^k \ket{i-1} \otimes \ket{\mathbf{w}_i}$
\State where $\ket{\mathbf{w}_i} = \bm{\ket{0}}^{\otimes k-i}\bm{\ket{\psi}}^{\otimes i}$

\State // Step 4: Construct symmetrized Carlemann operator
\State $\mathbf{Q}(\eta) \gets \sum_{i=0}^{k-1} \left(\ket{i}\bra{i} \otimes \mathbf{A}_{i+1} + \ket{i}\bra{i+1} \otimes \frac{\mathbf{B}_{i+1}}{\eta} + \ket{i+1}\bra{i} \otimes \frac{\mathbf{B}_{i+1}^\dagger}{\eta}\right)$

\State // Step 5: Construct block encoding
\State $\mathcal{U}(\mathbf{Q}(\eta)) \gets$ BlockEncodingConstruct($\mathbf{Q}(\eta)$)

\State // Step 6: Number of queries for Hamiltonian simulation
\State $G \gets \mathcal{O}\left(\alpha k^2 t + \frac{k-1}{2}\log\left(\frac{\|\mathbf{H}_2\|k(k+1)}{2}\left(1 + \beta\frac{1-\beta^k}{(1-\beta)\epsilon}\right)t\right)\right)$

\State // Step 7: Prepare evolved quantum state
\State $\ket{\phi(t)} \gets$ HamiltonianSimulation($\mathcal{U}(\mathbf{Q}(\eta)), G$)

\State // Step 8: Measure to get final state
\State $\ket{\widetilde{\phi}(t)} \gets \text{Measure}(\ket{\phi(t)})$

\State \Return $\ket{\widetilde{\phi}(t)}$
\EndFunction
\end{algorithmic}
\end{algorithm}

\subsection{Computation of Expectation Values}

Our encoding scheme stores the solution in the form:

\begin{align}
    \ket{\mathbf{\phi}(t)} = \frac{1}{\aleph} \left( \ket{0} \otimes \ket{\hat{\mathbf{w}}_0(t)}  + \ket{1} \otimes  \ket{\hat{\mathbf{w}}_1(t)} + .. \ket{k-1} \otimes \ket{\hat{\mathbf{w}}_{k-1}(t)} \right)
\end{align}
Here, 
\begin{align}
    \ket{\hat{\mathbf{w}}_{k-1}(t)} = \ket{\bm{\psi}(t)}^{\otimes k}
\end{align}
To measure expectation value of operator $\mathbf{O}$ defined on $ \ket{\bm{\psi}(t)}$, it is sufficient to measure the expectation value of the operator of the form:
\begin{align}
    \hat{\mathbf{O}} = \ket{0} \bra{0} \otimes \mathbf{O}' + \ket{1} \bra{1} \otimes \mathbf{0} ..+ \ket{k-1} \bra{k-1} \otimes \mathbf{0}
\end{align}
Here $\mathbf{O}'$ is an operator defined as:
\begin{equation}
    \mathbf{O}' = \begin{bmatrix}
                    \mathbf{O} & \mathbf{0} & \mathbf{0} & \cdots \\
                    \mathbf{0} & \mathbf{0} & \mathbf{0} & \cdots \\
                    \mathbf{0} & \mathbf{0} & \mathbf{0} & \cdots \\
                    \vdots & \vdots & \vdots & \ddots \\
                    \end{bmatrix}
\end{equation}
\begin{align}
  \bra{{\mathbf{\phi}(t)}}   \hat{\mathbf{O}}  \ket{{\mathbf{\phi}(t)}} =  \frac{1}{\aleph^2}\bra{\hat{\mathbf{w}}_0(t)}  \mathbf{O}'\ket{\hat{\mathbf{w}}_0(t)} = \frac{1}{\aleph^2 \eta^{k-1}} \bra{{\bm{\psi}}(t)}  \mathbf{O}\ket{{\bm{\psi}}(t)}
\end{align}

%

Using a method based on high-confidence amplitude amplification, we can compute this expectation with a probability $1-\delta$ using $\mathcal{O}(\log(1/\delta)\frac{1}{\epsilon}) $ queries to Unitary Matrix obtained from \autoref{thm: nonlinear_Schr\"{o}dinger_subnormalized_state}.

\section{Nonlinear Oscillator Systems}
\label{sec:nonlinear_oscillator_system}
In this section, we consider the problem of simulating a coupled nonlinear oscillator system, focusing on the special case where the nonlinear term is quadratic. Consider an oscillator system of $N$ coupled masses, where each mass is influenced by its neighbors through linear and nonlinear interactions. The dynamics of such a system can be captured by a set of second-order differential equations. In the case of quadratic nonlinearity, the equations of motion can be written as: $\mathbf{M}\ddot{\mathbf{x}} = -\mathbf{K_1} \mathbf{x} + \mathbf{K_2}\mathbf{x}\otimes\mathbf{x}$. The formal definition of problem is given below. 

\nonlinearoscillator*

We first show that the dynamics of a nonlinear oscillator system can be reduced to the dynamics of a quantum system evolving according to the nonlinear Schr\"{o}dinger Equation. Using techniques described in the earlier section, we can efficiently simulate the nonlinear Schr\"{o}dinger Equation, when the nonlinear Schr\"{o}dinger Equation satisfies certain constraints. Specifically, the norm of the solution of the nonlinear Schr\"{o}dinger Equation has to be upper bounded by a constant, and the system should satisfy non resonance conditions. These conditions place a restriction on the class of nonlinear classical oscillator systems that can be efficiently simulated on quantum computers. 

\subsection{Reduction to Nonlinear Schr\"{o}dinger Equation}   
\label{section:reduction_to_nonlinear_Schr\"{o}dinger}
First, we apply a coordinate transformation to transform the second-order nonlinear ODE to a first-order nonlinear ODE.\\
Let
\begin{align}
    \mathbf{u} = \sqrt{\mathbf{M}} \mathbf{x} \in \mathbb{R}^{N}, \text{ where } N = 2^n \text{ for some } n \in \mathbb{N}
\end{align}
Applying this coordinate transformation to Eq.~\eqref{eq:nonlinear_oscillator} yields:
\begin{align}
    \ddot{\mathbf{u}} = -\sqrt{\mathbf{M}}^{-1} \mathbf{K}_1 \sqrt{\mathbf{M}}^{-1} \mathbf{u} + \sqrt{\mathbf{M}}^{-1} \mathbf{K}_2 [\sqrt{\mathbf{M}}^{-1} \otimes \sqrt{\mathbf{M}}^{-1} ] \mathbf{u} \otimes \mathbf{u}
\end{align}
Which is of the form:
\begin{align}
  \ddot{\mathbf{u}} = -\mathbf{A}_1 \mathbf{u} + \mathbf{A}_2 \mathbf{u} \otimes \mathbf{u}  
\end{align}
Here $\mathbf{A}_1 = \sqrt{\mathbf{M}}^{-1} \mathbf{K}_1 \sqrt{\mathbf{M}}^{-1} $ is a negative definite matrix and, the rectangular matrix $\mathbf{A}_2 =  \sqrt{\mathbf{M}}^{-1} \mathbf{K}_2 [\sqrt{\mathbf{M}}^{-1} \otimes \sqrt{\mathbf{M}}^{-1} ]  $  and its norm satisfies $\| \mathbf{A}_2\| \leq \| \mathbf{K}_2\| \| \sqrt{\mathbf{M}^{-1}} \|^3 $. Since $\mathbf{A}_1$ is a positive definite matrix, we can decompose  $\mathbf{A}_1$ as \(\mathbf{A_1} = \mathbf{B}\mathbf{B}^{\dagger}\). Here, the rectangular matrix $\mathbf{B} \in \mathbb{C}^{N \times M}$, and its transpose, $\mathbf{B}^{\dagger}$ acts as generalized square root of the matrix $\mathbf{A}_1$. Now, we encode the positions and momenta of classical oscillator system in a quantum state vector using the mapping adopted in \cite{babbush2023exponential}. Specifically, we define $\ket{\bm{\psi}(t)} \in \mathbb{C}^{N +M} $:

\begin{align}
    \ket{\bm{\psi}(t)} = \begin{bmatrix}
       \dot{\mathbf{u}} \\
        -i \mathbf{B}^{\dagger}\mathbf{u}
    \end{bmatrix}
\end{align}
Here, $-i \mathbf{B}^{\dagger}\mathbf{u} \in \mathbb{C}^M$.
\begin{align}
\label{reduced_Schr\"{o}dinger_eq}
    \dot{\ket{\bm{\psi}} } = \begin{bmatrix}
        \ddot{\mathbf{u}} \\
        -i \mathbf{B}^{\dagger}\dot{\mathbf{u}}
    \end{bmatrix}
     &= -i\begin{bmatrix}
         \mathbf{0} & \mathbf{B}\\
         \mathbf{B}^{\dagger} & \mathbf{0}\\
     \end{bmatrix}
     \begin{bmatrix}
       \dot{\mathbf{u}} \\
        -i \mathbf{B}^{\dagger}\mathbf{u}
    \end{bmatrix} 
    \\ \nonumber
    &+ \begin{bmatrix}
        \mathbf{0}_{N \times N^2} & \mathbf{0}_{N \times NM} & \mathbf{0}_{N \times NM} & - \mathbf{A}_2 \mathbf{D}_2 \\
        \mathbf{0}_{M \times N^2} & \mathbf{0}_{M \times NM} & \mathbf{0}_{M \times NM} & \mathbf{0}_{M \times M^2}\\
    \end{bmatrix}
    \mathbf{P_2}
    \begin{bmatrix}
       \dot{\mathbf{u}} \\
        -i\mathbf{B}^{\dagger}\mathbf{u}
    \end{bmatrix} 
    \otimes
    \begin{bmatrix}
       \dot{\mathbf{u}} \\
        -i\mathbf{B}^{\dagger}\mathbf{u}
    \end{bmatrix}  
\end{align}

Here, $\mathbf{P}_2 \in \mathbb{R}^{(M+N)^2 \times (M+N)^2}$ is a permutation matrix which acts in the following fashion:
\begin{align}
    \mathbf{P}_2   \begin{bmatrix}
       \dot{\mathbf{u}} \\
        -i\mathbf{B}^{\dagger}\mathbf{u}
    \end{bmatrix} 
    \otimes
    \begin{bmatrix}
       \dot{\mathbf{u}} \\
        -i\mathbf{B}^{\dagger}\mathbf{u}
    \end{bmatrix}    = \begin{bmatrix}
        \dot{\mathbf{u}} \otimes \dot{\mathbf{u}} \\
         \dot{\mathbf{u}} \otimes -i\mathbf{B}^{\dagger} \mathbf{u}  \\
        -i\mathbf{B}^{\dagger} \mathbf{u} \otimes \dot{\mathbf{u}} \\
        -i\mathbf{B}^{\dagger} \mathbf{u} \otimes -i\mathbf{B}^{\dagger} \mathbf{u}  \\
    \end{bmatrix}
\end{align}

Now we choose the matrix, $\mathbf{D}_2 \in \mathbb{R}^{N^2 \times M^2}$ such that $\mathbf{D}_2 [i\mathbf{B}^{\dagger} \mathbf{u} \otimes -i\mathbf{B}^{\dagger} \mathbf{u}] = \mathbf{u} \otimes \mathbf{u} $. Recall that $i\mathbf{B}^{\dagger} \mathbf{u}$ can be written as:
\[
i \mathbf{B}^{\dagger}\mathbf{u} = i\begin{bmatrix}
                                    \sqrt{k_{11}} u_1 \\
                                    \sqrt{k_{22}} u_2 \\
                                    \vdots \\
                                    \sqrt{k_{12}}(u_1 -u_2) \\
                                    \vdots \\
                                    \sqrt{k_{N-1,N} }(u_{N-1} -u_{N})
                                    \end{bmatrix}
                                    = \begin{bmatrix}
                                        \sqrt{\mathbf{K}_1^{(1)}} \mathbf{u} \\
                                        \\
                                         \sqrt{\mathbf{K}_1^{(2)}} \Delta\mathbf{u}   
                                    \end{bmatrix}
\]
Here, $\Delta\mathbf{u}^T = [(u_1 -u_2),(u_1-u_3),\cdots, (u_{N-1} -u_{N})] \in \mathbb{C}^{ \frac{N(N-1)}{2}}$. Now we can write the tensor product $i\mathbf{B}^{\dagger} \mathbf{u} \otimes -i\mathbf{B}^{\dagger} \mathbf{u}$ as:

\begin{equation}
    i\mathbf{B}^{\dagger} \mathbf{u} \otimes -i\mathbf{B}^{\dagger} \mathbf{u} = \mathbf{P}' \begin{bmatrix}
                                       \left( \sqrt{\mathbf{K}_1^{(1)}} \otimes \sqrt{\mathbf{K}_1^{(1)}} \right) ( \mathbf{u} \otimes  \mathbf{u} ) \\
                        
                                         \left( \sqrt{\mathbf{K}_1^{(1)}} \otimes \sqrt{\mathbf{K}_1^{(2)}} \right) ( \mathbf{u} \otimes  \Delta\mathbf{u} ) \\
                                         \left( \sqrt{\mathbf{K}_1^{(2)}} \otimes \sqrt{\mathbf{K}_1^{(1)}} \right) ( \Delta \mathbf{u} \otimes  \mathbf{u} )\\
                                         \left( \sqrt{\mathbf{K}_1^{(2)}} \otimes \sqrt{\mathbf{K}_1^{(2)}} \right) ( \Delta \mathbf{u} \otimes \Delta \mathbf{u} )
                                    \end{bmatrix}
\end{equation}
Here, $\mathbf{P}'$ is a permutation matrix, and the matrices $\mathbf{K}_1^{(1)}$, and $\mathbf{K}_1^{(2)}$ are diagonal. Now, let us  choose the matrix $\mathbf{D}_2$ as:

\begin{equation}
    \mathbf{D}_2 = \begin{bmatrix}
        \left( \sqrt{(\mathbf{K}_1^{(1)})^{-1}} \otimes \sqrt{(\mathbf{K}_1^{(1)}})^{-1} \right)  & \mathbf{0} &\mathbf{0} & \mathbf{0}
    \end{bmatrix} (\mathbf{P}')^{\dagger}
\end{equation}
It can be seen that $\mathbf{D}_2 [i\mathbf{B}^{\dagger} \mathbf{u} \otimes -i\mathbf{B}^{\dagger} \mathbf{u}] = \mathbf{u} \otimes \mathbf{u} $, and since $\mathbf{K}_1^{(1)}$ is a diagonal matrix, its inverse can be computed efficiently, and therefore we can block encode the matrix $\mathbf{D}_2$ efficiently. The norm of the matrix $\mathbf{D}_2$, is given by:
\begin{equation}
    \|\mathbf{D}_2\| = \max_{ij} \frac{1}{\sqrt{k_{ii}k_{jj}}} = \frac{1}{k^{(1)}_{min}}
\end{equation}
Now, consider Eq.~\eqref{reduced_Schr\"{o}dinger_eq} which is of the form:
\begin{align}
    \dot{\ket{\bm{\psi}} } = -i\mathbf{H}_1 \ket{\bm{\psi}} + \mathbf{H}_2 \ket{\bm{\psi}} \otimes \ket{\bm{\psi}}
\end{align}
where,
\begin{align}
    \mathbf{H}_1 = \begin{bmatrix}
         \mathbf{0} & \mathbf{B}\\
         \mathbf{B}^{\dagger} & \mathbf{0}\\
     \end{bmatrix}
\end{align}
and 
\begin{align}
    \mathbf{H}_2 = \begin{bmatrix}
        \mathbf{0}_{N \times N^2} & \mathbf{0}_{N \times NM} & \mathbf{0}_{N \times NM} & - \mathbf{A}_2 \mathbf{D}_2 \\
        \mathbf{0}_{M \times N^2} & \mathbf{0}_{M \times NM} & \mathbf{0}_{M \times NM} & \mathbf{0}_{M \times M^2}\\
    \end{bmatrix}  \mathbf{P}_2
\end{align}
 Since it is easier to work with square matrices, we further pad these matrices with zeros. Specifically, first let's pad $\mathbf{B} \in \mathbb{C}^{N \times M}$ with zeros to create the padded rectangular matrix $\hat{\mathbf{B}} \in \mathbb{C}^{N \times N^2}$. Let the square matrix constructed from $\hat{\mathbf{B}}$ be $\bar{\mathbf{B}} = \ket{0}\otimes \hat{\mathbf{B}}$. We define the modified quantum state as follows:
\begin{align}
    \ket{\bm{\psi}(t)} = \begin{bmatrix}
       \ket{0} \otimes \dot{\mathbf{u}} \\
        -i \hat{\mathbf{B}}^{\dagger}\mathbf{u}
    \end{bmatrix}
\end{align}
and the modified nonlinear Schr\"{o}dinger Equation is given by:
\begin{align}
\label{nonlinear_Schr\"{o}dinger_eq_padded}
    \dot{\ket{\bm{\psi}} } = \begin{bmatrix}
        \ket{0} \otimes \ddot{\mathbf{u}} \\
        -i \hat{\mathbf{B}}^{\dagger}\dot{\mathbf{u}}
    \end{bmatrix}
     &= -i\begin{bmatrix}
         \mathbf{0} & \bar{\mathbf{B}}\\
         \bar{\mathbf{B}}^{\dagger} & \mathbf{0}\\
     \end{bmatrix}
     \begin{bmatrix}
      \ket{0} \otimes \dot{\mathbf{u}} \\
        -i \hat{\mathbf{B}}^{\dagger}\mathbf{u}
    \end{bmatrix} 
    \\ \nonumber
    &+ \begin{bmatrix}
        \mathbf{0}_{N^2 \times N^4} & \mathbf{0}_{N^2 \times N^4} & \mathbf{0}_{N^2 \times N^4} & - \bar{\mathbf{A}}_2 \bar{\mathbf{D}}_2 \\
        \mathbf{0}_{N^2 \times N^4} & \mathbf{0}_{N^2 \times N^4} & \mathbf{0}_{N^2 \times N^4} & \mathbf{0}_{N^2 \times N^4}\\
    \end{bmatrix}
    \bar{\mathbf{P}}_2
    \begin{bmatrix}
     \ket{0} \otimes \dot{\mathbf{u}}  \\
        -i\hat{\mathbf{B}}^{\dagger}\mathbf{u}
    \end{bmatrix} 
    \otimes
    \begin{bmatrix}
       \ket{0} \otimes \dot{\mathbf{u}}  \\
        -i\hat{\mathbf{B}}^{\dagger}\mathbf{u}
    \end{bmatrix}  
\end{align}

Here, $ \bar{\mathbf{P}}_2$ is a permutation matrix such that:
\[
\bar{\mathbf{P}}_2
   \left( \begin{bmatrix}
     \ket{0} \otimes \dot{\mathbf{u}}  \\
        -i\hat{\mathbf{B}}^{\dagger}\mathbf{u}
    \end{bmatrix} 
    \otimes
    \begin{bmatrix}
       \ket{0} \otimes \dot{\mathbf{u}}  \\
        -i\hat{\mathbf{B}}^{\dagger}\mathbf{u}
    \end{bmatrix} \right)  = \begin{bmatrix}
        \ket{0} \otimes  \dot{\mathbf{u}} \otimes  \ket{0} \otimes  \dot{\mathbf{u}} \\
        \ket{0} \otimes  \dot{\mathbf{u}}  \otimes -i\hat{\mathbf{B}}^{\dagger} \mathbf{u}  \\
        -i\hat{\mathbf{B}}^{\dagger} \mathbf{u} \otimes  \ket{0} \otimes  \dot{\mathbf{u}} \\
        -i\hat{\mathbf{B}}^{\dagger} \mathbf{u} \otimes -i\hat{\mathbf{B}}^{\dagger} \mathbf{u}  \\
    \end{bmatrix}
\]
The \autoref{nonlinear_Schr\"{o}dinger_eq_padded} is still a form of nonlinear Schr\"{o}dinger Equation, with 
\begin{equation}
    \bar{\mathbf{H}}_1 =  \begin{bmatrix}
         \mathbf{0} & \bar{\mathbf{B}}\\
         \bar{\mathbf{B}}^{\dagger} & \mathbf{0}\\
     \end{bmatrix}
\end{equation}
, and 
\begin{align}
    \bar{\mathbf{H}}_2 &=  \begin{bmatrix}
        \mathbf{0}_{N^2 \times N^4} & \mathbf{0}_{N^2 \times N^4} & \mathbf{0}_{N^2 \times N^4} & - \bar{\mathbf{A}}_2 \bar{\mathbf{D}}_2 \\
        \mathbf{0}_{N^2 \times N^4} & \mathbf{0}_{N^2 \times N^4} & \mathbf{0}_{N^2 \times N^4} & \mathbf{0}_{N^2 \times N^4}\\
    \end{bmatrix} 
    \bar{\mathbf{P}}_2 \\ \nonumber
     &= (\mathbf{C}_1 \otimes - \bar{\mathbf{A}}_2 \bar{\mathbf{D}}_2 )\bar{\mathbf{P}}_2
\end{align}
where $ \mathbf{C}_1$ is defined as: 
\begin{align}
    \mathbf{C}_1 = \begin{bmatrix} 
        0 & 0 & 0 & 1 \\ 
        0 & 0 & 0 & 0 
    \end{bmatrix}.
\end{align}
The norm of the matrix $ \bar{\mathbf{H}}_2$ is bounded by:
\begin{equation}
    \|  \bar{\mathbf{H}}_2 \| \leq \|\mathbf{A}_2\| \|\mathbf{D}_2 \| \leq \| \mathbf{K}_2\| \| \sqrt{\mathbf{M}^{-1}} \|^3 \|\mathbf{D}_2 \| \leq \frac{d k^{(2)}_{\max} }{k^{(1)}_{\min} (m_{\min})^{3/2}}
\end{equation}
where $k^{(2)}_{\max} = \max_{ij}\mathbf{K}_2(i,j)$, and $m_{\min} = \min_i \mathbf{M}(i,i)$, and $d$ is the sparsity of the matrix $\mathbf{K}_2$.  The matrix $\bar{\mathbf{D}}_2 $ is defined such that $\bar{\mathbf{D}}_2 ( -i\hat{\mathbf{B}}^{\dagger} \mathbf{u} \otimes -i\hat{\mathbf{B}}^{\dagger} \mathbf{u}) = \mathbf{u} \otimes \mathbf{u}$, and $\bar{\mathbf{A}}_2 = \ket{0}\otimes \mathbf{A}_2 \in \mathbb{C}^{N^2 \times N^2}$. 
The nonlinear Schr\"{o}dinger Equation can be simulated using the quantum algorithm discussed in \autoref{sec:nonlinear_Schrodinger_eq}.  
\subsection{Block Encodings}
To apply the quantum algorithm, we need to first create block encodings for $\bar{\mathbf{H}}_1$, and $\bar{\mathbf{H}}_2$ using oracles that access elements of the matrices $\mathbf{K}_1$, and $\mathbf{K}_2$. The block encoding for $\bar{\mathbf{H}}_1$ can be constructed by directly applying results from \cite{babbush2023exponential}. 
\begin{lemma}[Block Encoding $\bar{\mathbf{H}}_1$]
\label{Block_Encode_H1_nonlinear_oscillator}
    Assume access to elements of the matrix $\mathbf{G}_1$ through oracles defined in \autoref{Def:Oracles-nonlinear-oscillator-system}. Let $\mathbf{A}_1 = \sqrt{\mathbf{M}^{-1}} \mathbf{K}_1 \sqrt{\mathbf{M}^{-1}}$ and let $\mathbf{A}_1 = \mathbf{B}\mathbf{B}^{\dagger}$. Here, $\mathbf{B} \in \mathbb{C}^{N \times M}$ Let the Hamiltonian $\mathbf{H}_1$ be:
    \[
     \mathbf{H}_1 = \begin{bmatrix}
         \mathbf{0} & \mathbf{B}\\
         \mathbf{B}^{\dagger} & \mathbf{0}\\
     \end{bmatrix}
    \]
    Consider the padded $d$-sparse Hamiltonian $\bar{\mathbf{H}}_1 \in \mathbb{C}^{2N^2 \times 2N^2}$ constructed  from the padded matrices $\bar{\mathbf{B}} \in \mathbb{C}^{N^2 \times N^2}$, and $\bar{\mathbf{B}}^{\dagger} \in \mathbb{C}^{N^2 \times N^2}$. The matrix $\bar{\mathbf{B}}$, and its transpose is obtained from padding the matrix $\mathbf{B}$ with additional zeroes. 
    \[
    \bar{\mathbf{H}}_1 = \begin{bmatrix}
         \mathbf{0} & \bar{\mathbf{B}}\\
         \bar{\mathbf{B}}^{\dagger} & \mathbf{0}\\
     \end{bmatrix}
    \]
    There exists a unitary matrix $\mathcal{U}(\bar{\mathbf{H}}_1)$ block encoding the matrix $\bar{\mathbf{H}}_1$ such that $\|\alpha_1 (\bra{0} \otimes \mathbf{I})\mathcal{U}(\bar{\mathbf{H}}_1) (\ket{0} \otimes \mathbf{I}) - \mathbf{H}_1 \| \leq \epsilon'$. 
\end{lemma}
\begin{proof}
    We can construct this block encoding by applying Lemma 9 from \cite{babbush2023exponential}. The block encoding constant, $\alpha_1 = \sqrt{2 \aleph d}$, where $\aleph = \frac{k_{max}}{m_{min}}$. To construct the block encoding we need to use the oracles from \autoref{Def:Oracles-nonlinear_Schr\"{o}dinger_eq}, together with its inverses and controlled versions $\mathcal{O}(1)$ times. 
\end{proof}
Now, let us consider the strategy for block encoding $\bar{\mathbf{H}}_2$. Recall, that to apply the quantum algorithm described in previous section, we require a block encoding for the padded rectangular matrix $\bar{\mathbf{H}}_2$. Let the padded matrix be $(\bar{\mathbf{H}}_2)_{pad}$, 

\begin{align}
    (\bar{\mathbf{H}}_2)_{pad} &= \ket{0}^{\otimes{2n + 1}} \otimes \bar{\mathbf{H}}_2 \\ \nonumber
    &:= \ket{0}^{\otimes{2n}} \otimes \ket{0} \otimes (\mathbf{C}_1 \otimes - \bar{\mathbf{A}}_2 \bar{\mathbf{D}}_2 )\bar{\mathbf{P}}_2\\ \nonumber
    &:= \left[(\ket{0}\otimes \mathbf{C}_1) \otimes (\ket{0}^{\otimes{2n}}  \otimes ( - \bar{\mathbf{A}}_2 \bar{\mathbf{D}}_2 )) \right]\left[\mathbf{I}\otimes \bar{\mathbf{P}}_2 \right]
\end{align}
To block encode $(\bar{\mathbf{H}}_2)_{pad}$, we need to block encode these submatrices first and combine them appropriately.  
\begin{lemma}[Block Encoding $\mathbf{K}_2' = \ket{0}^{\otimes n} \otimes \mathbf{K}_2$]
\label{Block_Encode_K_2_matrix}
    Assume access to elements of the matrix $\mathbf{K}_2$ through oracles in \autoref{Def:Oracles-nonlinear-oscillator-system}.  Let $d_r$ be the row wise sparsity of $\mathbf{K}_2$, and let $d_c$ be the coloumnwise sparsity of the matrix $\mathbf{K}_2$. Let, $d=\max(d_r,d_c)$, and let the maximum absolute value of the entries in the matrix $\mathbf{K}_2$ be $\max_{ij}|\mathbf{K}_2(i,j)|$ and let $\alpha_2 = d \max_{ij}|\mathbf{K}_2(i,j)|$. There exists a unitary matrix $\mathcal{U}(\mathbf{K}_2')$ such that $\| \alpha_2(\bra{0} \otimes \mathbf{I})\mathcal{U}({\mathbf{K}}_2') (\ket{0} \otimes \mathbf{I})   -\mathbf{K}_2'\| \leq \epsilon'$. 
\end{lemma}
\begin{proof}
   We can create block encoding of the matrix by directly applying Lemma 48 from \cite{gilyen2019quantum}. Specifically, we can implement a $(\alpha_2,2\log(N) +3,\epsilon')$ block encoding of the matrix $\mathbf{K}_2'$ using $\mathcal{O}(1)$ uses of the access oracles. 
\end{proof}

\begin{lemma}[ Block Encoding $\sqrt{\mathbf{M}^{-1}}$] 
\label{Block_Encode_Inverse_Mass_Matrix}
    Assume access to elements of the diagonal matrix $\mathbf{M}$ through oracles defined in \autoref{Def:Oracles-nonlinear-oscillator-system}. Let $m_{\min} = \min_{j} \mathbf{M}(j,j)$ and $\alpha_3 = \frac{1}{\sqrt{m_{\min}}}$. There exists a unitary matrix $\mathcal{U}(\sqrt{\mathbf{M}^{-1}})$ such that $\| \alpha_3(\bra{0} \otimes \mathbf{I}) \mathcal{U}(\sqrt{\mathbf{M}^{-1}})  (\ket{0} \otimes \mathbf{I}) -  \sqrt{\mathbf{M}^{-1}}  \| \leq \epsilon'$
\end{lemma}
\begin{proof}
Given oracle access to the diagonal entries \(m_j = \mathbf{M}(j,j)\) via unitaries \(\mathcal{O}_M\) and \(\mathcal{O}_{\text{index}}\), we construct \(\mathcal{U}(\sqrt{\mathbf{M}^{-1}})\) as follows:

    For each basis state \(|j\rangle\), query \(\mathcal{O}_M\) to load \(m_j\) into a register:
    \[
   \ket{j}\ket{0} \xrightarrow{\mathcal{O}_M} \ket{j}\ket{\bar{m}_j}
    \]
    Let \(\tilde{s}_j = \sqrt{\frac{m_{\min}}{m_j}}\). Using state preparation via inequality testing, prepare the following quantum state:
    \[
    |j\rangle|\tilde{s}_j\rangle|0\rangle \mapsto |j\rangle|\tilde{s}_j\rangle\left(\sqrt{1 - \tilde{s}_j^2} \ket{1} + \tilde{s}_j\ket{0}\right).
    \]
    Uncompute \(\tilde{s}_j\) and \(m_j\) to disentangle ancilla registers:
    \[
    |j\rangle\left(\sqrt{1 - \tilde{s}_j^2}\ket{1} + \tilde{s}_j\ket{0}\right)|0\rangle.
    \]
    Define \(\mathcal{U}(\sqrt{\mathbf{M}^{-1}})\) as the composition of the above steps. Post-select on the ancilla qubit being $\ket{0}$:
    \[
    \alpha(\bra{0} \otimes \mathbf{I})\mathcal{U}(\sqrt{\mathbf{M}^{-1}})(\ket{0} \otimes \mathbf{I}) = \sum_j \frac{1}{\sqrt{m_j}} \ket{j}\bra{j}.
    \]
    Since \(\alpha\tilde{s}_j = \frac{1}{\sqrt{m_{\min}}}\tilde{s}_j \approx \frac{1}{\sqrt{m_j}}\), this approximates \(\sqrt{\mathbf{M}^{-1}}\). Thus, \(\mathcal{U}(\sqrt{\mathbf{M}^{-1}})\) is a valid block-encoding of \(\sqrt{\mathbf{M}^{-1}}\). 
\end{proof}

\begin{lemma}[Block Encoding $\mathcal{U}(\sqrt{\mathbf{M}^{-1}} \otimes \sqrt{\mathbf{M}^{-1}})$]
    
\end{lemma}
\begin{proof}
    We can block-encode the tensor product $(\sqrt{\mathbf{M}^{-1}} \otimes \sqrt{\mathbf{M}^{-1}})$ by applying Proposition 4 from \cite{takahira2021quantum}. Constructing this block encoding require $\mathcal{O}(1)$ uses of the unitary $\mathcal{U}(\sqrt{\mathbf{M}^{-1}})$. The block encoding constant is $\alpha_4$, where $\alpha_4 = \frac{1}{{m_{\min}}}$.
\end{proof}

Now, we can apply \autoref{Block_Encode_K_2_matrix} and \autoref{Block_Encode_Inverse_Mass_Matrix} to block encode the matrix $\bar{\mathbf{A}}_2 = \ket{0}^{\otimes n} \otimes (\mathbf{K}_2 \sqrt{\mathbf{M}^{-1}} \otimes \sqrt{\mathbf{M}^{-1}} )$.

\begin{lemma}[Block Encoding $ \ket{0}^{\otimes n} \otimes (\mathbf{K}_2 \sqrt{\mathbf{M}^{-1}} \otimes \sqrt{\mathbf{M}^{-1}} )$] 
\label{Block Encoding: state_tensored_with_K2_inverse_mass}
\end{lemma}

\begin{proof}
    Note that $\ket{0}^{\otimes n} \otimes (\mathbf{K}_2 \sqrt{\mathbf{M}^{-1}} \otimes \sqrt{\mathbf{M}^{-1}} ) = (\ket{0}^{\otimes n} \otimes \mathbf{K}_2) (\sqrt{\mathbf{M}^{-1}} \otimes \sqrt{\mathbf{M}^{-1}} )$. We can apply Lemma 53 from \cite{gilyen2019quantum}, and construct an $(\frac{\alpha_2}{m_{\min}},a+b, \epsilon')$ block encoding for the matrix product using $\mathcal{O}(1)$ queries to the unitaries $\mathcal{U}(\sqrt{\mathbf{M}^{-1}} \otimes \sqrt{\mathbf{M}^{-1}})$, and $\mathcal{U}(\ket{0}^{\otimes n} \otimes \mathbf{K}_2)$. Here, $\alpha = d\max_{ij}|\mathbf{K}_2(i,j)|$, where $d$ is the maximum number of elements in any row or column of the matrix $\mathbf{K}_2$.
\end{proof}

\begin{lemma}[Block Encoding $(\ket{0}\bra{0})^{\otimes n} \otimes (\sqrt{\mathbf{M}^{-1}} )$] 
\label{Block Encoding: project_tensored_with_inverse_mass_matrix}
    
\end{lemma}
\begin{proof}
   We can construct a $(\frac{1}{\sqrt{m_{\min}}}, 2\log(N),\epsilon')$ block encoding of the tensor product by applying Proposition 4 from \cite{takahira2021quantum}. Constructing this block encoding require $\mathcal{O}(1)$ uses of the unitaries $\mathcal{U}(\sqrt{\mathbf{M}^{-1}})$, and $\mathcal{U}((\ket{0}\bra{0})^{\otimes n} )$.
\end{proof}

\begin{lemma}[Block Encoding $\bar{\mathbf{A}}_2 = \ket{0}^{\otimes n} \otimes ( \sqrt{\mathbf{M}^{-1}} \mathbf{K}_2 \sqrt{\mathbf{M}^{-1}} \otimes \sqrt{\mathbf{M}^{-1}} ) $] There exists unitary matrix $\mathbf{\mathcal{U}}(\bar{\mathbf{A}}_2)$ such that $\| \alpha_5(\bra{0} \otimes \mathbf{I})\mathbf{\mathcal{U}}(\bar{\mathbf{A}}_2)(\ket{0}\otimes \mathbf{I}) - \bar{\mathbf{A}}_2 \| \leq \frac{\alpha_2}{m_{min}}\epsilon' + \frac{1}{\sqrt{m_{min}}}\epsilon'$. Here, $\alpha_5 = \frac{\alpha_2}{m_{min}^{3/2}}$
\label{Block Encoding: A_2_bar}
\end{lemma}
\begin{proof}
    Note that, the matrix $\bar{\mathbf{A}}_2$ can be written as a product of the matrices  $(\ket{0}\bra{0})^{\otimes n} \otimes (\sqrt{\mathbf{M}^{-1}} )$, and  $ \ket{0}^{\otimes n} \otimes (\mathbf{K}_2 \sqrt{\mathbf{M}^{-1}} \otimes \sqrt{\mathbf{M}^{-1}} )$. The unitary block encodings for these matrices are discussed in \autoref{Block Encoding: project_tensored_with_inverse_mass_matrix} and \autoref{Block Encoding: state_tensored_with_K2_inverse_mass}. Hence, applying \autoref{Block Encoding: project_tensored_with_inverse_mass_matrix}, \autoref{Block Encoding: state_tensored_with_K2_inverse_mass} and Lemma 53 from \cite{gilyen2019quantum}, we can construct a block encoding for the matrix $\bar{\mathbf{A}}_2$. The resultant block encoding constant is product of block encoding constants of the submatrices, and is given by: $\alpha_5 = \frac{\alpha_2}{m_{min}^{3/2}}$
\end{proof}
Now we have to create a block encoding for the matrix $\bar{\mathbf{A} }_2 \bar{\mathbf{D}}_2 $. Note that, however the matrix $\bar{\mathbf{D}}_2\in \mathbb{C}^{N^2 \times N^4}$ is a rectangular matrix and the matrix $\bar{\mathbf{A} }_2 \in \mathbb{C}^{N^2 \times N^2}$ is square. Therefore, we pad these matrices with zeroes, and create block encoding for the padded matrix containing $\bar{\mathbf{A} }_2 \bar{\mathbf{D}}_2 $ at the top left block. \\

\begin{lemma}[Block Encoding $ \ket{0}^{\otimes 2n} \otimes \bar{\mathbf{A}}_2 
\bar{\mathbf{D}}_2$] 
\label{block_encoding_A2D2_padded}
\end{lemma}
\begin{proof}
Note that,
    \begin{align}
        \ket{0}^{\otimes 2n} \otimes \bar{\mathbf{A}}_2 \bar{\mathbf{D}}_2 &=  \left[(\ket{0}\bra{0})^{\otimes 2n}  \otimes \bar{\mathbf{A}}_2\right] \left[ \ket{0}^{\otimes 2n} \otimes \bar{\mathbf{D}}_2 \right]
    \end{align}
Hence, we first Block encode the tensor products $ \left[(\ket{0}\bra{0})^{\otimes 2n}  \otimes \bar{\mathbf{A}}_2\right]$, and $\left[ \ket{0}^{\otimes 2n} \otimes \bar{\mathbf{D}}_2 \right]$, and then Block encode the product to construct the desired block encoding. We can construct an $(\alpha_5,4n+3,\epsilon'(\frac{\alpha_2}{m_{min}} + \frac{1}{\sqrt{m_{min}}}))$ encoding of the submatrix $\left[(\ket{0}\bra{0})^{\otimes 2n}  \otimes \bar{\mathbf{A}}_2\right]$ by applying Proposition 4 from \cite{takahira2021quantum}. To block encode the padded matrix $\ket{0}^{\otimes 2n} \otimes \bar{\mathbf{D}}_2$, we can apply Lemma 48 from \cite{gilyen2019quantum}. Now, we can apply Lemma 53 from \cite{gilyen2019quantum}, to block encode the product of two submatrices. The resulting block emcoding constant is given by $\frac{\alpha_5}{k^{(1)}_{min}}$.
\end{proof}

\begin{lemma}[Block Encoding $(\bar{\mathbf{H}}_2)_{pad} = \ket{0}^{\otimes{2n + 1}} \otimes \bar{\mathbf{H}}_2$]
\label{Block_encode_H_2_nonlinear_oscillator}
    There exists a unitary matrix $\mathbf{\mathcal{U}}((\bar{\mathbf{H}}_2)_{pad} )$ such that $\| \alpha(\bra{0}\otimes\mathbf{I})\mathbf{\mathcal{U}}((\bar{\mathbf{H}}_2)_{pad} ) (\ket{0} \otimes \mathbf{I}) -  (\bar{\mathbf{H}}_2)_{pad}\| \leq \epsilon$. Construction of this block encoding require $\mathcal{O}(1)$ queries to the unitary matrices block encoding $ \ket{0}^{\otimes 2n} \otimes \bar{\mathbf{A}}_2 
\bar{\mathbf{D}}_2 $, $(\ket{0}\otimes \mathbf{C}_1) $. 
\end{lemma}
\begin{proof}
Recall that, 
\begin{align}
    (\bar{\mathbf{H}}_2)_{pad} &= \left[(\ket{0}\otimes \mathbf{C}_1) \otimes (\ket{0}^{\otimes{2n}}  \otimes ( - \bar{\mathbf{A}}_2 \bar{\mathbf{D}}_2 )) \right]\left[\mathbf{I}\otimes \bar{\mathbf{P}}_2 \right] \nonumber 
\end{align}
Let $\mathbf{E} = \left[(\ket{0}\otimes \mathbf{C}_1) \otimes (\ket{0}^{\otimes{2n}}  \otimes ( - \bar{\mathbf{A}}_2 \bar{\mathbf{D}}_2 )) \right]$, and $\mathbf{F} = \left[\mathbf{I}\otimes \bar{\mathbf{P}}_2 \right]$. Our strategy here is to construct block encodings for the matrices $\mathbf{E} = \left[(\ket{0}\otimes \mathbf{C}_1) \otimes (\ket{0}^{\otimes{2n}}  \otimes ( - \bar{\mathbf{A}}_2 \bar{\mathbf{D}}_2 )) \right] $, $\mathbf{F} = \left[\mathbf{I}\otimes \bar{\mathbf{P}}_2 \right]$ and then construct a block encoding for the matrix product $\mathbf{E}\mathbf{F}$ using these block encodings. Since the matrix $\mathbf{E}$ is a tensor product of two matrices, the block encoding for $\mathbf{E}$ can be constructed by applying Proposition 4 from \cite{takahira2021quantum}. To apply Proposition 4, we need block encodings for the matrices $(\ket{0}\otimes \mathbf{C}_1)$, and $(\ket{0}^{\otimes{2n}}  \otimes ( - \bar{\mathbf{A}}_2 \bar{\mathbf{D}}_2 ))$. It is trivial to construct a block encoding for the matrix $(\ket{0}\otimes \mathbf{C}_1)$ and the block encoding for the matrix $(\ket{0}^{\otimes{2n}}  \otimes ( - \bar{\mathbf{A}}_2 \bar{\mathbf{D}}_2 ))$ can be constructed by applying \autoref{block_encoding_A2D2_padded}. Since $\bar{\mathbf{P}}_2$ is a unitary matrix, the matrix $\mathbf{F}$ is also unitary, and hence we do not need an additional block encoding step. 

\end{proof}

    
\NOQuantumAlgorithmThm*

\begin{proof}
Let the governing differential equation of the nonlinear oscillator system be:
\begin{equation}
\label{eq: nonlinear_oscillator_gov}
    \mathbf{M}\ddot{\mathbf{x}} = -\mathbf{K_1} \mathbf{x} + \mathbf{K_2}\mathbf{x}\otimes\mathbf{x}
\end{equation}
Applying a coordinate transformation, we rewrite \autoref{eq: nonlinear_oscillator_gov} as:\\
\begin{equation}
    \label{eq: nonlinear_oscillator_padded_gov}
    \ddot{\mathbf{u}} = -\mathbf{A}_1 \mathbf{u} + \mathbf{A}_2 \mathbf{u} \otimes \mathbf{u} 
\end{equation}
Here, $\mathbf{u} = \sqrt{\mathbf{M}}\mathbf{x}$, and $\mathbf{A}_1 = \sqrt{\mathbf{M}}^{-1} \mathbf{K}_1 \sqrt{\mathbf{M}}^{-1} $,and $\mathbf{A}_2 = \sqrt{\mathbf{M}}^{-1} \mathbf{K}_2 [\sqrt{\mathbf{M}}^{-1} \otimes \sqrt{\mathbf{M}}^{-1} ] $. We know from \autoref{section:reduction_to_nonlinear_Schr\"{o}dinger}, that the governing differential equation of the oscillator system in \autoref{eq: nonlinear_oscillator_padded_gov} can be reduced to a nonlinear Schr\"{o}dinger Equation of the form:
\begin{equation}
    \dot{\ket{\bm{\psi}} } = -i\bar{\mathbf{H}}_1 \ket{\bm{\psi}} + \bar{\mathbf{H}}_2 \ket{\bm{\psi}} \otimes \ket{\bm{\psi}}
\end{equation}
The matrices $\bar{\mathbf{H}}_1$, and $\bar{\mathbf{H}}_2$ are constructed from $\mathbf{A}_1$, and $\mathbf{A}_2$. 
Now, applying Theorem \autoref{thm: nonlinear_Schr\"{o}dinger_subnormalized_state}, we can simulate the nonlinear Schr\"{o}dinger Evolution using $G$ queries to unitary block encoding of the symmetrized Carlemann Operator constructed from the matrices $\bar{\mathbf{H}}_1$ and $\bar{\mathbf{H}}_2$. We can construct this block encoding by first applying \autoref{Block_Encode_H1_nonlinear_oscillator}, and \autoref{Block_encode_H_2_nonlinear_oscillator} to construct unitary block encodings of $\bar{\mathbf{H}}_1$, and $\bar{\mathbf{H}}_2$ using oracles defined in \autoref{Def:Oracles-nonlinear-oscillator-system} to access elements of the matrices $\mathbf{M}$, $\mathbf{K}_1$, and $\mathbf{K}_2$. The unitary block encoding of the symmetrized Carlemann Operator $\mathcal{U}(\mathbf{Q}(\eta))$ can be then constructed by applying \autoref{block_encoding_q_eta}. The number of queries to $\mathcal{U}(\mathbf{Q}(\eta))$, as per \autoref{thm: nonlinear_Schr\"{o}dinger_subnormalized_state} is given by:

\begin{align}
    G = \mathcal{O}\left( \alpha k^2 t + \frac{(k-1)}{2}\log\left(\left[\frac{\|\mathbf{H}_2 \| k\left( k +1 \right)}{2} \left(1 + \beta\frac{1-\beta^k}{(1-\beta)\epsilon} \right)t  \right] \right) \right)
\end{align}
Note that, $\|\mathbf{H}_1 \| \leq d \max_{i,j}(\sqrt{\frac{k_{ij}}{m_j}})$, where $d$ is the sparsity of matrix $\mathbf{K}_1$. Furthermore,
\begin{align}
    \|\mathbf{H}_2\| \leq \| \mathbf{A}_2\|\| \mathbf{D}_2\| \leq \| \mathbf{K}_2\| \|\sqrt{M^{-1}} \|^3 \| \mathbf{D}_2\| \leq \frac{d k^{(2)}_{\max} }{k^{(1)}_{\min} (m_{\min})^{3/2}}
\end{align}

where $k^{(2)}_{\max} = \max_{ij}\mathbf{K}_2(i,j)$, and $m_{\min} = \min_i \mathbf{M}(i,i)$, and $d$ is the sparsity of the matrix $\mathbf{K}_2$. Substituting the above norm bounds, we get:

\begin{equation}
    G=  \mathcal{O}\left( \alpha k^2 t + \frac{(k-1)}{2}\log\left(\left[\frac{d k^{(2)}_{\max} k\left( k +1 \right)}{2k^{(1)}_{\min} (m_{\min})^{3/2}} \left(1 + \beta\frac{1-\beta^k}{(1-\beta)\epsilon} \right)t  \right] \right) \right)
\end{equation}
The truncation order $k$, is given by:

\begin{align}
k &\in \frac{W\left( \frac{\epsilon}{CT} \cdot R_r \log(R_r) \right)}{\log(R_r)} \\ \nonumber
 &\in \widetilde{\mathcal{O}} \left( \frac{\log(\frac{CT}{\epsilon})}{\log(1/R_r)}\right)
\end{align}

Here, $W(.)$ is the Lambert W function, the parameter, $C = \|\mathbf{H}_2\|\beta^2 \leq \frac{d k^{(2)}_{\max} \beta^2 }{2k^{(1)}_{\min} (m_{\min})^{3/2}} $ and $R_r = \frac{4 e \beta d k^{(2)}_{\max}}{k^{(1)}_{\min} (m_{\min})^{3/2}\Delta }$, where $\Delta$ is defined as follows:
     \[
\Delta := \inf_{k \in [n]} \inf_{\substack{m_j \geq 0 \\ \sum_{j=1}^n m_j \geq 2}}
\left(|
\lambda_k - \sum_{i=1}^n m_j \lambda_j
|\right)
\]
\end{proof}

\autoref{thm:nonlinear_oscillator_simulation} tells us that the nonlinear oscillator simulation problem can be solved by reducing it to nonlinear Schr\"{o}dinger equation, and then applying Algorithm \autoref{alg:nonlinear_Schr\"{o}dinger} to simulate the resultant nonlinear Schr\"{o}dinger equation. The complexity of simulation is proportional to the square of the carlemann truncation order, which is logarithmic in evolution time and error tolerance when the nonlinear oscillator system satisfies non-resonance conditions.


\section{Time Dependent Stiffness Matrix}

In this section, we examine the problem outlined in \autoref{prblm time_var_stiffness}, which involves simulating oscillator systems whose stiffness matrices vary with time. To solve this problem, we show that a nonlinear oscillator system can be constructed such that a subspace of the system evolves according to the dynamics of the oscillator system with time-dependent stiffness matrix. By appropriately selecting the parameters of this nonlinear system, we can emulate the dynamics of oscillator systems with time-dependent stiffness matrices. The formal statement of the problem is given below: 
\label{sec:time_dependent_stiffness}
\begin{problem}
\label{prblm time_var_stiffness}
(Time-varying Stiffness Matrix) Let $\mathbf{M}$ be an $N \times N$ diagonal matrix with positive entries, and let $\mathbf{K}(t)$ be an $N \times N$ negative definite matrix. Assume we are given oracle access to $\mathbf{M}$, oracle access to Fourier coefficients of terms in $\mathbf{K}(t)$ and a state preparation oracle that prepares a Quantum state proportional to $[\mathbf{x}(t=0), \dot{\mathbf{x}}(t=0) ]^T$. Our goal is to solve the following second-order differential equation: 
\begin{equation}
    \label{prblm-eq-time_variying_stiffness}
    \mathbf{M}\ddot{\mathbf{x}} = -\mathbf{K}(t)\mathbf{x} 
\end{equation}
and prepare a Quantum State $\ket{\psi{(t)}}$ such that $\Pi\ket{\psi{(t)}}  = [\mathbf{x},\dot{\mathbf{x}}]^T$, where 
$\Pi$ is a projection matrix that maps the evolved quantum state to the state of classical dynamical system defined in Eq \eqref{prblm-eq-time_variying_stiffness}.
\end{problem}

The governing differential equation of the oscillator systems with time-dependent stiffness matrices can be written as:
\begin{align}
    \label{eq: time_dependent_stiffness_matrix}
    \mathbf{M}\ddot{\mathbf{x}} = -\mathbf{K}(t)\mathbf{x}
\end{align}
Here, $\mathbf{M}$ is a diagonal matrix with positive masses $m_i$ and $\mathbf{K}(t)$ is a negative definite matrix $\forall t \in [0,T]$. The diagonal entries of $\mathbf{K}(t)$ are given by: $k_{jj} = -\sum_i k_{ji}(t)$ and the off-diagonal entires are $k_{ij}(t) = k_{ji}(t)$. First, we show that Eq \eqref{eq: time_dependent_stiffness_matrix} can be reduced to nonlinear oscillator system with quadratic nonlinearity. Then, we apply lemmas introduced in previous sections to reduce the governing differential equations of the nonlinear oscillator system to quantum evolution. Finally, we apply Hamiltonian simulation methods to solve the resultant time independent Schr\"{o}dinger equation and thereby solving Eq \eqref{eq: time_dependent_stiffness_matrix}. \\

First, let's assume that the time-dependent spring coefficients can be expressed in the following form:\\
\begin{align}
\label{eq: time_varying_stiffness}
    k_{ij}(t) = \alpha_{ij,0} + \sum_{l=1}^L  \alpha_{ij,l} y_{ij,l}(t)
\end{align}
where $\alpha_{ij,0} = 0$ if $i\neq j$, and the time-dependent function $y_{ij,l}(t)$ is defined as:
\begin{equation}
    y_{ij,l}(t) = A_{ij,l} \cos(\omega_{ij,l}t + \phi_{ij,l})
\end{equation}
 To understand how this substitution change the system of equations formed by \eqref{eq: time_dependent_stiffness_matrix}, let us analyze the dynamics of mass $m_i$. 
\begin{align}
m_i \ddot{x}_i &= -\left(\sum_jk_{ij}(t) \right)x_i(t) + \sum_{j\neq i}k_{ij}(t) x_j(t) \\ \nonumber
\end{align}
After substituting \autoref{eq: time_varying_stiffness}, for the time-dependent stiffness functions, we get:
\begin{align}
\label{eq:nonlinear_coupling_dynamics}
    m_i \ddot{x}_i &= -\alpha_{ii,0}x_i -\left(\sum_j \sum_{l=1}^m  \alpha_{ij,l} y_{ij,l}(t) \right) x_i(t) + \sum_{j\neq i} \left(  \sum_{l=1}^m  \alpha_{ij,l} y_{ij,l}(t)\right)  x_j(t) \\ \nonumber
\end{align}
Note that \autoref{eq:nonlinear_coupling_dynamics} involves quadratic nonlinear terms, with nonlinear coupling between $y_{ij,l}$ and the displacement of mass $x_i$, and other masses $x_{j\neq i }$. Hence, we can construct a higher dimensional nonlinear oscillator system  such that the dynamics of a subspace evolves according to \autoref{eq: time_dependent_stiffness_matrix}. The parameters of the nonlinear oscillator system have to be carefully chosen to emulate the dynamics of the oscillator system with time-varying stiffness matrices. Now we formalize this intuitive idea in \autoref{thm:quad_nonlinear_to_time_dependent_K}.

\begin{theorem}[Construction of larger Coupled Nonlinear Oscillator System]
\label{thm:quad_nonlinear_to_time_dependent_K}
  Given a dynamical system of the form $\mathbf{M}\ddot{\mathbf{x}} = -\mathbf{K}(t)\mathbf{x}$, where $\mathbf{K}(t)$ is time dependent stiffness matrix, satisfying $ \mathbf{K}(t) \succ 0$, $  \forall t \in [0,T]$, there exists a higher dimensional nonlinear coupled oscillator system with dynamics: $\mathbf{M}'\ddot{\mathbf{z}} = \mathbf{K}_1 \mathbf{z} + \mathbf{K}_2 \mathbf{z} \otimes \mathbf{z}$ and initial conditions $\mathbf{z}(0) = \mathbf{z}_0$ and $\dot{\mathbf{z}}(0) = \dot{\mathbf{z}}_0$ such that $\forall t \in [0,T]$, s.t $\mathbf{P}\mathbf{z}(t) =  \ket{0}\otimes\mathbf{x}(t)$, where the matrix $\mathbf{P} = \ket{\bm{0}}\bra{\bm{0}} \otimes \mathbf{I}$. 
\end{theorem}
 
\begin{proof}
 Let the state vector of the nonlinear oscillator system be defined as:
\[
{\mathbf{z}} = \ket{0}\otimes {\mathbf{x}} +  \sum_{i\leq j} \ket{f(i,j,N)}\otimes {\mathbf{y}_{ij}} \in \mathbb{R}^{ \frac{N^2(N+1)}{2} + N }
\]
where $ f(i,j,N): \{(i,j) \mid i \leq j\} \rightarrow \{0,1,2,,\frac{N(N+1)}{2} -1\} $ 
is a one-to-one function that maps each pair $ (i, j) $ with $ 1 \leq i \leq j \leq N $  
to a unique index in the range $ \{0,1, 2, \dots, \frac{N(N+1)}{2} -1 \} $.
\\

First, note that,

\begin{align}
    {\mathbf{z}} \otimes {\mathbf{z}} &=  \left(\ket{0}\otimes{\mathbf{x}} +  \sum_{i\leq j} \ket{f(i,j,N)}\otimes {\mathbf{y}_{ij}} \right) \otimes \left(\ket{0}\otimes{\mathbf{x}} +  \sum_{i\leq j} \ket{f(i,j,N)}\otimes {\mathbf{y}_{ij}} \right) \\ \nonumber
    &=  \ket{0}\otimes{\mathbf{x}} \otimes\ket{0}\otimes{\mathbf{x}} + \sum_{i\leq j} \ket{f(i,j,N)} \otimes{\mathbf{y}_{ij}}\otimes\ket{0}\otimes{\mathbf{x}} \\ \nonumber 
    &\quad + \sum_{i\leq j} \ket{0}\otimes{\mathbf{x}}\otimes\ket{f(i,j,N)} \otimes{\mathbf{y}_{ij}} + \sum_{i \leq j} \sum_{k \leq l}  \ket{f(i,j,N)} \otimes{\mathbf{y}_{ij}} \otimes\ket{f(k,l,N)}\otimes{\mathbf{y}_{kl}} 
\end{align}
\\
We further define a transformation $\mathcal{T}$ which would be useful in the definition of the stiffness matrices. \\
\begin{align}
    \mathcal{T}( {\mathbf{z}} \otimes {\mathbf{z}}) &= \ket{0}\otimes\ket{0} \otimes{\mathbf{x}}\otimes{\mathbf{x}} + \sum_{i\leq j} \ket{f(i,j,N)}\ket{0} \otimes{\mathbf{y}_{ij}}\otimes{\mathbf{x}} \\ \nonumber
     &+  \sum_{i\leq j} \ket{0}\ket{f(i,j,N)}\otimes{\mathbf{x}} \otimes{\mathbf{y}_{ij}} + \sum_{i \leq j} \sum_{k \leq l}  \ket{f(i,j,N)} \ket{f(k,l,N)}\otimes{\mathbf{y}_{ij}} \otimes{\mathbf{y}_{kl}} 
\end{align}
Now, let us consider a nonlinear oscillator system with the following dynamics:
\begin{equation}
    \mathbf{M}'\ddot{\mathbf{z}} = -\mathbf{K}_1 \mathbf{z} + \mathbf{K}_2 \mathbf{z} \otimes \mathbf{z}
\end{equation}
Our goal is to choose the matrices $\mathbf{M}'$,  $\mathbf{K}_1$, and $\mathbf{K}_2$ such that $\mathbf{x}(t)$ evolves according to the dynamics of time time-dependent oscillator system. First, we choose the mass matrix of the higher-dimensional oscillator system, which is a diagonal block matrix defined as follows:

\begin{equation}
    \mathbf{M}' = \ket{0}\bra{0} \otimes \mathbf{M} + \sum_{i\leq j} \ket{f(i,j,N)} \bra{f(i,j,N)} \otimes \mathbf{I}
\end{equation}
The top left block of the mass matrix, $\mathbf{M}'$ is $\mathbf{M}$, which is the mass matrix of the time-dependent oscillator system. The mass matrix of the auxiliary system is chosen to be the identity matrix. Next, we turn to the stiffness matrices, starting with the linear component, $\mathbf{K}_1$. To introduce the required time dependence, we select $\mathbf{K}_1$ to have the following form. 
\begin{align}
 \mathbf{K}_1 =  \ket{\mathbf{0}}\bra{\mathbf{0}} \otimes \mathbf{K}_{00} + \sum_{i\leq j} \ket{f(i,j,N)}\bra{f(i,j,N)} \otimes \mathbf{K}_{ij}^{(1)}
\end{align}
Here, $\mathbf{K}_{00}, \mathbf{K}_{ij}^{(1)} \in \mathbb{R}^{N \times N}$ are diagonal matrices, with positive real-valued entries. Specifically, the matrix $\mathbf{K}_{00}$ is defined as:
\begin{equation}
    \mathbf{K}_{00} =  \begin{bmatrix}
\alpha_{11,0} & 0 & \cdots & 0 \\
0 & \alpha_{22,0} & \cdots & 0 \\
\vdots & \vdots & \ddots & \vdots \\
0 & 0 & \cdots & \alpha_{NN,0}
\end{bmatrix}
\end{equation}
We choose the parameters of $\mathbf{K}_{ij}^{(1)} $ depending upon the frequency parameters of the time-dependent stiffness function $k_{ij}(t)$. 
\begin{equation}
    \mathbf{K}_{ij}^{(1)} = \begin{bmatrix}
\omega_{ij,1}^2 & 0 & \cdots & 0 \\
0 & \omega_{ij,2}^2 & \cdots & 0 \\
\vdots & \vdots & \ddots & \vdots \\
0 & 0 & \cdots & \omega_{ij,N}^2
\end{bmatrix} 
\end{equation}
Now we consider the elements of the coupling matrix $\mathbf{K}_2$. If we choose $\mathbf{K}_2$ of this form:
\begin{align}
 \mathbf{K}_2 &= \sum_{ij}\left(\ket{0}\bra{0} \bra{f(i,j,N)} \otimes \mathbf{K}_{ij}^{(2)} \right)\mathcal{T}
 \end{align}
then, 
\begin{align}
    \mathbf{K}_2 (\mathbf{z} \otimes \mathbf{z}) = \sum_{i \leq j} \ket{0} \otimes \mathbf{K}_{ij}^{(2)} \ket{\mathbf{x}} \ket{\mathbf{y}_{ij}}
\end{align}
Now, we define $\mathbf{K}_{ij}^{(2)}$ as:\\
when $i=j$,\\
\[
\bra{i-1} \mathbf{K}_{ij}^{(2)} = \left[ \underbrace{0, 0, 0, \ldots, 0}_{(i-1)L \text{ zeros}}, -\alpha_{ii,1}, -\alpha_{ii,2}, \ldots, -\alpha_{ii,l}, \underbrace{0, 0, 0, \ldots, 0}_{N^2-iL \text{ zeros}} \right]
\]
for all other $k$,
\[
\bra{k} \mathbf{K}_{ij}^{(2)} = \left[ \underbrace{0, 0, 0, \ldots, 0}_{N^2 \text{ zeros}}\right]
\]

when $i \neq j$,\\
\[
\bra{i-1} \mathbf{K}_{ij}^{(2)} = \left[ \underbrace{0, 0, 0, \ldots, 0}_{(i-1)L \text{ zeros}}, -\alpha_{ij,1}, -\alpha_{ij,2}, \ldots, -\alpha_{ij,l}, \underbrace{0, 0, 0, \ldots, 0}_{(j-i-1)L \text{ zeros}},  \alpha_{ij,1}, \alpha_{ij,2}, \ldots, \alpha_{ij,l},\underbrace{0, 0, 0, \ldots, 0}_{N^2-jL \text{ zeros}}, \right]
\]
\[
\bra{j-1} \mathbf{K}_{ij}^{(2)} = \left[ \underbrace{0, 0, 0, \ldots, 0}_{(i-1)L \text{ zeros}}, \alpha_{ij,1}, \alpha_{ij,2}, \ldots, \alpha_{ij,l}, \underbrace{0, 0, 0, \ldots, 0}_{(j-i-1)L \text{ zeros}},  -\alpha_{ij,1}, -\alpha_{ij,2}, \ldots, -\alpha_{ij,l},\underbrace{0, 0, 0, \ldots, 0}_{N^2-jL \text{ zeros}}, \right]
\]
for all other $k$,
\[
\bra{k} \mathbf{K}_{ij}^{(2)} = \left[ \underbrace{0, 0, 0, \ldots, 0}_{N^2 \text{ zeros}}\right]
\]

Our choice of $\mathbf{K}_{ij}^{(2)}$ results in the following dynamical equation for $x_j$:
\begin{equation}
    m_j \ddot{x}_j =      -\left(\sum_{i}\alpha_{ji}^Ty_{ji} x_j + \sum_l \alpha_{jl}^Ty_{jl}x_{j+l}  + \sum_l \alpha_{jl}^Ty_{jl} x_{j-l} \right)
\end{equation}

Since $k_{ij}(t) = \sum_{l=1}^m  \alpha_{ij,l} y_{ij,l}(t) $, Our choice of $\mathbf{K}_2$ introduces the necessary nonlinear coupling to simulate the effect of time-dependent stiffness coefficients.\\

Now, since 
\begin{align}
 \mathbf{K}_1 =   \ket{\mathbf{0}}\bra{\mathbf{0}} \otimes \mathbf{K}_{00} + \sum_i \ket{f(i,j,N)}\bra{f(i,j,N)} \otimes \mathbf{K}_{ij}^{(1)},
\end{align}
the dynamics of $y_{ij},l(t)$ is given by:
\begin{equation}
  \ddot{y}_{ij,l}(t)  = -\omega_{ij,l} ^2 y_{ij,l}(t)
\end{equation}
Hence,
\begin{equation}
 {y}_{ij,l}(t)  =  A_{ij,l} \cos( \omega_{ij,l}t + \phi_{ij,l})
\end{equation}
We choose appropriate initial conditions for $ {y}_{ij,l}(t) $ to create desired amplitudes $A_{ij,l}$ and phase factors $ \phi_{ij,l}$. Thus, our choice of stiffness matrices $\mathbf{K}_1, \mathbf{K}_2$ and appropriate initial conditions for $\mathbf{Y}_{ij}$ results in a nonlinear oscillator system whose subspace evolves according to $\mathbf{M}\ddot{\mathbf{x}} = -\mathbf{K}(t) \mathbf{x}$. 
\end{proof}


        


    

\begin{lemma}
\label{lemma:norm_bound_time_dep_oscillator}
Let \( \mathbf{M} \ddot{\mathbf{x}}(t) = \mathbf{K}(t) \mathbf{x}(t) \) be a second-order differential equation, where \( \mathbf{M} \) is a diagonal matrix containing positive entries and \( \mathbf{K}(t) \) is a time-dependent negative definite matrix . Define the augmented state vector \( \mathbf{y}(t) = \begin{pmatrix} \mathbf{x}(t) \\ \dot{\mathbf{x}}(t) \end{pmatrix} \), and assume that \( \mu(\mathbf{A}(t)) \) is the logarithmic norm of the augmented matrix 
\[
\mathbf{A}(t) = \begin{pmatrix} \mathbf{0} & \mathbf{I} \\ \mathbf{M}^{-1} \mathbf{K}(t) & \mathbf{0} \end{pmatrix}.
\]
Furthermore, let us assume that \( \mu(\mathbf{A}(t)) \leq \beta, \forall t \in [0, T] \),
Then the norm \( \|\mathbf{x}(t)\| \) is bounded by:
\[
\|\mathbf{x}(t)\| \leq \|\mathbf{y}(0)\| \exp\left( \beta t\right),
\]
where \( \|\mathbf{y}(0)\| = \sqrt{\|\mathbf{x}(0)\|^2 + \|\dot{\mathbf{x}}(0)\|^2} \).
If $\beta \leq 0$, then $\|\mathbf{x}(t)\| \leq  \sqrt{\|\mathbf{x}(0)\|^2 + \|\dot{\mathbf{x}}(0)\|^2}  $
\end{lemma}

\begin{proof}
The second-order system \( \mathbf{M} \ddot{\mathbf{x}}(t) = \mathbf{K}(t) \mathbf{x}(t) \) is reduced to a first-order system:
\[
\dot{\mathbf{y}}(t) = \mathbf{A}(t) \mathbf{y}(t),
\]
where \( \mathbf{y}(t) = \begin{pmatrix} \mathbf{x}(t) \\ \dot{\mathbf{x}}(t) \end{pmatrix} \) and
\[
\mathbf{A}(t) = \begin{pmatrix} \mathbf{0} & \mathbf{I} \\ \mathbf{M}^{-1} \mathbf{K}(t) & \mathbf{0} \end{pmatrix}.
\]
The logarithmic norm \( \mu(\mathbf{A}(t)) \) gives an upper bound on the rate of growth of \( \|\mathbf{y}(t)\| \), leading to the differential inequality:
\[
\frac{d}{dt} \|\mathbf{y}(t)\| \leq \mu(\mathbf{A}(t)) \|\mathbf{y}(t)\|.
\]
The solution to this inequality is:
\[
\|\mathbf{y}(t)\| \leq \|\mathbf{y}(0)\| \exp\left( \int_0^t \mu(\mathbf{A}(s)) \, ds \right).
\]
Since \( \|\mathbf{y}(0)\| = \sqrt{\|\mathbf{x}(0)\|^2 + \|\dot{\mathbf{x}}(0)\|^2} \), we obtain the upper bound for \( \|\mathbf{x}(t)\| \) when the logarithmic norm is upper bounded by $\beta \leq 0$.
\end{proof}


\TDOQuantumAlgorithmThm*

\begin{proof}
Let the governing differential equation of the time-dependent oscillator system be given as:
\[
\mathbf{M}\ddot{\mathbf{x}} = -\mathbf{K}(t) \mathbf{x}
\]
Let $\mathbf{G}(t)$ be the time dependent stiffness matrix of the system, where $\mathbf{G}_{ij}(t) = k_{ij}(t) = A_{ij}\cos(\omega_{ij}t + \phi_{ij})$. First, we apply \autoref{thm:quad_nonlinear_to_time_dependent_K} to embed the dynamics of the oscillator system $\mathbf{M} \ddot{\mathbf{x}} = \mathbf{K}(t) \mathbf{x}$ in a higher dimensional nonlinear oscillator system of the form : $\mathbf{M}' \ddot{\mathbf{z}} = \mathbf{K}_1 \mathbf{z} + \mathbf{K}_2 \mathbf{z}\otimes \mathbf{z}$. Here, the state vector of the higher dimensional nonlinear oscillator is given by:
\begin{equation}
 {\mathbf{z}} = \ket{0}\otimes {\mathbf{x}} +  \sum_{i\leq j} \ket{f(i,j,N)}\otimes {\mathbf{y}_{ij}}  
\end{equation}
The mass matrix $\mathbf{M}'$ is defined as:
\begin{equation}
    \mathbf{M}' = \ket{0}\bra{0} \otimes \mathbf{M} + \sum_{ij}\ket{f(i,j,N)}\bra{f(i,j,N)}  \otimes \mathbf{I}
\end{equation}
 The matrices $\mathbf{K}_1$, and $\mathbf{K}_2$ are chosen to guarantee that $ \| \mathbf{P}\mathbf{z}(t) - \mathbf{x}(t) \| \leq \epsilon, \forall t \in [0,\infty)$. From  \autoref{thm:quad_nonlinear_to_time_dependent_K}, the matrix $\mathbf{K}_1$ is:

\begin{equation}
    \mathbf{K}_1 = \ket{\mathbf{0}}\bra{\mathbf{0}} \otimes \mathbf{K}_{00} + \sum_{ij} \ket{f(i,j,N)}\bra{f(i,j,N)} \otimes \mathbf{K}_{ij}^{(1)}
\end{equation}
Here, $\mathbf{K}_{ij}^{(1)}$ is a diagonal matrix containing frequency terms of the time-dependent stiffness matrix. Specifically, the 
\begin{equation}
    \mathbf{K}_{ij}^{(1)} = \begin{bmatrix}
\omega_{ij,1}^2 & 0 & \cdots & 0 \\
0 & \omega_{ij,2}^2 & \cdots & 0 \\
0 & 0 & \ddots & 0 \\
\vdots & \vdots & \vdots & \omega_{ij,l}^2
\end{bmatrix}
\end{equation}
The matrix $\mathbf{K}_2$ is given as:
\begin{align}
 \mathbf{K}_2 &= \sum_{ij} \left(\ket{0}\bra{0} \bra{f(i,j,N)} \otimes \mathbf{K}_{ij}^{(2)} \right)\mathcal{T}
 \end{align}

Now, we apply \autoref{thm:nonlinear_oscillator_simulation} to simulate the dynamics of the resultant nonlinear oscillator system. Simulating the dynamics of the resultant nonlinear oscillator system requires $G$ queries to unitary matrices encoding $\mathbf{K}_1$, and $\mathbf{K}_2$. Here, $G$ is given by: 

\begin{equation}
G \in  \mathcal{O}\left( \alpha k^2 t + \frac{(k-1)}{2}\log\left(\left[\frac{d k^{(2)}_{\max} k\left( k +1 \right)}{2k^{(1)}_{\min} (m_{\min})^{3/2}} \left(1 + \beta\frac{1-\beta^k}{(1-\beta)\epsilon} \right)t  \right] \right) \right)
\end{equation}

where $k$ is the Carlemann truncation order, $\alpha = \max_{ij} \sqrt{\frac{k_{ij}}{m_j}} = \max_{ij} \omega_{ij}$ , $k_{\max}^{(2)} = \max_{ij}\mathbf{K}_2(i,j) = \max_{ij,l}\{\alpha_{ij,l}\}$, $k_{\min}^{(1)} = \min_{ij} \{\omega_{ij}\}$, $\alpha=\max_{ij}\omega_{ij}$ and $\beta$ is the initial total energy of the nonlinear oscillator system. 
The value of $\beta$ depends upon the norm of the initial conditions, which in our case, is given by:
\begin{equation}
  \beta =   \braket{\mathbf{z}(0)|\mathbf{z}(0)} = \braket{\mathbf{x}(0)|\mathbf{x}(0)} + \sum_{ij}\braket{\mathbf{y}_{ij}(0)|\mathbf{y}_{ij}(0)}
\end{equation}
Without loss of generality, we can rescale the initial conditions with a constant $\sqrt{E}$ such that $\beta = \frac{\braket{\mathbf{z}(0)|\mathbf{z}(0)}}{E} < 1$. Hence, the number of queries $G$ is given by:
\begin{align}
    G &\in  \mathcal{O}\left( \alpha k^2 t + k\log\left(\left[\frac{d k^{(2)}_{\max} k\left( k +1 \right)}{2k^{(1)}_{\min} (m_{\min})^{3/2}} \left(1 + \frac{\beta}{(1-\beta)\epsilon} \right)t  \right] \right) \right)
\end{align} 

Since $\mathbf{K}(t)$ is positive definite, from \autoref{lemma:norm_bound_time_dep_oscillator}, the norm of the solution is non-increasing, which is necessary for convergence of the Carlemann linearized system.  The truncation order $k$, is given by:
\begin{align}
k \in \widetilde{\mathcal{O}} \left( \frac{\log(\frac{CT}{\epsilon})}{\log(1/R_r)}\right)
\end{align}

Here, $C  \leq \frac{d k^{(2)}_{\max} \beta^2 }{2k^{(1)}_{\min} (m_{\min})^{3/2}} $
     Here, $R_r = \frac{4 e \beta \|\mathbf{K}_2\|\sqrt{M^{-1}} \|^3 \| \mathbf{D}_2\| \|}{\Delta}$, where $\Delta$ is defined as follows:
     \[
\Delta := \inf_{k \in [n]} \inf_{\substack{m_j \geq 0 \\ \sum_{j=1}^n m_j \geq 2}}
\left(|
\lambda_k - \sum_{i=1}^n m_j \lambda_j
|\right)
\]
Here, $\{\lambda_j\}$ are the eigenvalues of the matrix $\sqrt{\mathbf{M}^{-1}}\mathbf{K}_1\sqrt{\mathbf{M}^{-1}}$, which in this case is the set of frequencies $\{\omega_{ij}\}$. 
\end{proof}

\section{Time Dependent Stiffness Matrix and Time-Dependent External Forces}
\label{sec:time_dependent_force_and_stiffness}

In this section, we focus on \autoref{prblm: time_dependent_forced_oscillator}, where the stiffness matrix is time-varying, and external time-dependent forces act on individual masses. First, we show that the dynamics of the time-dependent forced oscillator system can be embedded in the dynamics of a higher dimensional nonlinear oscillator system with quadratic nonlinearity. Then, we apply \autoref{thm:nonlinear_oscillator_simulation} to simulate the resultant nonlinear oscillator system. Recall that the governing differential equation of the time-dependent forced oscillator system is of the form:

\begin{align}
    \label{eq: time_dependent_stiffness_matrix_and_force}
    \mathbf{M}\ddot{\mathbf{x}} = -\mathbf{K}(t)\mathbf{x} + \mathbf{f}(t)
\end{align}

Here, $\mathbf{M}$ is a diagonal matrix with positive entries corresponding to the value of masses $m_i$ and $\mathbf{K}(t)$ is a positive definite matrix $\forall t \in [0, T]$. The diagonal entries of $\mathbf{K}(t)$ are given by: $k_{jj} = -\sum_i k_{ji}(t)$ and the off-diagonal entires are $k_{ij}(t) = k_{ji}(t)$ as per \autoref{Def: time_dependent_forced_coupled_oscillator_system}. Further, as per \autoref{Def: time_dependent_forced_coupled_oscillator_system}, the time-dependent forces acting on mass $m_i$ is of the form $f_i(t)  =\sum_j f_{ij}\cos(\mu_{ij}t + \xi_{ij})$.\\

First, we assume that the time-dependent spring coefficients can be expressed in the following form:\\

\begin{align}
    k_{ij}(t) = \sum_{l=1}^{L}  \alpha_{ij,l} y_{ij,l}(t)
\end{align}
Here, $ y_{ij,l}(t)$ are basis functions and $\alpha_{ij,l}$ are scalar, real-valued coefficients. The time dependent forces can be written as:
\[
f_i(t) = \sum_{l=1}^L \beta_{i,l}w_{i,l}(t)
\]
To understand how this substitution changes the system of equations formed by \eqref{eq: time_dependent_stiffness_matrix}, let us analyze the dynamics of any mass $m_1$. 

\begin{align}
m_i \ddot{x}_i &= -\left(\sum_j k_{ij}(t) \right)x_i(t) + \sum_{j\neq i}k_{ij}(t) x_j(t) +  \sum_{l=1}^L \beta_{i,l}w_{i,l}(t)\\ \nonumber
    m_i \ddot{x}_i &= -\left(\sum_j\sum_{l=1}^L \alpha_{ij,l} y_{ij,l}(t) \right)x_i(t) + \sum_{i\neq j}\sum_{l=1}^L  \alpha_{ij,l} y_{ij,l}(t)x_j + \sum_{l=1}^L \beta_{i,l}w_{i,l}(t)
\end{align}
Note that this equation of dynamics involves a quadratic nonlinearity, with nonlinear coupling between $y_{ij,l}$ and the displacement of mass $x_j$. There is also an extra linear term corresponding to $w_{i,l}(t)$ which is due to the time-dependent external force acting on individual masses. This implies we can simulate the dynamics of the time-dependent oscillator system in \autoref{eq: time_dependent_stiffness_matrix_and_force} using a nonlinear dynamical system in higher dimensions, with suitable coefficients for nonlinear terms. We formalize this intuitive idea in \autoref{thm:quad_nonlinear_to_time_dependent_K_and_F}.

\begin{theorem}[Construction of larger Coupled Nonlinear Oscillator System]
\label{thm:quad_nonlinear_to_time_dependent_K_and_F}
  Given a dynamical system of the form $\mathbf{M}\ddot{\mathbf{x}} = \mathbf{K}(t)\mathbf{x} + \mathbf{f}(t)$, where $\mathbf{K}(t)$ is time dependent stiffness matrix, satisfying $ \mathbf{K}(t) \prec 0$, $  \forall t \in [0,T]$, there exists a higher dimensional nonlinear coupled oscillator system with dynamics: $\mathbf{M}'\ddot{\mathbf{z}} = -\mathbf{K}_1 \mathbf{z} + \mathbf{K}_2 \mathbf{z} \otimes \mathbf{z}$, with $\mathbf{z} \in \mathbb{R}^{2N + \frac{N^2(N+1)}{2}}$ and initial conditions $\mathbf{z}(0) = \mathbf{z}_0$ and $\dot{\mathbf{z}}(0) = \dot{\mathbf{z}}_0$ such that $\forall t \in [0,T]$,s.t $\mathbf{P}\mathbf{z}(t)  =  \ket{0}\otimes\mathbf{x}(t)$, where $\mathbf{P}= \ket{\bm{0}}\bra{\bm{0}} \otimes \mathbf{I}$. 
\end{theorem}
 
\begin{proof}
Consider a nonlinear oscillator system of the following form:
\begin{equation}
\label{eq:nonlinear_oscillator_sys}
    \mathbf{M}'\ddot{\mathbf{z}} = -\mathbf{K}_1\mathbf{z} + \mathbf{K}_2 \mathbf{z} \otimes \mathbf{z}
\end{equation}
Let 

\begin{align}
  {\mathbf{z}}  =  \ket{0} \otimes {\mathbf{x}} + \sum_{i\leq j} \ket{c(i,j,N)}\otimes {\mathbf{y}_{ij}} + \sum_{k} \ket{g(k,N)}\otimes {\mathbf{w}_{k}} + \ket{m}\ket{\mathbf{p}}
\end{align}
Here, $ c(i,j,N): \{(i,j) \mid i \leq j\} \rightarrow \left[ \frac{N (N-1)}{2} \right] $ 
is a one-to-one function that maps each pair $ (i, j) $ with $ 1 \leq i \leq j \leq N $  
to a unique index in the range $ \{ 1, 2, \dots, \frac{N (N-1)}{2} \} $, and $g(k,N)$ maps each index $k$ to a unique index in the range $ \{\frac{N (N-1)}{2} +1,  \frac{N (N-1)}{2} + N+1 \}$, and let $m = \frac{N (N-1)}{2} + N+2 \}$. Further assume, $\dim(\mathbf{x}) = \dim(\mathbf{y}_{ij}) = \dim(\mathbf{w}_k) = \dim(\mathbf{p}) = N$, and $\mathbf{p}$ be a constant unit vector.\\

Recall, that the external time-dependent force acting on a mass $m_k$, can be written as: $ f_k(t) =  \sum_i \beta_i w_{k,i}(t)$. We emulate the effect of this time-dependent external force by evolving additional oscillators such that the displacement vector of these additional oscillators evolve according to the basis functions in the expansion of time-dependent force $f_k(t)$. Specifically, we let:
\begin{align}
    {\mathbf{w}_{k}} = \begin{bmatrix}
        w_{k,1}&
        w_{k,2}&
        \cdots &
        w_{k,q}
    \end{bmatrix}^T
\end{align}
Similar to \autoref{thm:quad_nonlinear_to_time_dependent_K}, the effect of time dependent stiffness elements $k_{ij}(t)$ are created by using nonlinear coupling between the displacement vector $\mathbf{y}_{ij}$ of the auxillary masses and the displacement vector $\mathbf{x}$. To see this, consider the tensor product $  \ket{\mathbf{z}} \otimes \ket{\mathbf{z}} $.


\begin{align}
   {\mathbf{z}} \otimes {\mathbf{z}} = \left( \ket{0} \otimes {\mathbf{x}} + \sum_{i\leq j} \ket{c(i,j,N)}\otimes {\mathbf{y}_{ij}} + \sum_{k} \ket{g(k,N)}\otimes {\mathbf{w}_{k}} + \ket{m}\otimes{\mathbf{p}} \right) \otimes \\ \nonumber
    \left( \ket{0} \otimes {\mathbf{x}} + \sum_{i\leq j} \ket{c(i,j,N)}\otimes {\mathbf{y}_{ij}} + \sum_{k} \ket{g(k,N)}\otimes {\mathbf{w}_{k}} + \ket{m}\otimes{\mathbf{p}} \right) 
\end{align}
Upon simplifying, we get: 

\begin{align}
{\mathbf{z}} \otimes {\mathbf{z}} &=  \ket{0}\otimes {\mathbf{x}}\otimes \ket{0}\otimes {\mathbf{x}} + \sum_{i \leq j} \ket{c(i,j,N)} \otimes {\mathbf{y}_{ij}}\otimes \ket{0} \otimes{\mathbf{x}} + \sum_{k} \ket{g(k,N)} \otimes{\mathbf{w}_{k}} \otimes \ket{0} \otimes{\mathbf{x}}  \\ \nonumber 
&\quad + \ket{m}\otimes{\mathbf{p}}\otimes \ket{0} \otimes{\mathbf{x}} + \ket{0} \otimes{\mathbf{x}} \otimes\sum_{i \leq j} \ket{c(i,j,N)} \otimes{\mathbf{y}_{ij}} \\ \nonumber 
&\quad + \sum_{i \leq j} \sum_{l \leq m} \ket{c(i,j,N)} \otimes{\mathbf{y}_{ij}} \otimes \ket{c(l,m,N)} \otimes {\mathbf{y}_{l,m}} \\ \nonumber 
&\quad + \sum_{i \leq j} \sum_{k} \ket{g(k,N)} \otimes {\mathbf{w}_{k}}  \otimes\ket{c(i,j,N)} \otimes {\mathbf{y}_{ij}} \\ \nonumber 
&\quad + \sum_{i \leq j}\ket{m} \ket{\mathbf{p}} \ket{c(i,j,N)} \otimes{\mathbf{y}_{ij}}\\ \nonumber 
&\quad + \sum_k \ket{0}\otimes{\mathbf{x}}\otimes \ket{g(k,N)} {\mathbf{w}_{k}} + \sum_{i \leq j}\sum_k  \ket{c(i,j,N)} \ket{\mathbf{y}_{ij}}  \ket{g(k,N)} \otimes{\mathbf{w}_{k}} \\ \nonumber
&\quad + \sum_{j}\sum_k  \ket{g(j,N)} \otimes{\mathbf{w}_{j}} \otimes \ket{g(k,N)}\otimes{\mathbf{w}_{k}} + \sum_k \ket{m}\ket{\mathbf{p}}\ket{g(k,N)}\otimes{\mathbf{w}_{k}}  \\ \nonumber
&\quad +\left( \ket{0} \otimes {\mathbf{x}} + \sum_{i\leq j} \ket{c(i,j,N)}\otimes {\mathbf{y}_{ij}} + \sum_{k} \ket{g(k,N)}\otimes {\mathbf{w}_{k}} + \ket{m} \right)\otimes \ket{m}\ket{\mathbf{p}}
\end{align}
Let $\mathcal{T}$ be a transformation acting on $\ket{\mathbf{z}} \otimes \ket{\mathbf{z}}$ which swaps second and third qubit registers. \\
Specifically,
\begin{equation}
    \mathcal{T}(\ket{u_1}\ket{u_2}\ket{u_3}\ket{u_4}) = \ket{u_1}\ket{u_3}\ket{u_2}\ket{u_4}
\end{equation}
Now, we need to choose matrices $\mathbf{K}_1$, and $\mathbf{K}_2$, such that the dynamics of nonlinear oscillator introduce necessary time-dependent stiffness elements and time-dependent external forces. Since the basis functions are independent of each other, we can choose a diagonal block matrix as $\mathbf{K}_1$. 
\begin{equation}
    \mathbf{K}_1 = \ket{0}\bra{0} \otimes \mathbf{K}_{1,00} + \sum_{ij}\ket{c(i,j,N)}\bra{c(i,j,N)} \otimes \mathbf{K}_{1,ij} + \sum_{k}\ket{g(k,N)}\bra{g(k,N)}\otimes\mathbf{K}_{1,k}
\end{equation}
Here, $\mathbf{K}_{1,ij} , \mathbf{K}_{1,k}\in \mathbb{R}^{N \times N}$ are diagonal matrices, with positive real valued entries. We choose the parameters of $\mathbf{K}_{1,ij} $ depending upon the frequency parameters of the time-dependent stiffness function $k_{ij}(t)$. Now, let us consider the elements of the coupling matrix $\mathbf{K}_2$.\\
Now, if we choose $\mathbf{K}_2$ of the form:
\begin{align}
 \mathbf{K}_2 &= \left(\sum_{ij}  (\ket{0}\bra{0} \bra{c(i,j,N)} \otimes \mathbf{K}_{ij} ) + \sum_k \ket{0}\bra{m}\bra{g(k,N)}\bra{\mathbf{p}} \otimes \mathbf{F}_k\right) \mathcal{T}
\end{align}
then, 
\begin{align}
    \mathbf{K}_2 (\mathbf{z} \otimes \mathbf{z}) = \ket{0} \otimes \left( \sum_{i \leq j} \mathbf{K}_{ij} \ket{\mathbf{x}} \ket{\mathbf{y}_{ij}} + \sum_k \mathbf{F}_k \ket{\mathbf{w}_k} \right)
\end{align}
Now, we define $\mathbf{K}_{ij}$ as:\\
when $i=j$,\\
\begin{equation}
\bra{i-1} \mathbf{K}_{ij} = \left[ \underbrace{0, 0, 0, \ldots, 0}_{(i-1)l \text{ zeros}}, -\alpha_{ii,1}, -\alpha_{ii,2}, \ldots, -\alpha_{ii,l}, \underbrace{0, 0, 0, \ldots, 0}_{n^2-il \text{ zeros}} \right]
\end{equation}
for all other $k$,
\begin{equation}
\bra{k} \mathbf{K}_{ij} = \left[ \underbrace{0, 0, 0, \ldots, 0}_{n^2 \text{ zeros}}\right]
\end{equation}

when $i \neq j$,\\
\begin{equation}
\bra{i-1} \mathbf{K}_{ij} = \left[ \underbrace{0, 0, 0, \ldots, 0}_{(i-1)l \text{ zeros}}, -\alpha_{ij,1}, -\alpha_{ij,2}, \ldots, -\alpha_{ij,l}, \underbrace{0, 0, 0, \ldots, 0}_{(j-i-1)l \text{ zeros}},  \alpha_{ij,1}, \alpha_{ij,2}, \ldots, \alpha_{ij,l},\underbrace{0, 0, 0, \ldots, 0}_{n^2-jl \text{ zeros}}, \right]
\end{equation}
\begin{equation}
\bra{j-1} \mathbf{K}_{ij} = \left[ \underbrace{0, 0, 0, \ldots, 0}_{(i-1)l \text{ zeros}}, \alpha_{ij,1}, \alpha_{ij,2}, \ldots, \alpha_{ij,l}, \underbrace{0, 0, 0, \ldots, 0}_{(j-i-1)l \text{ zeros}},  -\alpha_{ij,1}, -\alpha_{ij,2}, \ldots, -\alpha_{ij,l},\underbrace{0, 0, 0, \ldots, 0}_{n^2-jl \text{ zeros}}, \right]
\end{equation}
for all other $k$,
\begin{equation}
\bra{k} \mathbf{K}_{ij} = \left[ \underbrace{0, 0, 0, \ldots, 0}_{n^2 \text{ zeros}}\right]
\end{equation}
This construction is similar to that of \autoref{thm:quad_nonlinear_to_time_dependent_K} and creates the necessary time-dependent stiffness matrix. To create the necessary time-dependent external forces, define $\mathbf{F}_k$ as follows:
\begin{equation}
    \bra{k-1}\mathbf{F}_k = [\beta_{k,1},\beta_{k,2},\cdots\beta_{k,q}]
\end{equation}
and for any $j \neq k$,
\begin{equation}
    \bra{j}\mathbf{F}_k = [0,0\cdots0,0]
\end{equation}
The mass matrix, $\mathbf{M}'$, can be defined as follows:
\begin{equation}
  \mathbf{M}' =   \ket{0}\bra{0} \otimes \mathbf{M} + \sum_{i\leq j} \ket{c(i,j,N)} \bra{c(i,j,N)} \otimes \mathbf{I} + \sum_k \ket{g(k,N)}\bra{g(k,N)}\otimes \mathbf{I}
\end{equation}

Now, after plugging in $\mathbf{M}'$, $\mathbf{K}_1$, and $\mathbf{K}_2$ in \autoref{eq:nonlinear_oscillator_sys}, we get the following set of differential equations:
\begin{align}
\label{eq:reduced_time_dep_eq}
    \ket{0} \otimes \mathbf{M} \ddot{\mathbf{x}} &= -\ket{0}\otimes\mathbf{K}_{1,00}\mathbf{x} + \ket{0}\otimes\sum_{ij} \mathbf{K}_{ij} (\mathbf{x}\otimes \mathbf{y}_{ij}) + \ket{0}\otimes \sum_k \mathbf{F}_k \mathbf{w}_k \\ \nonumber
    \ket{c(i,j,N)} \otimes \ddot{\mathbf{y}}_{ij} &= -\ket{c(i,j,N)} \otimes \mathbf{K}_{1,ij}\mathbf{y}_{ij} \\ \nonumber
    \ket{g(k,N)} \otimes \ddot{\mathbf{w}}_k &= - \ket{g(k,N)}\otimes\mathbf{K}_{1,k}\mathbf{w}_k \\ \nonumber
    \ket{m} \otimes \ddot{\mathbf{p}} &= \mathbf{0}
\end{align}

Note that $\mathbf{y}_{ij}$, and $\mathbf{w}_k$ is coupled with $\mathbf{x}$, and their evolution is influenced by the interaction terms in the differential equations. Specifically, the coupling terms $\mathbf{K}_{ij} (\mathbf{x} \otimes \mathbf{y}_{ij})$ and $\mathbf{F}_k \mathbf{w}_k$ introduce dependencies between the primary system dynamics $\mathbf{x}$ and the auxiliary variables $\mathbf{y}_{ij}$ and $\mathbf{w}_k$. To see this, consider the dynamics of auxillary system.\\

Since $\mathbf{K}_{1,ij}$ is diagonal, the dynamics of $\mathbf{y}_{ij}(t) = [y_{ij,1}(t),\cdots y_{ij,L}(t)]^T$ is given by:
\begin{equation}
    y_{ij,l}(t) = A_{ij,l}\cos(\omega_{ij,l}t +\phi_{ij,l})
\end{equation}
Similarly, the dynamics of $\mathbf{w}_k(t) = [w_{k,1}(t),\cdots w_{k,q}(t)]^T$ is given by:
\begin{equation}
    w_{k,i}(t) = A_{k,i}\cos(\mu_{k,i}t + \xi_{k,i})
\end{equation}
Now, let us consider the dynamics of any mass $m_i$ in the nonlinear coupled oscillator system. Let $i \in \left[N\right]$. Then, $m_i\ddot{x}_i = \bra{i-1}\mathbf{M}\ddot{\mathbf{x}}$.  
\begin{align}
 \bra{i-1} \mathbf{M} \ddot{\mathbf{x}} &= -\bra{i-1}\mathbf{K}_{1,00}\mathbf{x} + \bra{i-1}\sum_{ij} \mathbf{K}_{ij} (\mathbf{x}\otimes \mathbf{y}_{ij}) + \bra{i-1}\sum_k \mathbf{F}_k \mathbf{w}_k 
\end{align}
Applying the definitions for the matrices $\mathbf{K}_{1,00}$, $\mathbf{K}_{ij}$, and $\mathbf{F}_k$, we get the following equation of motion for any mass $m_i$. 
\begin{equation}
     m_i \ddot{x}_i = -\left(\sum_j\sum_{l=1}^L \alpha_{ij,l} y_{ij,l}(t) \right)x_i(t) + \sum_{i\neq j}\sum_{l=1}^L  \alpha_{ij,l} y_{ij,l}(t)x_j + \sum_{l=1}^L \beta_{i,l}w_{i,l}(t)
\end{equation}
Hence, our choice of coupling matrices in the larger nonlinear oscillator system emulates the effect of time-dependent stiffness matrices and time-dependent external forces.
\end{proof}
Applying \autoref{thm:quad_nonlinear_to_time_dependent_K_and_F}, we can embed the dynamics of time dependent forced oscillator system in a high dimensional nonlinear dynamical system. By simulating the dynamics of this higher-dimensional system, we can effectively capture the behavior of the original forced oscillator system. 

\TDFOQuantumAlgorithmThm*
\begin{proof}
Recall that the dynamics of the forced oscillator system can be written as: 
\[
\mathbf{M}\ddot{\mathbf{x}}=\mathbf{K}(t)\mathbf{x}+\mathbf{F}(t),
\]
where our input model gives access to Fourier components of the time varying stiffness and force functions. By applying \autoref{thm:quad_nonlinear_to_time_dependent_K}, we can construct a ${2N + \frac{N^2(N+1)}{2}}$ dimensional nonlinear oscillator system such that a subspace of the nonlinear oscillator system evolves according to the dynamics of the time dependent forced oscillator system. The state vector of the nonlinear oscillator system can be written as: 
\[
{\mathbf{z}}=\ket{0}\otimes{\mathbf{x}}+\sum_{i\leq j}\ket{c(i,j,N)}\otimes{\mathbf{y}_{ij}}+\sum_{k}\ket{g(k,N)}\otimes{\mathbf{w}_{k}}+\ket{m}\otimes{\mathbf{p}},
\]
where the auxiliary states \({\mathbf{y}_{ij}}\) and \({\mathbf{w}_{k}}\) encode the evolution associated with the time-dependent stiffness and force terms, respectively.

The corresponding augmented mass matrix is defined as
\[
\mathbf{M}'=\ket{0}\bra{0}\otimes\mathbf{M}+\sum_{i\leq j}\ket{c(i,j,N)}\bra{c(i,j,N)}\otimes\mathbf{I}+\ket{m}\bra{m}\otimes\mathbf{I}.
\]
 
The linear operator \(\mathbf{K}_1\) is defined as: 
\[
 \mathbf{K}_1 = \ket{0}\bra{0} \otimes \mathbf{K}_{1,00} + \sum_{ij}\ket{c(i,j,N)}\bra{c(i,j,N)} \otimes \mathbf{K}_{1,ij} + \sum_{k}\ket{g(k,N)}\bra{g(k,N)}\otimes\mathbf{K}_{1,k}
\]
where each \(\mathbf{K}_{ij}^{(1)}\) is a diagonal matrix encoding the squared frequencies:
\[
\mathbf{K}_{1,ij}=\begin{bmatrix}
\omega_{ij,1}^2&0&\cdots&0\\[1mm]
0&\omega_{ij,2}^2&\cdots&0\\[1mm]
\vdots&\vdots&\ddots&\vdots\\[1mm]
0&0&\cdots&\omega_{ij,l}^2
\end{bmatrix}.
\]

The quadratic coupling operator \(\mathbf{K}_2\) is defined as:
\begin{align}
 \mathbf{K}_2 &= \left(\sum_{ij}  (\ket{0}\bra{0} \bra{c(i,j,N)} \otimes \mathbf{K}_{ij} ) + \sum_k \ket{0}\bra{m}\bra{g(k,N)}\bra{\mathbf{p}} \otimes \mathbf{F}_k\right) \mathcal{T}
\end{align}
with \(\mathcal{T}\) a transformation that rearranges the registers of \(\ket{\mathbf{z}}\otimes\ket{\mathbf{z}}\) and, the force term is incorporated through additional coupling matrices \(\mathbf{F}_k\) satisfying
\[
\bra{k-1}\mathbf{F}_k=\Bigl[0,\dots,0,\beta_{k,1},\beta_{k,2},\dots,\beta_{k,q},0,\dots,0\Bigr],
\]
and \(\bra{j}\mathbf{F}_k=\mathbf{0}\) for \(j\neq k\).


Now, we can apply the quantum simulation algorithm described in \autoref{thm:nonlinear_oscillator_simulation} to simulate the dynamics of the higher dimensional nonlinear dynamical system, and thereby obtaining a simulation of the forced dynamics such that the projected state \(\mathbf{P}\mathbf{z}(t)\) approximates \(\mathbf{x}(t)\) within error \(\epsilon\) for all \(t\ge 0\). The number of queries \(G\) to unitaries encoding the matrices $\mathbf{K}_1$, and $\mathbf{K}_2$, satisfies
    \[
   G \in  \mathcal{O}\left( \alpha k^2 t + \frac{(k-1)}{2}\log\left(\left[\frac{d k^{(2)}_{\max} k\left( k +1 \right)}{2k^{(1)}_{\min} (m_{\min})^{3/2}} \left(1 + \beta\frac{1-\beta^k}{(1-\beta)\epsilon} \right)t  \right] \right) \right)
    \]
where,
\[
     k \in \mathcal{O}\left(\frac{\log(CT/\epsilon)}{\log(1/R_r)}\right)
     \]
     
Here, $k_{\max}^{(2)} = \max_{ij}\mathbf{K}_2(i,j) = \max_{ij,l}\{\alpha_{ij,l}, \beta_{i,j}\}$, $k_{\min}^{(1)} = \min_{ij} \{\omega_{ij},\mu_{ij}\}$, $\alpha=\max_{ij}\omega_{ij}$ and $\beta$ is the initial total energy of the nonlinear oscillator system. The value of $\beta$ is given by:
\begin{align*}
    \beta &= \frac{1}{2} \left( \dot{\mathbf{z}}^T(0) \mathbf{M}'\dot{\mathbf{z}}(0) + \mathbf{z}^T(0)\mathbf{K}_1\mathbf{z}(0)\right) \\ \nonumber
    &= \frac{1}{2} ( \dot{\mathbf{x}}^T(0) \mathbf{M}\dot{\mathbf{x}}(0) + \sum_{ij} \|\dot{\mathbf{y}}_{ij}(0)\| + \sum_{k} \|\dot{\mathbf{w}}_{k}(0)\| + \\ \nonumber
     &\quad \mathbf{x}^T(0)\mathbf{K}_{1,00}\mathbf{x}(0) + \sum_{ij}\mathbf{y}_{ij}^T(0)\mathbf{K}_{1,ij}\mathbf{y}_{ij}(0) + \sum_k\mathbf{w}_{k}^T(0) \mathbf{K}_{1,k}\mathbf{w}_{k}(0) )
\end{align*}


This completes the proof.
\end{proof}

\appendix

\section{Conclusion}

We have present quantum algorithms for simulating a broad class of classical oscillator systems with time-dependent force and parameters in this paper.  These algorithms are based on a strategy that we call ``Nonlinear Schr\"{o}dingerization'' which involves mapping a dynamical system to a nonlinear Schr\"{o}dinger equation that is then linearized and simulated on a quantum computer. This approach significantly expands the applicability of quantum algorithms for simulating the dynamics of non-conservative and nonlinear oscillator systems that use the encoding of~\cite{babbush2023exponential}. Our key contributions include the development of quantum algorithms for simulating coupled oscillator systems with time-dependent external forces.  Despite using a nonlinear Schr\"{o}dinger equation, we introduce a variant of Carleman linearization that can be used to approximate the nonlinear dynamics by unitary dynamics in a higher dimensional space.  This allows us to, under certain strong assumptions about the spectra of the oscillator,  overcome a key limitation of such approaches that required constant total energy.

We introduced perturbation techniques to embed non-conservative dynamics within larger conservative systems, enabling efficient quantum simulation. Our novel symmetrization methods for the Carlemann linearization allow simulation of nonlinear Schrödinger equations with provable error bounds. Furthermore, we extended these methods to handle systems with time-dependent stiffness matrices, including those with both time-varying parameters and external forces. Our algorithms achieve complexity that is logarithmic in the number of oscillators and almost linear in evolution time when the system properties can be efficiently queried and the initial state can be efficiently prepared.  As this algorithm generalizes existing quantum algorithms for harmonic oscillators, the exponential speedups that we can attain in these cases cannot be matched by classical methods in generality.  Consequently, our work pro a significant improvement over classical methods for simulating large coupled oscillator systems. 


Our work leaves several avenues open for further investigation.  First, since the positions and momenta of oscillator systems are encoded in quantum states, computing the exact value of these quantities would require tomography, and that would be computationally expensive.  These issues could be addressed in part by coupling an amplitude encoding of the variables, such as that used here or in~\cite{somma2025shadow}, to a traditional bit encoding for the observables.  The perturbative approaches introduced here are likely to be essential for dealing with the nonlinearity that emerge when trying to couple such variables as amplitudes.  Second, to apply the algorithms developed to simulate nonlinear dynamical systems and the nonlinear Schr\"{o}dinger Equation, the equations should satisfy certain constraints, which limits the generality of the method.  Generalization of these ideas to tackle shadow Hamiltonian simulation would provide a useful extension of these ideas.  Lastly, nonlinear Schr\"{o}dingerization necessarily imposes strict restrictions on the nonlinear Schr\"{o}dinger equations considered in this work.  This is necessary because without making such assumptions, we could violate state discrimination lower bounds as argued in~\cite{Xue2023Further}.  Understanding the full range of cases where linearization can be used to efficiently reduce nonlinear dynamics to linear dynamics will allow us to understand the scope of improvements that nonlinear Schr\"odingerization can provide as well as giving us insight into the computational power of nonlinearity and its strengths and weaknesses.

\section{Acknowledgements}
This
material is primarily based upon work supported by the U.S. Department of Energy, Office of Science, National Quantum Information Science Research Centers, Co-design Center for Quantum Advantage (C2QA) under contract number DE- SC0012704 (PNNL FWP 76274). This research is also supported by PNNL’s
Quantum Algorithms and Architecture for Domain Science (QuAADS) Laboratory Directed Research and Development (LDRD) Initiative. The Pacific Northwest National Laboratory is operated by Battelle for the U.S. Department of Energy under Contract DE-AC05-76RL01830. NW also acknowledges support from
Google Inc.  AM acknowledges support from the Quantum Science Consortium funded by NSERC.  We thank Sophia Simon, Dominic Berry, Rolando Somma, Robin Kothari, Robbie King and Ryan Babbush for useful discussions and feedback.
\printbibliography

@article{babbush2023exponential,
  title={Exponential quantum speedup in simulating coupled classical oscillators},
  author={Babbush, Ryan and Berry, Dominic W and Kothari, Robin and Somma, Rolando D and Wiebe, Nathan},
  journal={arXiv preprint arXiv:2303.13012},
  year={2023}
}

@article{liu2021efficient,
  title={Efficient quantum algorithm for dissipative nonlinear differential equations},
  author={Liu, Jin-Peng and Kolden, Herman {\O}ie and Krovi, Hari K and Loureiro, Nuno F and Trivisa, Konstantina and Childs, Andrew M},
  journal={Proceedings of the National Academy of Sciences},
  volume={118},
  number={35},
  pages={e2026805118},
  year={2021},
  publisher={National Academy of Sciences}
}

@article{berry2007efficient,
  author = {Berry, Dominic W. and Ahokas, Graeme and Cleve, Richard and Sanders, Barry C.},
  title = {Efficient quantum algorithms for simulating sparse Hamiltonians},
  journal = {Communications in Mathematical Physics},
  volume = {270},
  number = {2},
  pages = {359--371},
  year = {2007},
  doi = {10.1007/s00220-006-0150-x},
  eprint = {quant-ph/0508139},
  bibcode = {2007CMaPh.270..359B},
  S2CID = {37923044}
}

@article{childs2012hamiltonian,
  author = {Childs, Andrew M. and Wiebe, Nathan},
  title = {Hamiltonian simulation using linear combinations of unitary operations},
  journal = {Quantum Information \& Computation},
  volume = {12},
  number = {11-12},
  pages = {901--924},
  year = {2012}
}

@article{low2019hamiltonian,
  author = {Low, Guang Hao and Chuang, Isaac L.},
  title = {Hamiltonian Simulation by Qubitization},
  journal = {Quantum},
  volume = {3},
  pages = {163},
  year = {2019},
  doi = {10.22331/q-2019-07-12-163},
  eprint = {1610.06546},
  archivePrefix = {arXiv},
  primaryClass = {quant-ph}
}

@article{krovi2023improved,
  title={Improved quantum algorithms for linear and nonlinear differential equations},
  author={Krovi, Hari},
  journal={Quantum},
  volume={7},
  pages={913},
  year={2023},
  publisher={Verein zur F{\"o}rderung des Open Access Publizierens in den Quantenwissenschaften}
}

@article{Xue2023Further,
  author    = {Jin Xue and Andrew M. Childs and others},
  title     = {Further improving quantum algorithms for nonlinear differential equations},
  journal   = {arXiv preprint arXiv:2312.09518},
  year      = {2023},
}

@article{Berry2017Quantum,
  author    = {Dominic W. Berry and Andrew M. Childs and Robin Kothari},
  title     = {Quantum algorithm for linear differential equations with exponentially improved dependence on precision},
  journal   = {Commun.\ Math.\ Phys.},
  volume    = {356},
  pages     = {1--17},
  year      = {2017},
  doi       = {10.1007/s00220-017-3002-y},
}

@article{jin2024quantum,
  title={Quantum Simulation of Partial Differential Equations via Schr{\"o}dingerization},
  author={Jin, Shi and Liu, Nana and Yu, Yue},
  journal={Physical Review Letters},
  volume={133},
  number={23},
  pages={230602},
  year={2024},
  publisher={APS}
}

@article{berry2014high,
  title={High-order quantum algorithm for solving linear differential equations},
  author={Berry, Dominic W},
  journal={Journal of Physics A: Mathematical and Theoretical},
  volume={47},
  number={10},
  pages={105301},
  year={2014},
  publisher={IOP Publishing}
}

@article{clader2013preconditioned,
  title={Preconditioned quantum linear system algorithm},
  author={Clader, B David and Jacobs, Bryan C and Sprouse, Chad R},
  journal={Physical review letters},
  volume={110},
  number={25},
  pages={250504},
  year={2013},
  publisher={APS}
}

@article{an2023linear,
  title={Linear combination of hamiltonian simulation for nonunitary dynamics with optimal state preparation cost},
  author={An, Dong and Liu, Jin-Peng and Lin, Lin},
  journal={Physical Review Letters},
  volume={131},
  number={15},
  pages={150603},
  year={2023},
  publisher={APS}
}

@article{costa2019quantum,
  title={Quantum algorithm for simulating the wave equation},
  author={Costa, Pedro CS and Jordan, Stephen and Ostrander, Aaron},
  journal={Physical Review A},
  volume={99},
  number={1},
  pages={012323},
  year={2019},
  publisher={APS}
}

@article{costa2023further,
  title={Further improving quantum algorithms for nonlinear differential equations via higher-order methods and rescaling},
  author={Costa, Pedro and Schleich, Philipp and Morales, Mauro ES and Berry, Dominic W},
  journal={arXiv preprint arXiv:2312.09518},
  year={2023}
}

@article{harrow2009quantum,
  title={Quantum algorithm for linear systems of equations},
  author={Harrow, Aram W and Hassidim, Avinatan and Lloyd, Seth},
  journal={Physical review letters},
  volume={103},
  number={15},
  pages={150502},
  year={2009},
  publisher={APS}
}

@article{childs2017quantum,
  title={Quantum algorithm for systems of linear equations with exponentially improved dependence on precision},
  author={Childs, Andrew M and Kothari, Robin and Somma, Rolando D},
  journal={SIAM Journal on Computing},
  volume={46},
  number={6},
  pages={1920--1950},
  year={2017},
  publisher={SIAM}
}

@article{subacsi2019quantum,
  title={Quantum algorithms for systems of linear equations inspired by adiabatic quantum computing},
  author={Suba{\c{s}}{\i}, Yi{\u{g}}it and Somma, Rolando D and Orsucci, Davide},
  journal={Physical review letters},
  volume={122},
  number={6},
  pages={060504},
  year={2019},
  publisher={APS}
}

@article{somma2025shadow,
  title={Shadow hamiltonian simulation},
  author={Somma, Rolando D and King, Robbie and Kothari, Robin and O’Brien, Thomas E and Babbush, Ryan},
  journal={Nature Communications},
  volume={16},
  number={1},
  pages={2690},
  year={2025},
  publisher={Nature Publishing Group UK London}
}

@article{takahira2021quantum,
  title={Quantum algorithms based on the block-encoding framework for matrix functions by contour integrals},
  author={Takahira, Souichi and Ohashi, Asuka and Sogabe, Tomohiro and Usuda, Tsuyoshi Sasaki},
  journal={arXiv preprint arXiv:2106.08076},
  year={2021}
}

@article{stroeks2024solving,
  title={Solving Free Fermion Problems on a Quantum Computer},
  author={Stroeks, Maarten and Lenterman, Daan and Terhal, Barbara and Herasymenko, Yaroslav},
  journal={arXiv preprint arXiv:2409.04550},
  year={2024}
}

@article{wu2024quantum,
  title={Quantum Algorithms for Nonlinear Dynamics: Revisiting Carleman Linearization with No Dissipative Conditions},
  author={Wu, Hsuan-Cheng and Wang, Jingyao and Li, Xiantao},
  journal={arXiv preprint arXiv:2405.12714},
  year={2024}
}

@inproceedings{gilyen2019quantum,
  title={Quantum singular value transformation and beyond: exponential improvements for quantum matrix arithmetics},
  author={Gily{\'e}n, Andr{\'a}s and Su, Yuan and Low, Guang Hao and Wiebe, Nathan},
  booktitle={Proceedings of the 51st Annual ACM SIGACT Symposium on Theory of Computing},
  pages={193--204},
  year={2019}
}

@article{meyer2013nonlinear,
  title={Nonlinear quantum search using the Gross--Pitaevskii equation},
  author={Meyer, David A and Wong, Thomas G},
  journal={New Journal of Physics},
  volume={15},
  number={6},
  pages={063014},
  year={2013},
  publisher={IOP Publishing}
}

@article{an2023quantum,
  title={Quantum algorithm for linear non-unitary dynamics with near-optimal dependence on all parameters},
  author={An, Dong and Childs, Andrew M and Lin, Lin},
  journal={arXiv preprint arXiv:2312.03916},
  year={2023}
}

@article{jin2024schr,
  title={On Schr$\backslash$" odingerization based quantum algorithms for linear dynamical systems with inhomogeneous terms},
  author={Jin, Shi and Liu, Nana and Ma, Chuwen},
  journal={arXiv preprint arXiv:2402.14696},
  year={2024}
}

@article{berry2024quantum,
  title={Quantum algorithm for time-dependent differential equations using Dyson series},
  author={Berry, Dominic W and Costa, Pedro CS},
  journal={Quantum},
  volume={8},
  pages={1369},
  year={2024},
  publisher={Verein zur F{\"o}rderung des Open Access Publizierens in den Quantenwissenschaften}
}

@article{danz2024calculating,
  title={Calculating response functions of coupled oscillators using quantum phase estimation},
  author={Danz, Sven and Berta, Mario and Schr{\"o}der, Stefan and Kienast, Pascal and Wilhelm, Frank K and Ciani, Alessandro},
  journal={arXiv preprint arXiv:2405.08694},
  year={2024}
}

@article{sanavio2025explicit,
  title={Explicit quantum circuit for simulating the advection-diffusion-reaction dynamics},
  author={Sanavio, Claudio and Mauri, Enea and Succi, Sauro},
  journal={IEEE Transactions on Quantum Engineering},
  year={2025},
  publisher={IEEE}
}

@article{jennings2024cost,
  title={The cost of solving linear differential equations on a quantum computer: fast-forwarding to explicit resource counts},
  author={Jennings, David and Lostaglio, Matteo and Lowrie, Robert B and Pallister, Sam and Sornborger, Andrew T},
  journal={Quantum},
  volume={8},
  pages={1553},
  year={2024},
  publisher={Verein zur F{\"o}rderung des Open Access Publizierens in den Quantenwissenschaften}
}

@article{li2025potential,
  title={Potential quantum advantage for simulation of fluid dynamics},
  author={Li, Xiangyu and Yin, Xiaolong and Wiebe, Nathan and Chun, Jaehun and Schenter, Gregory K and Cheung, Margaret S and M{\"u}lmenst{\"a}dt, Johannes},
  journal={Physical Review Research},
  volume={7},
  number={1},
  pages={013036},
  year={2025},
  publisher={APS}
}

@article{forets2017explicit,
  title={Explicit error bounds for Carleman linearization},
  author={Forets, Marcelo and Pouly, Amaury},
  journal={arXiv preprint arXiv:1711.02552},
  year={2017}
}

@article{lloyd2020quantum,
  title={Quantum algorithm for nonlinear differential equations},
  author={Lloyd, Seth and De Palma, Giacomo and Gokler, Can and Kiani, Bobak and Liu, Zi-Wen and Marvian, Milad and Tennie, Felix and Palmer, Tim},
  journal={arXiv preprint arXiv:2011.06571},
  year={2020}
}

@article{brustle2024quantum,
  title={Quantum and classical algorithms for nonlinear unitary dynamics},
  author={Br{\"u}stle, Noah and Wiebe, Nathan},
  journal={arXiv preprint arXiv:2407.07685},
  year={2024}
}

@article{carleman1932application,
  title={Application of the theory of linear integral equations to systems of nonlinear differential equations},
  author={Carleman, Torsten},
  year={1932}
}

@article{lloyd1996universal,
  title={Universal quantum simulators},
  author={Lloyd, Seth},
  journal={Science},
  volume={273},
  number={5278},
  pages={1073--1078},
  year={1996},
  publisher={American Association for the Advancement of Science}
}

@article{feynman1985quantum,
  title={Quantum mechanical computers},
  author={Feynman, Richard},
  journal={Optics news},
  volume={11},
  number={2},
  pages={11--20},
  year={1985}
}

@book{landau1976mechanics,
  author = {Landau, L. D. and Lifshitz, E. M.},
  title = {Mechanics},
  volume = {1},
  publisher = {Butterworth-Heinemann},
  edition = {3rd},
  year = {1976},
  isbn = {0750628960}
}

@book{jackson1998classical,
  author = {Jackson, John David},
  title = {Classical Electrodynamics},
  publisher = {Wiley},
  address = {New York},
  edition = {3rd},
  year = {1998},
  isbn = {047130932X}
}

@article{neven1992rate,
  author = {Neven, Hartmut and Aertsen, Ad},
  title = {Rate coherence and event coherence in the visual cortex: a neuronal model of object recognition},
  journal = {Biological Cybernetics},
  volume = {67},
  pages = {309--322},
  year = {1992},
  doi = {10.1007/BF02414887}
}

@book{wilson1980molecular,
  author = {Wilson, E. Bright and Decius, J. C. and Cross, Paul C.},
  title = {Molecular Vibrations: The Theory of Infrared and Raman Vibrational Spectra},
  publisher = {Dover Publications},
  address = {New York},
  year = {1980},
  isbn = {048663941X}
}
\end{document}